\documentclass[11pt]{article}
\usepackage[left=1in,top=1in,right=1in,bottom=1in,head=0in,nofoot]{geometry}

\usepackage{setspace,url,bm,amsmath} 

\usepackage{titlesec} 
\titlelabel{\thetitle.\quad} 
\titleformat*{\section}{\bf\Large\center}

\usepackage{graphicx} 
\usepackage{bbm}
\usepackage{latexsym}
\usepackage{float, caption}
\usepackage{subcaption}
\usepackage{enumerate}

\usepackage{arydshln}

\usepackage{tikz}

\usepackage{tikzsymbols}

\usepackage{wasysym}

\newcommand{\GG}[1]{}

\usepackage{amsthm}
\usepackage{amssymb}
\usepackage{amsmath}

\usepackage{comment}
\theoremstyle{definition}
\newtheorem{assumption}{Assumption}
\newtheorem*{theorem*}{Theorem}
\newtheorem{theorem}{Theorem}
\newtheorem*{rmk*}{Remark}
\newtheorem{proposition}{Proposition}
\newtheorem{lemma}{Lemma}

\newtheorem{corollary}{Corollary}
\newtheorem*{corollary*}{Corollary}

\usepackage{natbib} 
\bibpunct{(}{)}{;}{a}{}{,} 

\usepackage{etoolbox} 
\apptocmd{\sloppy}{\hbadness 10000\relax}{}{} 

\usepackage{multibib}
\newcites{sec}{References}

\usepackage{color}
\usepackage{listings}
\usepackage{hyperref}
\usepackage{booktabs}

\usepackage{lscape}

\def\T{{ \mathrm{\scriptscriptstyle T} }}
\def\v{{\varepsilon}}
\def\var{\text{var}}
\def\sumM{\sum_{i=1}^M}

\def\TG{{\mathcal{\scriptscriptstyle T}}}
\def\CG{{\mathcal{\scriptscriptstyle C}}}
\def\se{\textup{se}}
\def\i{{\textsc{i}}}
\def\a{{\textup{adj}}}
\def\t{{\textsc{t}}}
\def\leqholder{\leq_\textsc{hi}}

\RequirePackage[normalem]{ulem}

\allowdisplaybreaks
\begin{document}
\onehalfspacing
\title{ \Large   \bf 
Model-assisted analyses of cluster-randomized experiments
}
\author{Fangzhou Su and Peng Ding
\footnote{ Department of Statistics, University of California, Berkeley, CA 94720, USA (E-mails:
\url{fangzhou_su@berkeley.edu} and \url{pengdingpku@berkeley.edu}). 
}
}
\date{}
\maketitle
\begin{abstract}
Cluster-randomized experiments are widely used due to their logistical convenience and policy relevance. To analyze them properly, we must address the fact that the treatment is assigned at the cluster level instead of the individual level. Standard analytic strategies are regressions based on individual data, cluster averages, and cluster totals, which differ when the cluster sizes vary. These methods are often motivated by models with  strong and unverifiable assumptions, and the choice among them can be subjective. Without any outcome modeling assumption, we evaluate these regression estimators and the associated robust standard errors from the design-based perspective where only the treatment assignment itself is random and controlled by the experimenter. We demonstrate that regression based on cluster averages  targets a weighted average treatment effect, regression based on individual data is suboptimal in terms of efficiency, and regression based on cluster totals is consistent and more efficient with a large number of clusters. We highlight the critical role of covariates in improving estimation efficiency, and illustrate the efficiency gain via both simulation studies and data analysis. The asymptotic analysis also reveals the efficiency-robustness trade-off  by comparing the properties of various estimators using data at different levels with and without covariate adjustment.  Moreover, we show that the robust standard errors are convenient approximations to the true asymptotic standard errors under the design-based perspective. Our theory holds even when the outcome models are misspecified, so it is model-assisted rather than model-based. We also extend the theory to a wider class of weighted average treatment effects.

\end{abstract}

\newpage
\section{Introduction}

Cluster-randomized experiments, also known as group-randomized trials, are widely used in empirical research, where randomization over units within clusters is either unethical or logistically infeasible. For example, in many public health interventions, clusters are villages, units are households, and randomization is implemented at the village level \citep{donner2000design, turner2017reviewdesign, turner2017reviewanalysis}; in many educational interventions, clusters are classrooms, units are students, and randomization is implemented at the classroom level \citep{raudenbush1997statistical, schochet2013estimators, raudenbush2020randomized, schochet2021design}. In addition to their convenience, cluster-randomized experiments can circumvent the problem of interference among units and are policy-relevant if the intervention of a policy is at the cluster level.

A proper analysis of a cluster-randomized experiment must first clearly specify the population and parameter of interest and then address the fact that randomization is at the cluster level. Model-based analyses address these issues simultaneously by imposing certain parametric assumptions and correlation structure on the error terms \citep{graubard1994regression, donner2000design, green2008analysis}. However, these modeling assumptions are often too strong and can lead to bias under model misspecifications. We take an alternative perspective by first defining the parameter of interest based on potential outcomes for treatment effects and then deriving the properties of the regression estimators under the randomization of the treatment assignment. This design-based or randomization-based perspective has been exploited by \citet{miratrix:2013}, \citet{dasgupta2014causal}, \citet{imbens2015causal}, \citet{fogarty2018mitigating}, \citet{mukerjee2018using} and \citet{liu2019regression} for analyzing various randomized experiments, and by \citet{imai2009essential}, \citet{schochet2013estimators}, \citet{Middleton2015cluster}, \citet{ATHEY201773}, \citet{abadie2017should} and \citet{schochet2021design} for analyzing cluster-randomized experiments in particular.

We will discuss regression estimators based on individual data, cluster totals, and cluster averages, with and without adjusting for covariates. We demonstrate that regression estimators based on individual data and cluster totals are consistent for the average treatment effect, with the latter giving more efficient estimators when adjusted for cluster sizes and cluster-level covariates. Our results supplement \citet{schochet2021design} by providing a unified theory for regressions using data at different levels and regressions with treatment-covariates interaction. 
We quantify the efficiency-robustness trade-off by comparing the design-based properties of various estimators. In particular, our theory demonstrates that the covariate-adjusted estimator based on cluster totals has higher asymptotic efficiency with predictive covariates at the cost of robustness when the cluster sizes have large variability. In contrast, the covariate-adjusted estimator based on individual data and the unadjusted estimator based on cluster totals are suboptimal in asymptotic efficiency but  generally more robust. 
We also show that the cluster-robust standard errors \citep{liang1986longitudinal} are conservative estimators for the true asymptotic standard errors in individual-level regressions, and heteroskedasticity-robust standard errors \citep{huber::1967, white::1980} are conservative estimators for the true asymptotic standard errors in cluster-total regressions. The former has been implemented empirically \citep{green2008analysis, schochet2013estimators, ATHEY201773}. \citet{schochet2021design} provide some results for the case without covariates, and we extend their results  by establishing a more unified theory for the design-based properties of the cluster-robust standard errors especially for the individual-level regression estimators adjusted for covariates. The latter extends \citet{lin2013}'s theory for completely randomized experiments to accommodate the cluster structure.

 \citet{middleton2008bias} criticized the cluster-average regression because of its inconsistency for the average treatment effect. In contrast, we argue that cluster-average regressions view clusters themselves as the population and target a different estimand that is informative for many applications \citep{donner2000design, ATHEY201773}. We define a more general estimand as the weighted average of cluster-level effects and show that weighted least squares (WLS) based on cluster averages yields consistent estimators for this estimand. We further propose a new regression estimator based on weighted cluster averages, which dominates cluster-average regressions and can leverage covariate information to improve efficiency.

Although the estimators are motivated by parametric or semi-parametric models, we extend the current literature on agnostic regressions \citep{lin2013, fogarty2018regression, liu2019regression, schochet2021design}, and derive their design-based properties allowing for misspecifications of these models. Our theory has clear practical implications. For the average treatment effect, we recommend using the cluster-total regression adjusted for the cluster size and cluster-level of covariates. For the weighted average effect, we recommend using the new regression estimator based on weighted cluster averages adjusted for the weight and cluster-level of covariates. The heteroskedasticity-robust standard errors for both estimators provide convenient approximations for the true asymptotic standard errors in constructing confidence intervals. The recommended estimators can be easily implemented by standard software packages with appropriately specified covariates and weights. Importantly, these simple estimators have strong theoretical guarantees.

The paper proceeds as follows. Section \ref{sec::notation-framework} sets up the basic notation and framework for the design-based inference in cluster-randomized experiments, and introduces four regression estimators based on individual data or cluster totals, with or without covariates adjustment. Section \ref{sec::fourestimatorsfortau} discusses the asymptotic properties of the four estimators introduced in Section \ref{sec::notation-framework}. Section \ref{sec::general-estimand} proposes a general class of weighted average treatment effects over clusters, and analyzes four estimators that are analogous to those in Section \ref{sec::notation-framework}. Section \ref{sec::unification-recommendation} unifies the theory and gives practical recommendations. Section \ref{sec:simulation-section} uses simulation studies to evaluate the finite-sample properties of various estimators and robust standard errors. Section \ref{sec::application} illustrates the theory with a recent cluster-randomized experiment. Section \ref{sec::discussion} concludes and the supplementary material contains the proofs.  Replication files for this paper are publicly available at Harvard Dataverse: https://doi.org/10.7910/DVN/7FBJPU

We use the following notation throughout the paper. Let $o_\mathbb{P}(1)$ denote a random variable converging to zero in probability, with the probability measure $\mathbb{P}$ induced by the randomization of the treatment assignment. Let $R_M \rightsquigarrow \mathcal{N} (0,1) $ represent that $R_M$ converges in distribution to a standard Normal random variable as the number of clusters $M$ goes to infinity. Let ``$\se$'' and ``$\hat{\se}$'' denote the asymptotic standard error and the estimated asymptotic standard error of an estimator.

\section{Notation and framework}
\label{sec::notation-framework}

\subsection{Potential outcomes, treatment assignment, and observed data}
Consider a study with $N$ units, clustered, for example, by classrooms or villages. Cluster $i$ has $n_i$ units $(i=1,\ldots, M)$, and the total number of units is $N = \sumM n_i$. Let $(i,j)$ index the $j$th unit within cluster $i$ $(i=1,\ldots,M;\ j=1,\ldots, n_i)$. Unit $(i,j)$ has covariates $x_{ij}$, and cluster $i$ has covariates $c_i$. The experimenter randomly assigns $eM$ clusters to receive the treatment and $(1-e)M$ clusters to receive the control, where $0<e<1$ is a fixed number denoting the proportion of treated clusters. Let $Z_i$ be the treatment indicator for cluster $i$ and $Z_{ij} $ be the treatment indicator for unit $(i,j)$. In a cluster-randomized experiment, units within a cluster receive identical treatment levels. So if cluster $i$ receives treatment, then $Z_{ij} = Z_i = 1$; if cluster $i$ receives control, then $Z_{ij} = Z_i=  0$. Let $\sum_{ij}=\sum_{i=1}^M\sum_{j=1}^{n_{i}}$ denote the summation over all units. Let $\mathcal T=\{ (i,j):Z_{ij}=1 \}$ be the indices of units under treatment and $\mathcal C=\{ (i,j):Z_{ij}=0\}$ be the indices of units under control. Their cardinalities $n_\TG=\sum_{ij} Z_{ij}$ and $n_\CG = \sum_{ij} (1-Z_{ij})$ represent the total numbers of units under treatment and control, respectively, which are random if the cluster sizes vary.

For unit $(i,j)$, let $Y_{ij}(1)$ and $Y_{ij}(0)$ be the potential outcomes \citep{Neyman:1923} under treatment and control, respectively. The average potential outcome in treatment arm $z\ (z=0,1)$ is $\bar Y(z)= N^{-1}  \sum_{ij}Y_{ij}(z) $, and the average treatment effect is
$$
\tau =\bar{Y}(1)-\bar{Y}(0) = N^{-1} \sum_{ij} \{ Y_{ij}(1) - Y_{ij}(0) \}
$$
over all units.  The observed outcome is a function of the treatment indicator and potential outcomes: $Y_{ij}=Z_{ij} Y_{ij}(1)+(1- Z_{ij} )Y_{ij}(0)$. A central goal in analyzing a cluster-randomized experiment is to make inference about $\tau$ using the observed data $\{ (Z_{ij} ,Y_{ij}, x_{ij}, c_i): i=1,\ldots, M;j=1,\ldots, n_i\} $. Our analysis is design-based, that is, we condition on all covariates and potential outcomes, with randomness coming from $(Z_1,\ldots, Z_n)$,
a random permutation of $eM$ 1's and $(1-e)M$ 0's. Conditioning on the potential outcomes, we have a finite population of $N$ units and do not impose any stochastic outcome modeling assumption.

\subsection{A special case with $n_i=1$: completely randomized experiment}

To motivate model-assisted analyses of cluster-randomized experiments, we first review the existing design-based theory for completely randomized experiments. The special case has $n_i=1$, and simplifies the double index $(i,j)$ to the single index $i$. The average treatment effect reduces to $\tau = M^{-1} \sum_{i=1}^M  \{ Y_i(1) - Y_i(0) \} $. \citet{Neyman:1923} showed that the difference-in-means estimator $\hat{\tau}$ is unbiased for $\tau$. \citet{Fisher:1935} proposed to use the coefficient of $Z_i$ in the ordinary least squares (OLS) fit of $Y_{i}$ on $(1,Z_i,c_i)$ to estimate $\tau$, hoping to leverage the information in $c_i$ to improve estimation efficiency. However, \citet{freedman2008regression_a, freedman2008regression_b} criticized \citet{Fisher:1935}'s regression estimator because it can be even less efficient than $\hat{\tau}$. \citet{lin2013} proposed a slight modification to address \citet{freedman2008regression_a, freedman2008regression_b}'s critiques: with centered covariates such that $\bar{c} = M^{-1} \sum_{i=1}^M  c_i=0$, the coefficient of $Z_i$ in the OLS fit of $Y_{i}$ on $(1,Z_i,c_i, Z_ic_i)$ is consistent and asymptotically more efficient than $\hat{\tau}$, with the true asymptotic standard error conservatively estimated by the heteroskedasticity-robust standard error.

Importantly, \citet{freedman2008regression_a, freedman2008regression_b} and \citet{lin2013} did not assume that the linear models are correctly specified. They conditioned on the potential outcomes and evaluated the regression estimators from the design-based perspective. So their theories are model-assisted rather than model-based. An overarching goal of this paper is to develop the design-based theory for commonly-used regression estimators and robust standard errors for analyzing cluster-randomized experiments.

\subsection{Some common estimators and their regression formulations} \label{sec::four-estimators-for-tau}

Motivated by the definition $\tau =\bar{Y}(1)-\bar{Y}(0)$, it is intuitive to use the difference-in-means $\hat\tau_{\textsc{i}}=\bar Y_\TG-\bar Y_\CG$ as an estimator, where $\bar{Y}_\TG=n_\TG^{-1}\sum_{ij  \in \mathcal T}  Y_{ij}$ and $\bar Y_\CG = n_\CG^{-1} \sum_{ij \in \mathcal C} Y_{ij}$ are the average observed outcomes under treatment and control, respectively. Numerically, $\hat\tau_{\textsc{i}}$ is identical to the coefficient of $Z_{ij}$ in the OLS fit of $Y_{ij}$ on $(1,Z_{ij})$ with all $(i, j)$'s. The clustering structure of the data suggests using the cluster-robust standard error $\hat{\se}_\textsc{lz}(\hat\tau_{\textsc{i}})$ \citep{liang1986longitudinal, barrios2012clustering, schochet2013estimators, ATHEY201773} to construct a Wald-type confidence interval.

\citet{lin2013}'s analysis of completely randomized experiments motivates us to use covariates to further improve the estimation efficiency. We modify the above regression as follows: first center the  unit-level covariates such that $\bar x=N^{-1}\sum_{ij}x_{ij} =0$; then obtain $\hat\tau_\textsc{i}^\textup{adj}$, the coefficient of $Z_{ij}$ in the OLS fit of $Y_{ij}$ on $(1, Z_{ij}, x_{ij}, Z_{ij}x_{ij})$; finally report the cluster-robust standard error  $\hat{\se}_\textsc{lz}(\hat\tau_\textsc{i}^\textup{adj})$.

On the other hand, the parameter $\tau$ has an alternative form:
$$
\tau = M^{-1} \sum_{i=1}^M
\{\tilde{Y}_{i\cdot}(1)-\tilde{Y}_{i\cdot}(0)\},
$$
where $\tilde{Y}_{i\cdot}(z)=\sum_{j=1}^{n_i}Y_{ij}(z)M/N$ is the scaled cluster total potential outcome under treatment arm $z$ $(z=0,1)$. So $\tau$ is also the average treatment effect on the cluster level with outcomes defined as the scaled cluster totals. With the observed scaled cluster total $\tilde{Y}_{i\cdot}=Z_i\tilde{Y}_{i\cdot}(1)+(1-Z_i)\tilde{Y}_{i\cdot}(0)$, we can use
\begin{eqnarray*}
\hat\tau_{\textsc{t}}
&=& (eM)^{-1}\sum_{Z_i=1} \tilde{Y}_{i\cdot} - \{(1-e)M\}^{-1} \sum_{Z_i=0} \tilde{Y}_{i\cdot} \\
&=& (eN)^{-1}\sum_{ij\in \mathcal T}Y_{ij}-
\{(1-e)N\}^{-1}\sum_{ij\in \mathcal C}Y_{ij}
\end{eqnarray*}
to estimate $\tau$ \citep{Middleton2015cluster, li2017general, ATHEY201773}. Numerically, $\hat\tau_{\textsc{t}}$ is identical to the coefficient of $Z_i$ in the OLS fit of $\tilde{Y}_{i\cdot}$ on $(1,Z_{i})$ with the cluster total data. Since this regression uses aggregate data, we do not need to account for clustering and can simply use the heteroskedasticity-robust standard error $\hat{\se}_\textsc{hw}(\hat\tau_{\textsc{t}})$ \citep{huber::1967, white::1980}.

Again, we can modify the above regression to leverage the cluster-level covariates to improve efficiency. Because we have completely randomized experiments on clusters, \citet{li2017general} suggested to extend \citet{lin2013}'s estimator: first center the cluster-level covariates at $\bar c=M^{-1} \sumM c_i  = 0$; then obtain $\hat\tau_{\textsc{t}}^\textup{adj}$, the coefficient of $Z_i$ in the OLS fit of $\tilde{Y}_{i\cdot}$ on $(1, Z_{i}, c_i, Z_{i} c_i )$; finally report the heteroskedasticity-robust standard error $\hat{\se}_\textsc{hw}(\hat\tau_{\textsc{t}}^\textup{adj})$.

We have introduced four estimators with the subscript ``I'' or ``T'' denoting regression based on individual data or cluster-totals, and the superscript ``adj'' denoting covariate adjustment. For regressions based on individual data, we use the subscript ``LZ'' to denote cluster-robust standard error since the treatment assignment is at the cluster level. For regressions based on scaled cluster totals, we use the subscript ``HW'' to denote heteroskedasticity-robust standard error since the non-robust standard errors can be inconsistent for the true asymptotic standard errors even in the special case with $n_i=1$, as pointed out by \citet{freedman2008regression_a, freedman2008regression_b} and \citet{lin2013}.

The above four estimators all have reasoned justifications. However, their design-based statistical properties are not entirely clear. First, are these estimators consistent for $\tau$ when the linear models are misspecified? Analyzing the regression estimators based on individual data  is a novel contribution. Even though \citet{lin2013} and \citet{li2017general} suggested the consistency of the regression estimators based on cluster totals, the  previous theory does not directly cover the case with unequal cluster sizes, especially when some cluster sizes diverge. Second, cluster-robust standard errors are motivated by analyzing correlated data under certain modeling assumptions \citep{liang1986longitudinal, cameron2015practitioner}, so we must demonstrate that they are also applicable in the design-based inference without those modeling assumptions. The next section will answer these questions, and the following section will extend the theory to a general estimand.

\section{Asymptotics for the regression-based estimators}
 \label{sec::fourestimatorsfortau}

\subsection{Basic setup for the asymptotic analysis}

To derive the asymptotic properties of the four estimators under cluster-randomized experiments, we embed the $N$ units into a sequence of finite populations with clustering structure and an increasing number of clusters $M\rightarrow \infty$. All quantities depend on $M$, but we suppress the dependence in notation for simplicity. \citet{li2017general} gave a review on finite-population asymptotics, and \citet{Middleton2015cluster} gave more details on asymptotics in cluster-randomized experiments. We need to impose some regularity conditions to derive limiting theorems for cluster-randomized experiments.

\begin{assumption}\label{assume::4}
The proportion of treated clusters $e$ has a limit in $(0,1)$, and the dimensions $(p_x, p_c)$ of covariates $(x_{ij}, c_i)$ do not depend on $M$.
\end{assumption}

Assumption \ref{assume::4} requires that the numbers of treated and control clusters are not too imbalanced, and restricts the discussion to covariates with a fixed dimension. With some technical modifications, we can extend the results to the case with diverging $(p_x, p_c)$ using the techniques in \citet{bloniarz2015lasso} and  \citet{lei2018regression}.

\begin{assumption}\label{assume::1}
The potential outcomes satisfy
$N^{-1}\sum_{ij} Y_{ij}(z)^4 =O(1)$ for $z=0,1$. The covariates satisfy $N^{-1} \sum_{ij} ||x_{ij}||_{\infty}^4 =O(1),$ $M^{-1}\sum_{i=1}^M ||c_i||_{\infty}^2=O(1),$ and $M^{-1} \max_{1\le i \le M} ||c_i||_{\infty}^2 =o(1)$.
\end{assumption}

Assumption \ref{assume::1} restricts the moments of the potential outcomes and covariates. We impose stronger moment conditions on unit-level variables, for example, $Y_{ij}(z)$, than on the cluster-level variables to ensure that the scaled cluster totals $ \tilde{Y}_{i\cdot}(z) = \sum_{j=1}^{n_i}  Y_{ij}(z) M/N $ are well-behaved even when some cluster sizes are much larger than others.

The asymptotic properties depend crucially on
$$
\omega_i=n_i/N,\quad \Omega = \max_{1\leq i\leq M}  \omega_i,\quad \tilde{\omega}_i = \omega_i M =  n_i/(N/M),
$$
the proportion of units in cluster $i$, its maximum value, and its scaled value measuring the size of cluster $i$ divided by the average cluster size, respectively. When all clusters have equal sizes, $\omega_i = \Omega = 1/M$. When some clusters are much larger than others, $\Omega$ can grow faster than $1/M$. In some extreme cases, $\Omega$ can have order $O(1)$. The relative value $\tilde{\omega}_i $  has mean $M^{-1} \sum_{i=1}^M \tilde{\omega}_i = 1$ but can have diverging higher moments when some clusters have extreme sizes. With this notation, we can write $N^{-1}\sum_{ij} Y_{ij}(z)^4 = M^{-1} \sum_{i=1}^M \tilde{\omega}_i \mu_{4i}  $ where $\mu_{4i}=n_i^{-1} \sum_{j=1}^{n_i}Y_{ij}(z)^4$ is the mean of $Y_{ij}(z)^4$ within cluster $i$. So the condition $N^{-1}\sum_{ij} Y_{ij}(z)^4 =O(1)$ in Assumption \ref{assume::1} becomes more interpretable. In particular, if we further assume that the $(\tilde{\omega}_i,  \mu_{4i} )_{i=1}^M $ are independent and identically distributed draws from a population, then  Assumption \ref{assume::1} requires that   the relative cluster size and the cluster mean of $Y_{ij}(z)^4$ has a finite covariance. Other conditions in Assumption \ref{assume::1} have similar interpretations.

To express the asymptotic distributions, we introduce additional notation. Define the centered potential outcome as $\varepsilon_{ij}(z)=Y_{ij}(z)-\bar{Y}(z)$. Similar to $\tilde{Y}_{i\cdot}(z)$, we can define other scaled cluster totals, for example, $\tilde{\varepsilon}_{i\cdot}(z)=\sum_{j=1}^{n_i}\varepsilon_{ij}(z)M/N =\tilde{Y}_{i\cdot}(z) -  n_i  \bar{Y}(z) M/N$ for the centered potential outcome and $\tilde{x}_{i\cdot}=\sum_{j=1}^{n_i}x_{ij}M/N$ for covariates.
Let $\bar{Y}_{i\cdot}(z) = n_i^{-1} \sum_{j=1}^{n_i} Y_{ij}(z)$ and $\bar{x}_{i\cdot} = n_i^{-1} \sum_{j=1}^{n_i} x_{ij}$ be the cluster averages of the potential outcomes and covariates, respectively.
Let $\{ \mathcal E_i(1),  \mathcal E_i(0) \}_{i=1}^M$ be the cluster-level potential outcomes with zero averages over clusters.  The  $\mathcal E_i(z)$ can be $\tilde{\varepsilon}_{i\cdot}(z)$ or other transformed cluster-level potential outcomes.
Motivated by \citet{Neyman:1923}, we define
\begin{eqnarray*}
V_\textup{c}\{\mathcal E_i(z)\} &=&  \frac{1}{M}
\sum_{i=1}^M  \left\{\frac{\mathcal E_i(1)^2}{e}+
\frac{\mathcal E_i(0)^2}{1-e }\right\},\\
 V \{\mathcal E_i(z)\} &=& V_\textup{c}\{\mathcal E_i(z)\} - \frac{1}{M}
\sum_{i=1}^M \left\{ \mathcal E_i( 1 ) - \mathcal E_i( 0 ) \right\} ^ { 2 }
\end{eqnarray*}
as the  variances of the difference in means of $\{ \mathcal E_i(1),  \mathcal E_i(0) \}_{i=1}^M$, multiplied by $M$, in completely randomized experiments with and without the constant treatment effect assumption. 

\subsection{Regression estimators based on individual-level data}
\label{sec::asymptotics-individual-reg}

Define the sample versions of the centered potential outcomes as $\hat{\varepsilon}_{ij} = Y_{ij}-\bar{Y}_\TG$ for treated units and $\hat{\varepsilon}_{ij} = Y_{ij}-\bar{Y}_\CG$ for control units. We have the following theorem.

\begin{theorem}\label{thm::tau-i}
Let $\se^2(\hat\tau_{\textsc{i}}) =  V \{\tilde{\varepsilon}_{i\cdot}(z)\}/M$.
Under Assumptions \ref{assume::4} and \ref{assume::1}, if $
\Omega=o(1)
$, then $\hat\tau_{\textsc{i}} =\tau+o_\mathbb P(1)$; if $ \Omega=o( M^{-2/3} ) $ and $ V \{\tilde{\varepsilon}_{i\cdot}(z)\} \not \rightarrow 0$, then $( \hat\tau_{\textsc{i}} - \tau )/
\se(\hat\tau_{\textsc{i}}) \rightsquigarrow \mathcal{N} (0,1) $ and
$$
 M \times  \hat{\se}^2_\textsc{lz}(\hat\tau_{\textsc{i}})
=
\frac{M}{n_\TG^2}\sum_{ i=1 }^M Z_i
 \left(
 \sum_{j=1}^{n_{i}}
\hat{\varepsilon}_{ij} \right)^2 +
 \frac{M}{n_\CG^2}\sum_{ i = 1}^M (1-Z_i)
 \left(
 \sum_{j=1}^{n_{i}}
\hat{\varepsilon}_{ij}
 \right)^2
=
V_\textup{c}\left\{\tilde{\varepsilon}_{i\cdot}(z)\right\}+o_\mathbb P(1).
$$
\end{theorem}

The condition on $\Omega$ is stronger for asymptotic Normality than for consistency, and both hold if all cluster sizes have the same order $O(N/M)$. Theorem \ref{thm::tau-i} demonstrates that the cluster-robust standard error $\hat{\se}_\textsc{lz}(\hat\tau_{\textsc{i}})$ is a conservative estimator for the true asymptotic standard error due to the non-negative difference between $V_\textup{c}(\cdot)$ and $V(\cdot)$, which equals zero when $\tilde{\varepsilon}_{i\cdot}(1) - \tilde{\varepsilon}_{i\cdot}(0) = 0$ or $ \bar{Y}_{i\cdot}(1) - \bar{Y}_{i\cdot}(0) =\tau$ for all $i$. In contrast, without accounting for clustering, the heteroskedasticity-robust standard error $\hat{\se}_\textsc{hw}(\hat\tau_{\textsc{i}})$ can be anti-conservative because \citet[][chapter 8]{angrist2008mostly} implies
$$
N\times \hat{\se}^2_\textsc{hw}(\hat\tau_{\textsc{i}})
=
\frac{N}{n_\TG^2} \sum_{ij \in \mathcal T} \hat{\varepsilon}_{ij}^2
+ \frac{N}{n_\CG^2} \sum_{ij \in \mathcal C} \hat{\varepsilon}_{ij}^2 \\
=
 \frac{1}{N} \sum_{ij} \left\{ \frac{\varepsilon_{ij}(1)^2}{e}+
 \frac{\varepsilon_{ij}(0)^2}{1-e}   \right\}+o_\mathbb P(1).
$$
If $n_i=1$ for all clusters, then $\hat{\se}_\textsc{hw}^2(\hat\tau_{\textsc{i}})$ coincides with $\hat{\se}^2_\textsc{lz}(\hat\tau_{\textsc{i}})$. If $n_i > 1$, then $\hat{\se}_\textsc{hw}^2(\hat\tau_{\textsc{i}})$ is often much smaller than $\hat{\se}^2_\textsc{lz}(\hat\tau_{\textsc{i}})$. To fix ideas, we consider a simple case with constant treatment effect $Y_{ij}(1) - Y_{ij}(0) =\tau$. From Theorem \ref{thm::tau-i}, $\hat{\se}_\textsc{lz}(\hat\tau_{\textsc{i}})$ is consistent for the true asymptotic   standard error, but we can verify that $\hat{\se}^2_\textsc{lz}(\hat\tau_{\textsc{i}}) / \hat{\se}^2_\textsc{hw}(\hat\tau_{\textsc{i}})$ has probability limit close to
\begin{eqnarray*}
1 + \frac{  \sum_{i=1}^M \sum_{(j,k): j\neq k}  \varepsilon_{ij}(0)\varepsilon_{ik}(0) }{ \sum_{ij}  \varepsilon_{ij}(0)  ^2 },
\end{eqnarray*}
which is larger than $1$ if the centered control potential outcomes are positively correlated within clusters. The magnitude of the ratio $\hat{\se}^2_\textsc{lz}(\hat\tau_{\textsc{i}}) / \hat{\se}^2_\textsc{hw}(\hat\tau_{\textsc{i}})$ is closely related to the intraclass correlation coefficient in clustered sampling \citep{lohr2019sampling} and cluster-randomized experiments \citep{donner2000design}.
Therefore, without accounting for clustering, we may have invalid inference even with large samples, which is well-known since \citet{cornfield1978randomizationbygroup}.

\citet{schochet2021design} focus on a slightly more general design and discussed the design-based properties of the cluster-robust standard error from an individual-level regression without covariates. Their result hints at Theorem \ref{thm::tau-i}, and we present a slightly more formal result under weaker regularity conditions on the sizes of the clusters. We further extend the theory to the regressions with covariates. The asymptotic properties of the covariate-adjusted regression estimator based on individual data rely on more regularity conditions on the covariates and potential outcomes.

\begin{assumption}\label{a4}
For $z=0,1$, $N^{-1}\sum _ {ij} x_{ij} \varepsilon_{ij}(z)=M^{-1}\sum_{i=1}^M \tilde{\omega}_i  \{ n_i^{-1}  \sum_{j=1}^{n_i}x_{ij}\varepsilon_{ij}(z) \} $ converges to a finite vector, and $N^{-1}\sum _ {ij} x_{ij} x_{ij}^{\T}=M^{-1}\sum_{i=1}^M \tilde{\omega}_i  \{ n_i^{-1} \sum_{j=1}^{n_i}x_{ij}x_{ij}^\T\} $ converges to a finite and invertible matrix.
\end{assumption}

Let $ Q_\textsc{i}(z)$ be the coefficient of $x_{ij}$ in the OLS fit of $\varepsilon_{ij}(z)$ on $x_{ij}$ with individual data $(z=0,1)$. This theoretical quantity depends on all potential outcomes and is useful for expressing the asymptotic distributions of the estimators based on regressions with individual data.

\begin{theorem}\label{thm2}
Define $r_i(z)= \tilde{\varepsilon}_{i\cdot}(z)-\tilde{x}_{i\cdot}^{\T} Q_\textsc{i}(z) $ and $\se^2(\hat\tau_\textsc{i}^\textup{adj}) =  V \{ r_i(z)\} / M$.
Under Assumptions \ref{assume::4}--\ref{a4}, if
$\Omega=o(1)$, then $\hat\tau_\textsc{i}^\textup{adj}
=\tau+o_\mathbb P(1)$; if
$ \Omega=o(M^{-2/3}) $ and $ V \{r_i (z)\} \not \rightarrow 0$, then $( \hat\tau_\textsc{i}^\textup{adj} - \tau )/
\se(\hat\tau_\textsc{i}^\textup{adj}) \rightsquigarrow \mathcal{N} (0,1) $ and
$M \times \hat{\se}^2_\textsc{lz}(\hat\tau_\textsc{i}^\textup{adj})=V_\textup{c}\{r_i(z)\}+o_\mathbb P(1).$
 \end{theorem}

Theorem \ref{thm2} demonstrates that the regression estimator $\hat\tau_\textsc{i}^\textup{adj}$ adjusted for covariates $x_{ij}$ is consistent, and the associated cluster-robust standard error is a conservative estimator for the true asymptotic   standard error. However, the form of $r_i(z)$ in Theorem \ref{thm2} is $ \tilde{\varepsilon}_{i\cdot}(z)-\tilde{x}_{i\cdot}^{\T} Q$, which suggests that an optimal choice of $Q$ should be the coefficient of $\tilde{x}_{i\cdot}$ in the OLS fit of $\tilde{\varepsilon}_{i\cdot}(z)$ on $\tilde{x}_{i\cdot}$ based on cluster totals, rather than  $ Q_\textsc{i}(z)$ from the individual-level regression. Consequently, it is possible that $\hat\tau_\textsc{i}^\textup{adj}$ is even less efficient than $\hat\tau_{\textsc{i}}$. We give a numeric example in Section \ref{sec::6.2}. In the special case with $n_i = N/M$ and $x_{ij}$ not varying with $j$, $ Q_\textsc{i}(z)$ equals the coefficient of $\tilde{x}_{i\cdot}$ in the OLS fit of $\tilde{\varepsilon}_{i\cdot}(z)$ on $\tilde{x}_{i\cdot}$, so $\hat\tau_\textsc{i}^\textup{adj}$ has smaller asymptotic variance than $\hat\tau_{\textsc{i}}$. In particular, when $n_i=1$ for all clusters, Theorem \ref{thm2} recovers \citet{lin2013} and \citet{li2017general}.

\citet{lin2013} popularized the OLS estimator with treatment-covariates interactions, but it is common in social science to use $\hat\tau_{\textsc{i}}^{\textup{ancova}}$, the coefficient of $Z_{ij}$ in the OLS fit of $Y_{ij}$ on $(1, Z_{ij}, x_{ij})$ without interactions \citep[e.g.,][]{graubard1994regression, green2008analysis, schochet2013estimators, guiteras2015encouraging, schochet2021design}. Its properties are similar to $\hat\tau_\textsc{i}^\textup{adj}$ in Theorem \ref{thm2}.

\begin{corollary}\label{pro2}
Define $r_i(z)= \tilde{\varepsilon}_{i\cdot}(z)-\tilde{x}_{i\cdot}^{\T}\{e Q_\textsc{i}(1)+(1-e) Q_\textsc{i}(0)\}$ and $\se^2(\hat\tau_\textsc{i}^\textup{ancova}) =  V \{  r_i(z)  \} / M$. Under Assumptions \ref{assume::4}--\ref{a4}, if
$\Omega  =o(1)$, then $\hat\tau_{\textsc{i}}^\textup{ancova}
=\tau+o_\mathbb P(1)$; if
$ \Omega  =o(M^{-2/3} ) $ and $ V \{  r_i(z)  \}\not\rightarrow 0$, then $( \hat\tau_{\textsc{i}}^\textup{ancova} - \tau )/
\se (\hat\tau_\textsc{i}^\textup{ancova}) \rightsquigarrow \mathcal{N} (0,1) $ and
$M \times \hat{\se}_\textsc{lz}^2(\hat\tau_{\textsc{i}}^\textup{ancova})=V_\textup{c}\{  r_i(z)  \}+o_\mathbb P(1).$
 \end{corollary}

Corollary \ref{pro2} demonstrates that $\hat\tau_{\textsc{i}}^\textup{ancova}$ is consistent and the associated cluster-robust standard error is conservative for the true   asymptotic  standard error. Similar to \citet{freedman2008regression_b}, $\hat\tau_{\textsc{i}}^\textup{ancova}$ can be even less efficient than $\hat\tau_{\textsc{i}}$; different from \citet{lin2013}, $\hat\tau_\textsc{i}^\textup{adj}$ does not necessarily improve $\hat\tau_{\textsc{i}}^\textup{ancova}$ in terms of asymptotic efficiency. Because $\hat\tau_{\textsc{i}}^\textup{ancova}$ is suboptimal even in the case with $n_i=1$  \citep{freedman2008regression_b, lin2013}, we do not recommend it unless we believe that there is little treatment heterogeneity explained by $x_{ij}$. \citet{schochet2021design} discuss a general version of $\hat\tau_\textsc{i}^\textup{adj}$ in cluster-randomized experiments with a blocking factor. They establish the central limit theorem but do not give a proof for the validity of the cluster-robust standard error. We fill the gap by showing the asymptotic conservativeness of the cluster-robust standard error. Moreover, \citet{schochet2021design} do not discuss $\hat\tau_\textsc{i}^\textup{adj}$ that allows for treatment effect heterogeneity with respect to the covariates.

\subsection{Regression estimators based on scaled cluster totals}

Recall that  $\tau=M^{-1} \sum_{i=1}^M \{\tilde{Y}_{i\cdot}(1)-\tilde{Y}_{i\cdot}(0)\} $ is the average treatment effect on the scaled cluster totals $\tilde{Y}_{i\cdot}$, and a cluster-randomized experiment is essentially a completely randomized experiment on clusters. So we expect previous results on the regression analysis of completely randomized experiments to be applicable if we view the scaled cluster total as a new outcome. Although conceptually straightforward, we must address the technical difficulty that the cluster sizes can vary by extending the central limit theorems under completely randomized experiments. We show the asymptotic properties of $\hat\tau_{\textsc{t}}$ and $\hat\tau_{\textsc{t}}^\textup{adj}$, allowing some cluster sizes to be much larger than others. The following theorems extend \citet{lin2013} and \citet{li2017general}.

\begin{theorem}\label{thm::tau-T}
Let $r_{i}(z)= \tilde{Y}_{i\cdot}(z)-\bar{Y}(z)$ and $\se^2(\hat{\tau}_\textsc{t}) =  V \{r_{i}(z)\}/M$.
Under Assumptions \ref{assume::4} and \ref{assume::1}, if $\Omega=o(1)$, then $\hat\tau_{\textsc{t}}=\tau+o_\mathbb P(1)$; if $
\Omega=o(M^{-2/3})$ and $  V \{r_{i}(z)\} \not \rightarrow 0$, then $( \hat{\tau}_\textsc{t} - \tau )/
\se(\hat\tau_{\textsc{t}}) \rightsquigarrow \mathcal{N} (0,1) $ and
$M\times \hat{\se}^2_\textsc{hw}(\hat\tau_{\textsc{t}})=V_\textup{c}\{r_{i}(z)\}+o_\mathbb P(1).$
\end{theorem}

\begin{assumption}\label{a5}
For $z=0,1$, $M^{-1}\sum _ { i = 1 } ^ { M } c_i  \tilde{Y}_{i\cdot}(z)=M^{-1}\sum_{i=1}^M \tilde{\omega}_i  c_i\bar{Y}_{i\cdot}(z)$
converges to a finite vector, and
$M^{-1}  \sum _ { i = 1 } ^ { M } c_i c_i^{\T} $ converges to a finite and invertible matrix.
\end{assumption}

\begin{theorem}\label{thm::tau-T-adj}
Let $\ r_{i}(z)$ be the residual from the OLS fit of $  \tilde{Y}_{i\cdot}(z) $ on $ (1,c_i)$ and define $\se^2(\hat\tau_{\textsc{t}}^\textup{adj}) =  V \{r_{i}(z)\} / M$.
Under Assumptions \ref{assume::4}, \ref{assume::1} and \ref{a5}, if $\Omega=o(1)$, then $\hat\tau_{\textsc{t}}^\textup{adj}=\tau+o_\mathbb P(1)$; if $ \Omega =O(M^{-1})$ and $  V \{r_{i}(z)\} \not \rightarrow 0$, then $( \hat\tau_{\textsc{t}}^\textup{adj}- \tau )/
\se(\hat\tau_{\textsc{t}}^\textup{adj}) \rightsquigarrow \mathcal{N} (0,1) $ and
$M\times \hat{\se}^2_\textsc{hw}(\hat\tau_{\textsc{t}}^\textup{adj})=V_\textup{c}\{r_{i}(z)\}+o_\mathbb P(1)$.
\end{theorem}

Comparing Theorems \ref{thm::tau-T} and \ref{a5}, covariate adjustment always improves the asymptotic efficiency because $\se^2(\hat\tau_{\textsc{t}}^\textup{adj})  \leq \se^2(\hat{\tau}_\textsc{t})$ from the regressions based on cluster totals.
We omit the detailed discussion of $\hat{\tau}_\textsc{t}^\text{ancova}$, the coefficient of $Z_i$ in the OLS fit of $\tilde{Y}_{i\cdot}$ on $(1, Z_{i}, c_i)$. Similar to $\hat{\tau}_\textsc{t}$ and $\hat\tau_{\textsc{t}}^\textup{adj}$, it  is also consistent and asymptotically Normal for $\tau$, and the associated heteroskedasticity-robust standard error is conservative for the true asymptotic standard error. However, it may not even improve the efficiency of $\hat{\tau}_\textsc{t}$, and it is suboptimal compared to $\hat\tau_{\textsc{t}}^\textup{adj}$.

\subsection{Comparison}\label{sec::compare1}

We first compare  $\hat\tau_\textsc{i}$ and $\hat\tau_\textsc{t}$: $\hat\tau_\textsc{t}$ is the Horvitz--Thompson estimator and thus it is unbiased \citep{Middleton2015cluster}; $\hat\tau_\textsc{i}$ is the H\'ajek estimator and thus it has finite-sample bias unless the clusters have equal sizes. In some extreme cases, $\hat\tau_\textsc{t}$ has smaller mean squared error than $\hat\tau_\textsc{i}$. For example, we have the following finite-sample result. The conditions are very strong and unlikely to hold, but we keep it here for mathematical interest.

\begin{proposition}\label{proposition::1}
If $e=1/2$ and $\bar Y(z)=0$ for $z=0,1$, then $E(\hat{\tau}_{\textsc{t}}-\tau)^2\le E(\hat{\tau}_{\textsc{i}}-\tau)^2$.
\end{proposition}

Nevertheless, $\hat\tau_\textsc{i}$ is consistent and its bias diminishes with large $M$. In most cases, $\hat\tau_\textsc{t}$ has larger asymptotic variance than $\hat\tau_\textsc{i}$ because the residual $\tilde{Y}_{i\cdot}(z) -  n_i  \bar{Y}(z) M/N$ in Theorem \ref{thm::tau-i} adjusts for the cluster size while the residual $\tilde{Y}_{i\cdot}(z)-\bar{Y}(z)$ in Theorem \ref{thm::tau-T} does not. Moreover, it is well known \citep{fuller2009sampling, Middleton2015cluster} that $\hat\tau_\textsc{t}$ is not invariant to the location shift of the outcomes, that is, if we transform the potential outcomes to $ b+Y_{ij}(z)$ for $z=0,1$, then $\hat\tau_\textsc{t}$ will change but $\hat\tau_\textsc{i}$ will not. We do not recommend using $\hat\tau_\textsc{t}$.

We can fix the problems of  $\hat\tau_\textsc{t}$ by covariate adjustment with $c_i$ including $n_i$. Below we compare various regression-adjusted estimators. We write particular covariates in parentheses after an estimator to emphasize the change of the default covariates. For instance, $\hat\tau_{\textsc{t}}^\textup{adj}(n_i,\tilde{x}_{i\cdot})$ denotes the regression estimator based on cluster totals adjusted for centered $c_i$ with $c_i = (n_i, \tilde{x}_{i\cdot})$.

\begin{corollary}\label{corollary::improvingIbyT}
The asymptotic variances satisfy
$$
\se^2\{\hat\tau_{\textsc{t}}^\textup{adj}(n_i,\tilde{x}_{i\cdot})\}
\le \se^2\{\hat\tau_{\textsc{t}}^\textup{adj}(n_i)\}
\le \se^2(\hat\tau_{\textsc{i}}),\quad
\se^2\{\hat\tau_{\textsc{t}}^\textup{adj}(n_i,\tilde{x}_{i\cdot})\}
\le \se^2(\hat\tau_\textsc{i}^\textup{adj}).
$$
These inequalities hold asymptotically if we replace the standard errors by the estimated ones.
\end{corollary}

Inequality $\se^2\{\hat\tau_{\textsc{t}}^\textup{adj}(n_i,\tilde{x}_{i\cdot})\}
\le \se^2\{\hat\tau_{\textsc{t}}^\textup{adj}(n_i)\}$ follows because \cite{lin2013}'s regression adjustment will always improve the asymptotic variance. Inequality $\se^2\{\hat\tau_{\textsc{t}}^\textup{adj}(n_i)\}
\le \se^2(\hat\tau_{\textsc{i}})$ follows from the fact that the residual sum of squares from the OLS fit of $ {\sum_{j=1}^{n_{i}}Y_{ij}(z)}M/N-\bar Y (z)$ on $n_iM/N-1$ is smaller than $\sum_{i=1}^M\tilde{\varepsilon}_{i\cdot}(z)^2$.
So adjusting for $n_i$ in the regression based on cluster totals improves the simple difference-in-means, and adjusting for cluster totals of additional covariates leads to more improvement asymptotically.
 \citet{Middleton2015cluster} discussed the role of $n_i$ and $ \tilde{x}_{i\cdot}$ in improving efficiency, and our estimator has the smallest asymptotic variance within the class of linearly adjusted estimators which follows from the optimality result in completely randomized experiments \citep{li2017general}.

Inequality $\se^2\{\hat\tau_{\textsc{t}}^\textup{adj}(n_i,\tilde{x}_{i\cdot})\}
\le \se^2(\hat\tau_\textsc{i}^\textup{adj})$ follows from the fact that $\sum_{i=1}^M \{\tilde{r}_{i\cdot}(z)\}^2
=\sum_{i=1}^M \{ \tilde{\varepsilon}_{i\cdot}(z)-
\tilde{x}_{i\cdot}^\T Q_\i(z) \}^2
$ in Theorem \ref{thm2} is larger than the residual sum of squares from the OLS fit of $
\tilde{\varepsilon}_{i\cdot}(z)$ on $\tilde{x}_{i\cdot}$, which is larger than
that from the OLS fit of ${\sum_{j=1}^{n_{i}}Y_{ij}(z)}M/N-\bar Y (z)$ on $(n_iM/N-1,\tilde{x}_{i\cdot})$.
So adjusting for $(n_i, \tilde{x}_{i\cdot})$ in the regression based on cluster totals improves the regression adjustment of $x_{ij}$ based on the individual data. It demonstrates that regressions based on cluster totals dominate regressions based on individual data asymptotically. With a special choice of covariates $x_{ij} = \bar x_{i\cdot}$ in individual-level regression, Corollary \ref{corollary::improvingIbyT}  implies $\se^2\{\hat\tau_{\textsc{t}}^\textup{adj}(n_i,\tilde{x}_{i\cdot})\}
\le \se^2\{\hat\tau_\textsc{i}^\textup{adj}(\bar x_{i\cdot})\}$.
One may wonder whether regressions based on cluster total can still dominate regressions based on individual data
if the latter already includes $(n_i,x_{ij})$ as covariates.
The answer is yes because
$
\se^2\{\hat\tau_{\textsc{t}}^\textup{adj}(n_i, n_i^2, \tilde{x}_{i\cdot})\}
\le \se^2\{  \hat\tau_\textsc{i}^\textup{adj} ( n_i, x_{ij} ) \}
$
follows from the second inequality in Corollary \ref{corollary::improvingIbyT}.

However, we must address a technical issue related to the regression estimators adjusted for $n_i$. When $c_i$ includes $n_i$, Assumption \ref{a5} imposes moment conditions on $n_i$, but $n_i$ itself can diverge. Fortunately,
by the property of OLS, $\hat\tau_{\textsc{t}}^\textup{adj}(n_i) = \hat\tau_{\textsc{t}}^\textup{adj}(\tilde{\omega}_i)$ and $\hat\tau_{\textsc{t}}^\textup{adj}(n_i,\tilde{x}_{i\cdot}) = \hat\tau_{\textsc{t}}^\textup{adj}(\tilde{\omega}_i,\tilde{x}_{i\cdot})$, which remain the same up to a scale transformation of the covariates. Therefore, Assumption \ref{a5} imposes moment conditions on $\tilde{\omega}_i$, for example, $M^{-1} \sum_{i=1}^M \tilde{\omega}_i(\tilde{\omega}_i-1)\bar{Y}_{i\cdot}(z)$ converges to a finite vector. This technical issue has important implications in practice. It requires at least two moments of
the $\tilde{\omega}_i$'s, which is stronger than the conditions required for the estimators based on the individual-level regressions. Therefore, the asymptotic efficiency of the estimators based on cluster totals comes at a price of the robustness with respect to the heavy-tailedness of the cluster sizes. In some extreme cases, the cluster sizes vary with some outliers, which can violate the regularity conditions for the estimators based on cluster totals. We observe this phenomenon in simulation studies in Section \ref{sec::Extremeclusters1}.

\section{A general weighted estimand and its estimators}\label{sec::general-estimand}

\subsection{Choice of the estimand}

In addition to the regression estimators based on individual data and cluster totals, it is also common to use regression estimators based on cluster averages \citep[e.g.,][]{donner2000design, green2008analysis, imai2009essential}. However, \citet{middleton2008bias} argued against the use of them due to their biases for estimating $\tau$ even with large $M$. In contrast, our view differs from \citet{middleton2008bias}. First, we show that we can easily remove the biases by using WLS with weights proportional to  cluster sizes. Second, we argue that regression estimators based on cluster averages typically target different estimands beyond $\tau$, which are also of interest in practice. Third, we demonstrate that the heteroskedasticity-robust standard errors from these regressions are conservative estimators for the asymptotic standard errors.

To unify the discussion, we introduce a  weighted estimand. For cluster $i$, the average treatment effect is $\tau_i=\bar Y_{i\cdot}(1)-\bar Y_{i\cdot}(0)$, and the average observed outcome is $\bar Y_{i\cdot}=Z_{i}\bar Y_{i\cdot}(1)+(1-Z_{i})\bar Y_{i\cdot}(0)$. For weights $(\pi_i)_{i=1}^M$ such that $\pi_i >  0$ and $\sumM \pi_i = 1$, define
$$
\tau_\pi = \sum_{i=1}^M\pi_i \tau_i = \sum_{i=1}^M\pi_i \{ \bar Y_{i\cdot}(1)-\bar Y_{i\cdot}(0)\} = \bar{Y}_\pi(1) - \bar{Y}_\pi(0),
$$
as a weighted average treatment effect where $\bar Y_\pi(z)=\sum_{i=1}^M \pi_{i} \bar{Y}_{i\cdot}(z) $ for $z=0,1$. If $\pi_i = \omega_i$, then $\tau_\pi$ recovers $\tau$; if $\pi_i=1/M$, then $\tau_\pi = \bar{\tau} = M^{-1} \sum_{i=1}^M \tau_i $ imposes equal weights on the average effects within clusters regardless of their sizes. The estimand $\tau$ views $N$ units as the population, and the estimand $\bar{\tau}$ views $M$ clusters as the population. Both can be of policy relevance. For example, \citet[][chapter 1.5]{donner2000design} discussed a smoking cessation trial with randomization over communities, in which $\tau$ is of interest if the aim is to influence individuals to cease smoking, and $\bar{\tau}$ is of interest if the aim is to lower the rate of smoking in a community. In some applications, the units within clusters are sampled from larger clusters, and we can use $\pi_i$ to reflect the sampling weights to recover the effect on a larger population. 
In general, $\tau_\pi $ includes a class of estimands that are of interest in different scenarios.

\subsection{Four estimators}

Because $\tau_\pi = \sum_{i=1}^M\pi_i \tau_i$ is a weighted estimand, it is natural to use weighted regressions to estimate it. We can use $\hat\tau_{\textsc{a}\pi}$, the coefficient of $Z_i$ in the WLS fit of $\bar{Y}_{i\cdot}$ on $(1,Z_i)$ with weights $\pi_i$, and report $\hat{\se}_\textsc{hw}(\hat\tau_{\textsc{a}\pi})$. With cluster-level covariates $c_i$, we can first center them at $\bar c_{\pi}=\sum_{i=1}^M\pi_ic_i=0$, then obtain $\hat\tau_{\textsc{a}\pi}^\textup{adj}$, the coefficient of $Z_i$ in the WLS fit of $\bar{Y}_{i\cdot}$ on $\{ 1,Z_i,c_i-\bar c_\pi,Z_i(c_i-\bar c_\pi) \} $ with weights $\pi_i$, and finally report $\hat{\se}_\textsc{hw}(\hat\tau_{\textsc{a}\pi}^\textup{adj})$.

On the other hand,
$$
\tau_\pi
=  M^{-1}\sum_{i=1}^M\{   M\pi_i \bar{Y}_{i\cdot}(1) - M\pi_i \bar{Y}_{i\cdot}(0) \}
$$
equals the average treatment effect on the weighted outcome $M\pi_i \bar Y_{i\cdot}$ on the cluster level. This representation motivates to estimate $\tau_\pi$ using OLS fit based on the weighted outcome. \citet{basse2018analyzing} suggested this approach in a related setting. We can use $\hat\tau_{\pi\textsc{a}}$, the coefficient of $Z_i$ in the OLS fit of $M\pi_i\bar{Y}_{i\cdot}$ on $(1,Z_i)$, and report $\hat{\se}_\textsc{hw}(\hat\tau_{\pi\textsc{a}})$. With cluster-level covariates $c_i$, we can first center them at $\bar c = M^{-1} \sum_{i=1}^M\ c_i = 0$, then obtain $\hat\tau_{\pi\textsc{a}}^\textup{adj}$, the coefficient of $Z_i$ in the OLS fit of $M\pi_i\bar{Y}_{i\cdot}$ on $(1,Z_i,c_i,Z_ic_i)$, and finally report $\hat{\se}_\textsc{hw}(\hat\tau_{\pi\textsc{a}}^\textup{adj})$.

Analogous to Section  \ref{sec::four-estimators-for-tau}, we have introduced four estimators for $\tau_\pi$, all associated with heteroskedasticity-robust standard errors. The orders of ``A'' and ``$\pi$'' in the subscript reflect the computational processes, for example, to compute $\hat\tau_{\textsc{a}\pi}^\textup{adj}$, we first compute the cluster averages and then run WLS, while to compute $\hat\tau_{\pi\textsc{a}}^\textup{adj}$, we first weight the cluster averages and then run OLS. Moreover, we center covariates $c_i$ at different means in these two types of estimators. In the special case with $\pi_i=1/M$, they are identical with $\hat\tau_{\textsc{a}\pi}=\hat\tau_{\pi\textsc{a}}$ and $\hat\tau_{\textsc{a}\pi}^\textup{adj}=
\hat\tau_{\pi\textsc{a}}^\textup{adj}$, so we use $\hat\tau_{\textsc{a}}$ and $\hat\tau_{\textsc{a}}^\textup{adj}$ to denote both, omitting $\pi$.

The remaining parts of this section will derive the asymptotic properties of the above four estimators and compare their efficiencies. These asymptotic analyses will help to unify Section  \ref{sec::fourestimatorsfortau} and Section  \ref{sec::general-estimand}.

\subsection{Weighted regression estimators based on cluster averages}

\begin{assumption}
\label{assume::wls-cluster-average} The weights satisfy $\max_{1\le i \le M}\pi_{i}=o(1)$, and the potential outcomes satisfy
$\sum_{i=1}^M \pi_i\bar{Y}_{i\cdot}(z)^4=O(1)$ for $z=0,1$.
\end{assumption}

\begin{theorem}\label{thm::wls-without-covariates}
Let $r_i(z) =M \pi_i\{\bar{Y}_{i\cdot}(z)-\bar Y_\pi(z)\} $ and $\se^2(\hat\tau_{\textsc{a}\pi}) =  V \{ r_i(z)\} /M$.
Under Assumptions \ref{assume::4} and \ref{assume::wls-cluster-average}, $\hat\tau_{\textsc{a}\pi}=\tau_{\pi}+o_\mathbb P(1)$; if further $\max_{1\le i \le M}\pi_{i}=o(M^{-2/3})$ and $V \{r_i(z)\} \not \rightarrow 0$, then $(\hat\tau_{\textsc{a}\pi}-\tau_\pi)/ \se(\hat\tau_{\textsc{a}\pi}) \rightsquigarrow \mathcal{N} (0,1) $ and
 $M\times \hat{\se}^2_\textsc{hw}(\hat\tau_{\textsc{a}\pi})=V_\textup{c}\{r_i(z)\}+o_\mathbb P(1)$.
\end{theorem}

If $\pi_i=1/M$, then $\hat\tau_{\textsc{a}\pi} = \hat\tau_{\textsc{a}}$ reduces to the regression estimator based on cluster averages. \citet{middleton2008bias} showed that it is inconsistent for estimating $\tau$, but Theorem \ref{thm::wls-without-covariates} shows that it is consistent for estimating $\bar{\tau} = M^{-1} \sumM \tau_i $. We have $\bar{\tau}  = \tau+o(1)$ if and only if the finite-sample covariance between $\tilde{\omega}_i$ and $\tau_i$ has order $o(1)$. There are at least two special cases in which $\bar{\tau}  = \tau+o(1)$ holds.
First, if  $n_i=N/M$, then $\bar{\tau} = \tau$ and $\hat\tau_{\textsc{a}} = \hat\tau_{\textsc{t}}.$ Second, if $\max_{1\le i \le M} |\tau_i-\tau|=o(1)$, then $\hat\tau_{\textsc{a}} =\tau+o_\mathbb P(1)$, but this condition rules out treatment effect heterogeneity across clusters.

Another important special case is $\pi_i=\omega_i = n_i/N$, in which we can consistently estimate $\tau$ by $\hat\tau_{\textsc{a}\omega}$. In fact, in this case, $\hat\tau_{\textsc{a}\omega} = \hat\tau_{\textsc{i}}$. We will revisit this case in the next section.

\begin{assumption}\label{a66}
For $z=0,1$, $  \sum _ { i = 1 } ^ { M } \pi_{i}(c_i-\bar c_\pi) \bar{Y}_{i\cdot}(z)$ converges to a finite vector, $  \sum _ { i = 1 } ^ { M }\pi_{i} (c_i-\bar c_\pi) (c_i-\bar c_\pi)^{\T}$ converges to a finite and invertible matrix, and $\sum_{i=1}^M \pi_i ||c_i||_{\infty}^4=O(1)$.
\end{assumption}

Let $Q_{\textsc{a}\pi}(z)$ be the coefficient of $c_i$ in the WLS fit of $\bar{Y}_{i\cdot}(z)$ on $(1,c_i)$ with weights $\pi_i$ $(z=0,1)$. It depends on all potential outcomes and simplifies the discussion below.

\begin{theorem}\label{thm::wls-with-covariates}
Define $r_i(z)=M\pi_i \{\bar{Y}_{i\cdot}(z)-\bar Y_\pi(z)-(c_i-\bar c_\pi) ^\T Q_{\textsc{a}\pi}(z)\}$   and $\se^2(\hat\tau_{\textsc{a}\pi}^\textup{adj})= V  \{r_i(z)\}/M.$
Under Assumptions \ref{assume::4}, \ref{assume::wls-cluster-average} and \ref{a66}, $\hat\tau_{\textsc{a}\pi}^\textup{adj}=\tau_{\pi}+o_\mathbb P(1)$; if further $\max_{1\le i \le M}\pi_{i}=o(M^{-2/3})$ and $  V  \{r_i(z)\} \not \rightarrow 0$, then $(\hat\tau_{\textsc{a}\pi}^\textup{adj}-\tau_\pi)/\se(\hat\tau_{\textsc{a}\pi}^\textup{adj}) \rightsquigarrow \mathcal{N} (0,1) $
 and
 $M\times \hat{\se}^2_\textsc{hw}(\hat\tau_{\textsc{a}\pi}^\textup{adj})=V_\textup{c}\{r_i(z)\}+o_\mathbb P(1)$.
\end{theorem}

The form of $r_i(z)$ in Theorem \ref{thm::wls-with-covariates} is $M\pi_i \{\bar{Y}_{i\cdot}(z)-\bar Y_\pi(z)\} - M\pi_i   (c_i-\bar c_\pi) ^\T Q$, which suggests that an optimal choice of $Q$ is the coefficient of  $\pi_i(c_i-\bar c_\pi)$ in the OLS fit of $\pi_i\{\bar{Y}_{i\cdot}(z)-\bar Y_\pi(z)\}$ on $\pi_i(c_i-\bar c_\pi)$, rather than $Q_{\textsc{a}\pi}(z)$ from the WLS based on cluster averages.
Similar to the discussion of Theorem \ref{thm2}, $\hat\tau_{\textsc{a}\pi}^\textup{adj}$ may not have better efficiency compared to $\hat\tau_{\textsc{a}\pi}$. We will give an example in Section \ref{sec::6.2}.

\subsection{Regression estimators based on weighted cluster averages}

 In the estimators $\hat\tau_{\pi\textsc{a}}$ and $\hat\tau_{\pi\textsc{a}}^\textup{adj}$, we view $M\pi_i\bar{Y}_{i\cdot}$ as a new outcome and apply \cite{lin2013}'s results for completely randomized experiments. Again the following theorems extend \citet{lin2013} and \citet{li2017general}, which include Theorems \ref{thm::tau-T} and \ref{thm::tau-T-adj} as special cases with $\pi_i=\omega_i$.

\begin{theorem}\label{thm8}
Let $r_i(z) = M \pi_i\bar{Y}_{i\cdot}(z)-\bar Y_\pi(z)$ and $\se^2(\hat\tau_{\pi\textsc{a}})= V \{r_i(z) \}/M.$
Under Assumption \ref{assume::4}, if $\sum_{i=1}^M \pi_i^2\bar{Y}_{i\cdot}(z)^2=o(1)$, then $\hat\tau_{\pi\textsc{a}}=\tau_{\pi}+o_\mathbb P(1)$; if further $\sum_{i=1}^M M^2\pi_i^4\bar{Y}_{i\cdot}(z)^4=o(1)$ and $  V \{ r_i(z)  \} \not \rightarrow 0$, then $(\hat\tau_{\pi\textsc{a}}-\tau_\pi)/\se(\hat\tau_{\pi\textsc{a}})\rightsquigarrow \mathcal{N} (0,1) $ and
$M\times \hat{\se}^2_\textsc{hw}(\hat\tau_{\pi\textsc{a}})=V_\textup{c}\{ r_i(z)\}+o_\mathbb P(1)$.
\end{theorem}

\begin{assumption}\label{a666}
For $z=0,1$, $  \sum _ { i = 1 } ^ { M } \pi_{i}c_i\bar{Y}_{i\cdot}(z)$ converges to a finite vector, and $  M^{-1}\sum _ { i = 1 } ^ { M } c_i c_i^{\T}$ converges to a finite and invertible matrix.
\end{assumption}

\begin{theorem}\label{thm9}
Let $r_i(z)$ be the residual from the OLS fit of $ M\pi_i \bar{Y}_{i\cdot}(z) $ on $(1,c_i)$ and $\se^2(\hat\tau_{\pi\textsc{a}}^\textup{adj})= V  \{r_i(z)\}/M.$
Under Assumptions \ref{assume::4}, \ref{assume::1} and \ref{a666}, if $\sum_{i=1}^M \pi_i^2\bar{Y}_{i\cdot}(z)^2=o(1)$, then $\hat\tau_{\pi\textsc{a}}^\textup{adj}=\tau_{\pi}+o_\mathbb P(1)$; if further $\sum_{i=1}^M M^3\pi_i^4\bar{Y}_{i\cdot}(z)^4=O(1)$ and $  V  \{r_i(z)\} \not \rightarrow 0$, then $(\hat\tau_{\pi\textsc{a}}^\textup{adj}-\tau_\pi)/\se(\hat\tau_{\pi\textsc{a}}^\textup{adj}) \rightsquigarrow \mathcal{N} (0,1) $ and
$M\times \hat{\se}^2_\textsc{hw}(\hat\tau_{\pi\textsc{a}}^\textup{adj})=V_\textup{c}\{r_i(z)\}+o_\mathbb P(1)$.
\end{theorem}

Comparing Theorems \ref{thm8} and \ref{thm9}, $\se^2(\hat\tau_{\pi\textsc{a}}^\textup{adj})\le \se^2(\hat\tau_{\pi\textsc{a}})$, so covariate adjustment always improves the asymptotic efficiency for estimating $\tau_\pi$ in the regression estimators based on weighted cluster averages.

\subsection{Comparison}

Similar to Corollary \ref{corollary::improvingIbyT}, we can reduce the asymptotic variances of the WLS estimators by  the OLS estimators based on weighted cluster averages with appropriately chosen covariates. Again, we write covariates in parentheses to emphasize them if they are different from the default choices.

\begin{corollary}\label{corollary::improve-Api-by-piA}
The asymptotic variances satisfy
$$
\se^2\{ \hat\tau_{\pi\textsc{a}}^\textup{adj}(\pi_i,\pi_i c_i )\} \le
\se^2\{\hat\tau_{\pi\textsc{a}}^\textup{adj}(\pi_i)\}\le \se^2(\hat\tau_{\textsc{a}\pi}),\quad
\se^2\{ \hat\tau_{\pi\textsc{a}}^\textup{adj}( \pi_i,\pi_i c_i) \}  \le
\se^2(\hat\tau_{\textsc{a}\pi}^\textup{adj}).
$$
These inequalities hold asymptotically if we replace the standard errors by the estimated ones.
\end{corollary}

Inequality $\se^2\{ \hat\tau_{\pi\textsc{a}}^\textup{adj}(\pi_i,\pi_i c_i )\} \le
\se^2\{\hat\tau_{\pi\textsc{a}}^\textup{adj}(\pi_i)\}$ follows because \cite{lin2013}'s regression adjustment will always improve the asymptotic variance. Inequality $\se^2\{\hat\tau_{\pi\textsc{a}}^\textup{adj}(\pi_i)\}\le \se^2(\hat\tau_{\textsc{a}\pi})$ follows from the fact that the residual sum of squares from the OLS fit of $\pi_i\bar{Y}_{i\cdot}(z)-\bar Y_\pi(z)/M$ on $\pi_i-1/M$ is smaller than $\sum_{i=1}^M [
\pi_i \{ \bar{Y}_{i\cdot}(z)-\bar Y_\pi(z) \}
]^2$. Inequality $\se^2\{ \hat\tau_{\pi\textsc{a}}^\textup{adj}( \pi_i,\pi_i c_i) \}  \le
\se^2(\hat\tau_{\textsc{a}\pi}^\textup{adj})$ follows from the fact that
$\sum_{i=1}^M \{r_i(z)/M\}^2=
\sum_{i=1}^M [\pi_i \{ \bar{Y}_{i\cdot}(z)-\bar Y_\pi(z)\}-
\pi_i(c_i-\bar c_\pi)^\T Q_{\textsc{a}\pi}(z)  ]^2
$ in Theorem \ref{thm::wls-with-covariates} is larger than the residual sum of squares from the OLS fit of $\pi_i\{\bar{Y}_{i\cdot}(z)-\bar Y_\pi(z)\}$ on $\pi_i(c_i-\bar c_\pi)$, which is larger than that from the OLS fit of $\pi_i \bar{Y}_{i\cdot}(z)-\bar Y_\pi (z)/M$ on $(\pi_i-1/M,\pi_ic_i-\bar c_\pi/M)$.

Corollary \ref{corollary::improve-Api-by-piA} highlights the importance of including $\pi_i$ as a cluster-level covariate to improve the asymptotic efficiency. Overall, given estimators $\hat\tau_{\textsc{a}\pi}$ and $\hat\tau_{\textsc{a}\pi}^\textup{adj}$, we can always improve their efficiency by using $ \hat\tau_{\pi\textsc{a}}^\textup{adj}( \pi_i,\pi_i c_i)  $.

\citet{schochet2021design} propose an even more general weighted estimand, motivated by the need to adjust for nonresponses in cluster-randomized experiments. Their weights can vary with individuals but our weights only vary with clusters. By allowing for general weights, we can also discuss individual-level regression with weights. This is beyond the scope of this paper.

\section{Unification and recommendation}
\label{sec::unification-recommendation}

So far, we have introduced four estimators for each of $\tau$ and $\tau_\pi$. Table \ref{tb::comparisonofestimators} summarizes them and the corresponding standard errors. Table \ref{tb::comparisonofestimators} presents them in parallel, and the ones in the each row are equivalent when $\pi_i=\omega_i$, as demonstrated by the proposition below.

\begin{proposition}\label{t14}
If $ \pi_i = \omega_i $ for all $i=1,\ldots, M,$ then
$$
\begin{array}{lllllll}
\hat\tau_{\textsc{a}\omega} &=&\hat\tau_{\textsc{i}},
& \quad &
\hat{\se}^2_\textsc{hw}(\hat\tau_{\textsc{a}\omega}) &=& \hat{\se}^2_\textsc{lz}(\hat\tau_{\textsc{i}}),\\
\hat\tau_{\textsc{a}\omega}^{\textup{adj}} &=&\hat\tau_\textsc{i}^\textup{adj}(c_i),
& \quad &
\hat{\se}^2_\textsc{hw}(\hat\tau_{\textsc{a}\omega}^{\textup{adj}}) &=& \hat{\se}^2_\textsc{lz}\{\hat\tau_\textsc{i}^\textup{adj}(c_i)\},\\
\hat\tau_{\omega\textsc{a}} &=& \hat\tau_{\textsc{t}},
& \quad &
\hat{\se}^2_\textsc{hw}(\hat\tau_{\omega\textsc{a}}) &=& \hat{\se}^2_\textsc{hw}(\hat\tau_{\textsc{t}}),\\
\hat\tau_{\omega\textsc{a}}^{\textup{adj}} &=& \hat\tau_\textsc{t}^\textup{adj},
& \quad &
\hat{\se}^2_\textsc{hw}(\hat\tau_{\omega\textsc{a}}^{\textup{adj}}) &=& \hat{\se}^2_\textsc{hw}(\hat\tau_\textsc{t}^\textup{adj}).
\end{array}
$$
\end{proposition}

 \begin{table}[t]
\centering
\caption{Estimands, estimators, standard errors, and corresponding theorems in this paper}\label{tb::comparisonofestimators}
\begin{tabular}{llcllllc}
\hline
\multicolumn{3}{c}{estimand $\tau$} &&& \multicolumn{3}{c}{estimand $\tau_\pi$} \\
\cline{1-3} \cline{6-8}
estimator & standard error & Theorem  &&& estimator & standard error & Theorem \\
$\hat\tau_{\textsc{i}}$ & $\hat{\se}_\textsc{lz}(\hat\tau_{\textsc{i}} )$& \ref{thm::tau-i}
&&&$\hat\tau_{\textsc{a}\pi}$ & $\hat{\se}_\textsc{hw}(\hat\tau_{\textsc{a}\pi})$&\ref{thm::wls-without-covariates}  \\
$\hat\tau_{\textsc{i}}^\textup{adj}$  & $\hat{\se}_\textsc{lz}(\hat\tau_{\textsc{i}}^\textup{adj}  ) $&  \ref{thm2}
&&&$\hat\tau_{\textsc{a}\pi}^\textup{adj} $ & $\hat{\se}_\textsc{hw}(\hat\tau_{\textsc{a}\pi}^\textup{adj})$& \ref{thm::wls-with-covariates} \\
$\hat\tau_{\textsc{t}} $  &$ \hat{\se}_\textsc{hw}(\hat\tau_{\textsc{t}} ) $& \ref{thm::tau-T}
&&&$\hat\tau_{\pi\textsc{a}}$& $\hat{\se}_\textsc{hw}(\hat\tau_{\pi\textsc{a}})$& \ref{thm8}  \\
$\hat\tau_{\textsc{t}}^\textup{adj}$  &$\hat{\se}_\textsc{hw}(\hat\tau_{\textsc{t}}^\textup{adj}  ) $& \ref{thm::tau-T-adj}
&&&$\hat\tau_{\pi\textsc{a}}^\textup{adj}$& $ \hat{\se}_\textsc{hw}(\hat\tau_{\pi\textsc{a}}^\textup{adj}) $& \ref{thm9} \\
\hline
\end{tabular}
\end{table}

The equivalence relationships in Proposition \ref{t14} suggest that, in general, $( \hat\tau_{\textsc{a}\pi}, \hat\tau_{\textsc{a}\pi}^\textup{adj})$ are similar to $( \hat\tau_{\textsc{i}}, \hat\tau_{\textsc{i}}^\textup{adj})$, and $( \hat\tau_{\pi\textsc{a}}, \hat\tau_{\pi\textsc{a}}^\textup{adj})$ are similar to $(\hat\tau_{\textsc{t}}, \hat\tau_{\textsc{t}}^\textup{adj} ) $. Proposition \ref{t14} also explains some discussions before: $\hat\tau_{\textsc{a}\pi}^\textup{adj}$ may not gain efficiency over $\hat\tau_{\textsc{a}\pi}$ for the same reason as that $\hat\tau_{\textsc{i}}^\textup{adj}$ may not gain efficiency over $\hat\tau_{\textsc{i}}$.

Based on the comparison of asymptotic efficiency in Corollaries \ref{corollary::improvingIbyT} and \ref{corollary::improve-Api-by-piA}, we finally recommend the use of $\hat\tau_{\textsc{t}}^\textup{adj} (n_i,\tilde{x}_{i\cdot})$ for estimating $\tau$ and $\hat\tau_{\pi\textsc{a}}^\textup{adj} (\pi_i,\pi_i c_i )$ for estimating $\tau_\pi$. Importantly, the former should include $n_i$ and $\tilde{x}_{i\cdot}$ as cluster-level covariates, and the latter should include $\pi_i$ and $\pi_i c_i$ as cluster-level covariates. For both estimators, the heteroskedasticity-robust standard errors are convenient approximations to the true asymptotic standard errors.

The regression estimators based on cluster totals and averages are more efficient although they use only summary statistics of the individual data. This feature allows for using them when only clustered or grouped administrative data are available and the individual-level data are difficult to obtain due to privacy concerns. \citet{schochet2020analyzing} discusses the efficiency issue in this setting, comparing some estimators based on aggregate data and individual data.

Our theory deals with general individual-level covariates $x_{ij}$ and cluster-level covariates $c_i$, but we must address the issue of  defining them in practice. Assume that $x_{ij}^*$ are the basic covariates varying with both clusters and individuals, $c_i^*$ are the basic covariates varying only with clusters, and $n_i$ is the cluster size. For instance, in an experiment with classroom-level intervention, $x_{ij}^*$ can be the pre-test score of student $j$ in classroom $i$, and $c_i^*$ can be the years of education of the instructor for classroom $i$. To conduct individual-level regressions, we can choose $x_{ij}$ to include $(x_{ij}^*, c_i^*, n_i)$ and their functions, and to conduct cluster-level regressions, we can choose $c_i$ to include $(\tilde{x}_{i\cdot}, c_i^*, n_i)$ and their functions. 
The final covariates $x_{ij}$ or $c_i$ must be properly centered for different estimators  discussed above. 
Since our theory is model-assisted, the consistency and asymptotic Normality hold with any choice of the covariates as long as the regularity conditions hold. However, the number of clusters is not extremely large in most cluster-randomized experiments, so we must specify a parsimonious regression model to avoid the bias and efficiency loss due to overfitting.

\section{Simulation}
\label{sec:simulation-section}

This section uses simulation to evaluate the finite-sample properties of various estimators under cluster-randomized experiments, complementing the asymptotic theory in previous sections that assumes $M\rightarrow \infty$. We compare various estimators and standard errors under different data generating processes. We present the simulation results using Figure \ref{fig:simulation-studies} in the main paper but relegate more detailed information to Tables \ref{tab:simulation-studies} and \ref{tab:simulation-studies2} in the appendix.

\subsection{Comparing estimators and standard errors}\label{sec::compare-est-se}

Assume $M=160$ and $e=0.3$. Cluster $i$ has $n_{i}\sim \text{round}\{\text{Unif}(3000/M\times0.6,3000/M\times1.4)\}$ units. Unit $(i,j)$ has covariate $x_{ij} \sim  i/M+\text{Unif}(-1,1)$, and potential outcomes $Y_{ij}(1)\sim \mathcal N\{2n_iM/N+(x_{ij}-\bar x)^3,1\}, \ Y_{ij}(0)\sim \mathcal N\{i/M+(x_{ij}-\bar x)^2,1\}$, which are nonlinear in $x_{ij}$. We draw independent samples and fix them in simulation. Over 1000 replications, we draw cluster-level treatment indicators, obtain observed individual outcomes, and calculate various estimators and standard errors. Figure \ref{fig1} presents the boxplots of the estimators and coverage rates of the $95\%$ Wald-type confidence intervals based on various variance estimators.

In Figure \ref{fig1}, the horizontal line in the boxplots shows the true $\tau$ and that in the barplots shows the $95\%$ nominal level, respectively. 
Numerically, $\hat{\tau}_{\omega\textsc{a}}$ and $\hat{\tau}_{\omega\textsc{a}}^\textup{adj}$ are the same as $\hat{\tau}_{\textsc{t}}$ and $\hat{\tau}_{\textsc{t}}^\textup{adj}$, so we omit them. From Figure \ref{fig1}, $\hat\tau_{\textsc{t}}^\textup{adj}(n_i)$ performs better than $\hat\tau_{\textsc{i}}$, and $\hat{\tau}_{\textsc{t}}^\textup{adj}(n_i,\tilde{x}_{i\cdot})$ performs the best among all the estimators, confirming Corollary \ref{corollary::improvingIbyT}; $\hat\tau_{\textsc{a}} $ and $\hat\tau_{\textsc{a}}^\textup{adj}(\bar x_{i\cdot})$ are inconsistent for $\tau$, confirming Theorems \ref{thm::wls-without-covariates} and \ref{thm::wls-with-covariates}; $\hat\tau_{\textsc{a}\omega}$ and $\hat\tau_{\textsc{a}\omega}^{\textup{adj}}(\bar x_{i\cdot})$ perform the same as $\hat{\tau}_\textsc{i}$ and $ \hat{\tau}_{\textsc{i}}^\textup{adj}(\bar x_{i\cdot})$, confirming Proposition \ref{t14}. The coverage rates are often lower than the $95\%$ nominal level when using $\hat\se_\textsc{ols}$ and $\hat\se_\textsc{hw}$ for regression estimators based on individual data, and $\hat\se_\textsc{ols}$ for regression estimators based on cluster totals and averages. Therefore, we do not recommend those standard errors, and do not use them in other simulation studies and the application in the next section.

We end this subsection with a note on the improvement of efficiency  of $\hat\tau_\t(n_i)$ over $\hat\tau_\i$ asymptotically. 
The key reason for this result is that the residual sum of squares from the OLS fit of $\sum_{j=1}^{n_i}Y_{ij}(z)M/N -\bar Y(z)$ on $n_iM/N-1$ equals to the residual sum of squares from the OLS fit of $\tilde{\varepsilon}_{i\cdot}(z)$ on $n_iM/N-1$, which is smaller than $
\sum_{i=1}^M\{\tilde{\varepsilon}_{i\cdot}(z)\}^2
$. 
In the above data generating process, we include $n_i$ for generating $Y_{ij}(1)$ so that the improvement of $\hat\tau_\t(n_i)$ over $\hat\tau_\i$ is significant. 
However, if $n_i$ is not predictive to $\tilde{\varepsilon}_{i\cdot}(z)$ at all,  then the difference between $\hat\tau_\t(n_i)$ and $\hat\tau_\i$ is small.
Section \ref{sec::simulation-real-data} will give simulation results to show this.

\subsection{Unrelated covariates}\label{sec::noise_covariate}

The last subsection shows a case in which covariates are predictive to the potential outcomes and thus covariate adjustment greatly improves the estimation efficiency. However, covariates may not be predictive in some applications. Although covariate adjustment cannot hurt the asymptotic efficiency based on our theoretical analysis, it is important to see whether the finite-sample results match the asymptotic analysis.
We modify the data generating process to $Y_{ij}(1)\sim \mathcal N (2n_iM/N,1)$ and $Y_{ij}(0)\sim \mathcal N( i/M,1)$, removing all the $x_{ij}$ term so that $x_{ij}$ is completely unrelated to the potential outcomes. 

Figure \ref{fig::noise_covariate} shows the boxplots and coverage rates. Covariate adjustment based on $x_{ij}$ does not improve the efficiency although the finite-sample cost is negligible in our setting. It is important to include the cluster size $n_i$ in the cluster total regressions because $\hat{\tau}_{\t}^\a(n_i)$ still improves both $\hat{\tau}_{\t}$ and $\hat{\tau}_\i$. The cluster size can be an important predictor of the cluster totals, which is a distinct feature of cluster-randomized experiments compared to completely randomized experiments on individuals.

Nevertheless, our results are fundamentally asymptotic with $M\rightarrow \infty $. With unrelated covariates and small $M$, we should apply covariate adjustment with caution since it might result in undesirable finite-sample bias without any efficiency gain.

\subsection{Counter-example for efficiency loss in covariate adjustment}\label{sec::6.2}

We give an example where the asymptotic variance of $\hat\tau_\textsc{i}^\textup{adj}$ is larger than that of $\hat\tau_{\textsc{i}}$, so covariate adjustment in individual-level regression may cause efficiency loss. By Proposition \ref{t14}, this example also apply to $\se^2(\hat\tau_{\textsc{a}\pi}^\textup{adj})\ge \se^2(\hat\tau_{\textsc{a}\pi})$ by choosing $\pi_i=\omega_i$ and $c_i=\bar{x}_{i\cdot}$. Assume $M=100$ and $e=1/2$. We summarize the data generating process below, where $k=1,\ldots,M/4$ and $j=1,\ldots,n_i$:

\begin{center}
\begin{tabular}{ccccc}
\hline
& $i=4k-3$ & $i=4k-2$& $i=4k-1$& $i=4k$  \\
 \hline
$n_{i}$    & $20$ & $20$ & $30$ & $10$\\
$x_{ij}$   & $-5$ & $-5$ & $4 $ & $8 $\\
$Y_{ij}(1)$& $-1$ & $0 $ & $-1$ & $5 $\\
$Y_{ij}(0)$& $0 $ & $0 $ & $ 0$ & $0 $ \\
 \hline
\end{tabular}
\end{center}

The finite population consists of $M/4$ replications of each of the four types of clusters.
 For instance, the first cluster above has size $n_i=20$, and the covariates $x_{ij}=-5$ and potential outcomes $\{ Y_{ij}(1), Y_{ij}(0)\} =( -1, 0)$ do not vary across units. Other clusters are similar.
Because $x_{ij}$ do not vary within cluster $i$, and we have shown that $\hat\tau_{\textsc{a}\omega}$ and $\hat\tau_{\textsc{a}\omega}^{\textup{adj}}(\bar{x}_{i\cdot})$ perform the same as $\hat\tau_{\textsc{i}}$ and $\hat\tau_\textsc{i}^\textup{adj}(\bar{x}_{i\cdot})$. So we omit them. Figure \ref{fig::2} shows that $\hat{\tau}_{\textsc{i}}^\textup{adj}$ is slightly more variable than $\hat{\tau}_\textsc{i}$, and the former indeed has larger variance than the latter.

\subsection{Cluster-randomized experiment with a large but not dominant cluster}
\label{sec::Extremeclusters1}

In this subsection, we show the finite-sample properties of various estimators under a case with a large but not dominant cluster. The data generating process is the same as Section \ref{sec::6.2} except that the last cluster has $n_M=50$ units. Figure \ref{fig::largecluster1} shows the boxplots and coverage rates. The estimators based on regressions with scaled cluster totals are biased when the cluster-level covariates include the cluster size. To demonstrate that it is the cluster size $n_i$ causing the inconsistency, we also report $\hat{\tau}_{\textsc{t}}^\textup{adj}(\bar{x}_{i\cdot})$ and find that it is nearly unbiased. In contrast, all the estimators based on regressions with individual-level data have small biases, illustrating the trade-off of efficiency and robustness issue we discussed in Section  \ref{sec::compare1}.

 \begin{figure}
\begin{subfigure}{\textwidth}
  \centering
\includegraphics[width=\textwidth]{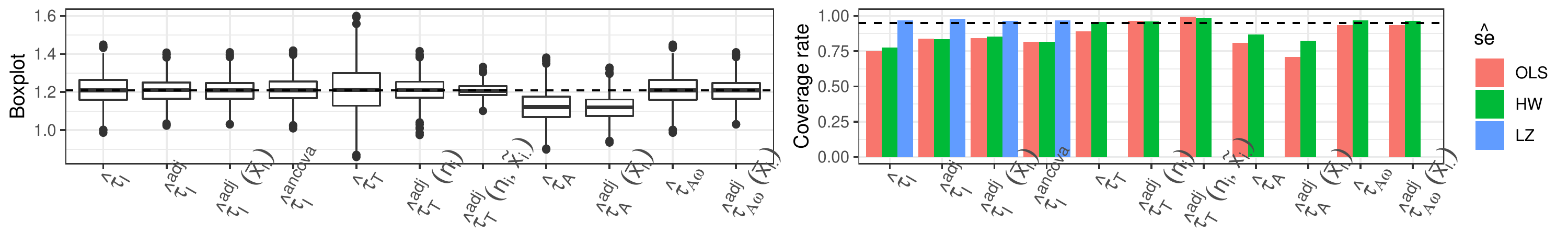}
\caption{Simulation in Section \ref{sec::compare-est-se}: comparing the estimators  }\label{fig1}
\end{subfigure}

\bigskip \bigskip
\begin{subfigure}{\textwidth}
  \centering
\includegraphics[width=\textwidth]{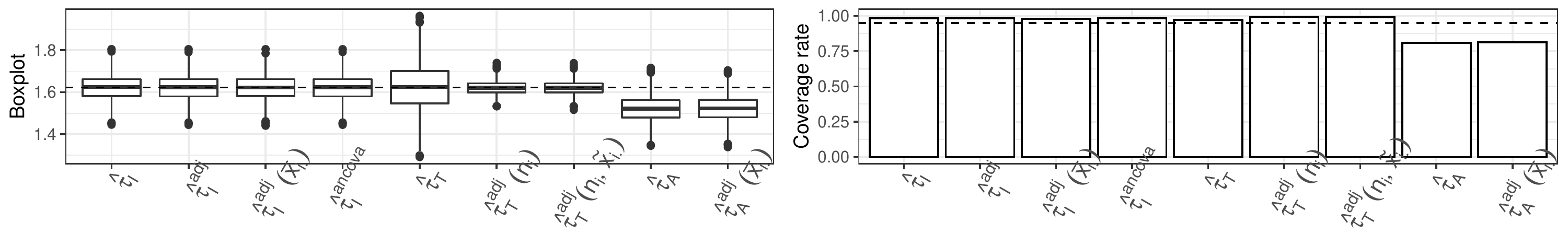}
\caption{Simulation in Section  \ref{sec::noise_covariate}: noise covariate}\label{fig::noise_covariate}
\end{subfigure}

\bigskip \bigskip
\begin{subfigure}{\textwidth}
  \centering
\includegraphics[width=\textwidth]{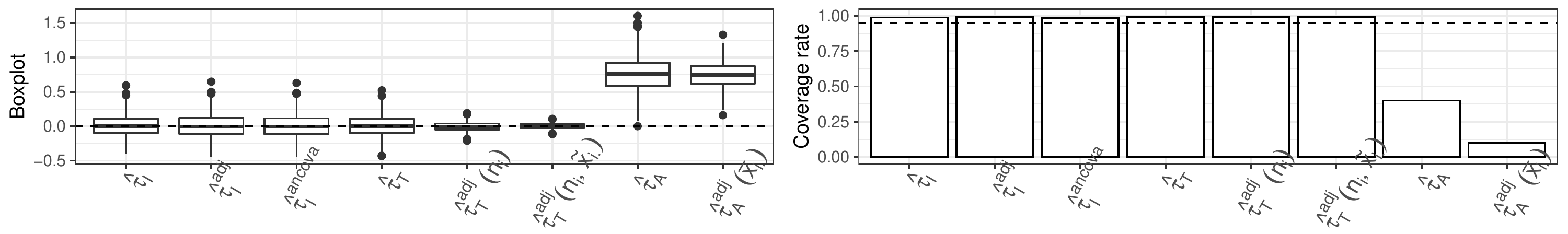}
\caption{Simulation in Section  \ref{sec::6.2}: a counter-example}\label{fig::2}
\end{subfigure}

\bigskip \bigskip
\begin{subfigure}{\textwidth}
  \centering
\includegraphics[width=\textwidth]{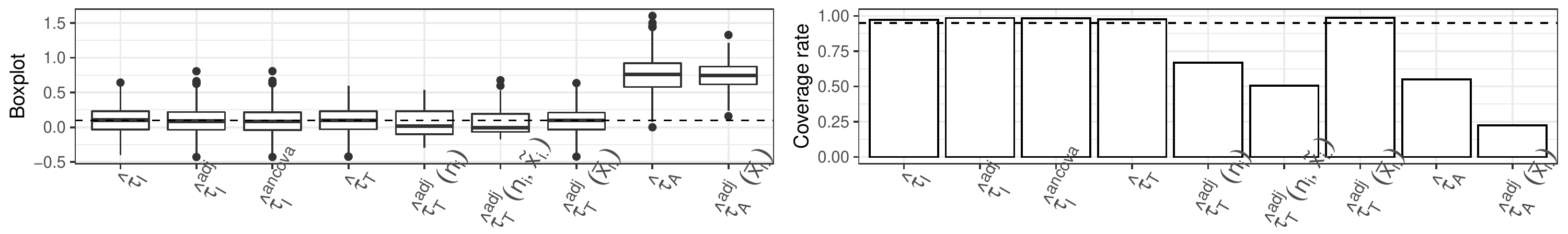}
\caption{Simulation in Section  \ref{sec::Extremeclusters1}: with a large cluster}\label{fig::largecluster1}
\end{subfigure}

\bigskip \bigskip
\begin{subfigure}{\textwidth}
  \centering
\includegraphics[width=\textwidth]{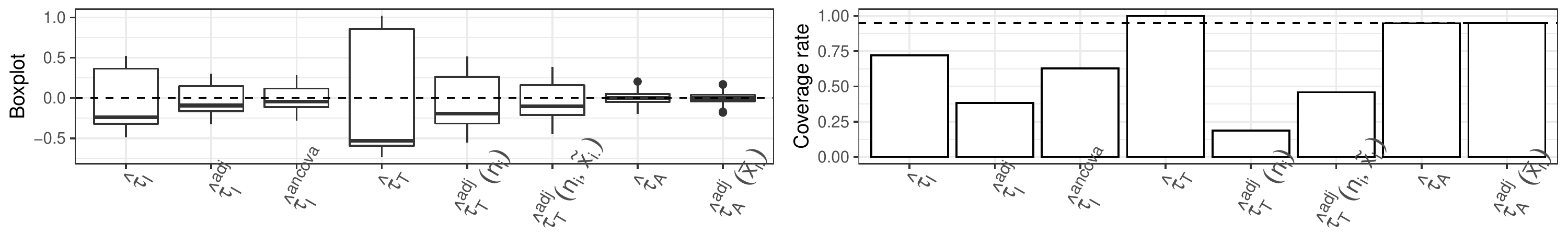}
\caption{Simulation in Section  \ref{sec::extremeclusters2}: with a dominant cluster }\label{fig2}
\end{subfigure}
\caption{Simulation}
\label{fig:simulation-studies}
\end{figure}

\subsection{Cluster-randomized experiment with a dominant cluster}
\label{sec::extremeclusters2}

We further investigate a case with a dominant cluster. Assume $M=100$ and $e=0.4$. For $i=1,\ldots,M-1$, cluster $i$ has $n_{i}\sim \text{round}\{\text{Unif}(2000/M\times0.6,2000/M\times1.4)\}$ units. Cluster $M$ has $n_M=800$ units. Unit $(i,j)$ has covariate $x_{ij}=i/M$ and potential outcomes $Y_{ij}(1)=Y_{ij}(0)\sim \mathcal N\{(i/M)^2+(x_{ij}-\bar x)^2,1\}$ so that the sharp null hypothesis of no treatment effect holds. Figure \ref{fig2} shows the boxplots and coverage rates.  In this case, the regularity conditions on the cluster sizes are violated rendering estimators based on regression with individual data and cluster totals not asymptotically Normal. Ironically, the unbiased estimator $\hat\tau_{\textsc{t}}$ is the worst among all estimators. In contrast, the estimators based on regression with cluster averages $\hat\tau_{\textsc{a}}$ and $\hat\tau_{\textsc{a}}^{\textup{adj}}$ are well-behaved: they are centering at the truth zero with coverage rates close to the nominal level as shown in  Figure \ref{fig2}. 

This example shows that changing the estimand from $\tau$ to $\tau_\pi$, we can use regressions based on cluster averages to reduce bias and increase efficiency when the cluster sizes vary dramatically. They can be particularly useful as test statistics for testing the sharp null hypothesis that $Y_{ij}(1)=Y_{ij}(0) \ (i=1,\ldots,M; \ j=1,\ldots,n_i)$ because $\tau = \tau_\pi = 0$ under this null hypothesis.

\section{Application}\label{sec::application}

\subsection{Data analysis}

Poor sanitation leads to morbidity and mortality in developing countries. In 2012, \citet{guiteras2015encouraging} conducted a cluster-randomized experiment in rural Bangladesh to evaluate the effectiveness of different policies on the use of hygienic latrines. To illustrate our theory, we use a subset of their original data: 63 villages were randomly assigned to the treatment arm receiving subsidies, and 22 villages were randomly assigned to the control arm without subsidies.  We exclude the households not eligible for subsidies or with missing outcomes, resulting in 10,125 households in total. The median, mean, and maximum village sizes are 83, 119, and 500, respectively. We choose the outcome $Y_{ij}$ as the binary indicator for whether the household $(i,j)$ had access to a hygienic latrine or not, measured in June 2013, and covariate $x_{ij}$ as the access rate to hygienic latrines in the community that household $(i,j)$ belonged to, measured in January 2012 before the experiment.

We report various estimators in  Figure \ref{fig::app} and present more details in Table \ref{table::version2} in the appendix. We omit $\hat\tau_{\textsc{a}\omega}$ and $\hat\tau_{\textsc{a}\omega}^{\textup{adj}}(\bar{x}_{i\cdot})$ because they are identical to $\hat\tau_{\textsc{i}}$ and $\hat\tau_\textsc{i}^\textup{adj}(\bar{x}_{i\cdot})$.  Figure \ref{fig::app} shows that $\hat{\tau}_{\textsc{t}}^\textup{adj}(n_i,\tilde{x}_{i\cdot})$ has the smallest standard error, and $\hat\tau_{\textsc{t}}$ has the largest standard error. Comparing these two estimators, covariate adjustment in the cluster total regression reduces the standard error to about 25\%, which is equivalent to increase the number of clusters about $16$ times. This dramatic improvement is not surprising given the poor performance of $\hat\tau_{\textsc{t}}$ in our simulation studies. Compared to the regression estimators with individual data, $\hat{\tau}_{\textsc{t}}^\textup{adj}(n_i,\tilde{x}_{i\cdot})$ has modest improvement in the estimated standard error. Our theory suggests the use of $\hat{\tau}_{\textsc{t}}^\textup{adj}(n_i,\tilde{x}_{i\cdot})$, based on which we conclude that averaged over all households, subsidies increased the probability of getting access to hygienic latrines by 15.2\% with statistical significance at level 95\%, corroborating the analysis of \citet{guiteras2015encouraging}.

\subsection{Simulation based on the data}\label{sec::simulation-real-data}

We only have $22$ villages under control, making the asymptotic approximations suspicious. To evaluate the estimators further, we conduct additional simulation with missing potential outcomes in \citet{guiteras2015encouraging} imputed based on model fitting. Although not all estimators perform well, $\hat{\tau}_{\textsc{t}}^\textup{adj}(n_i,\tilde{x}_{i\cdot})$ has small bias and the associated Wald-type confidence interval achieves the nominal coverage rate, reassuring our choice of the estimator and confidence interval.

We fit a linear probability model of $Y_{ij}$ on $(1,Z_{ij},x_{ij},Z_{ij}x_{ij})$, and then impute the missing potential outcomes as Bernoulli random variables. Pretending that these are the true potential outcomes, we can compute $\tau$, various estimators and the associated standard errors. Figure \ref{fig::simulation-real-data} shows the boxplots  and coverage rates of estimators.
In this case, we have a relatively small number of clusters in the control group, and several estimators perform poorly: $\hat\tau_{\textsc{a}}$
and $\hat\tau_{\textsc{a}}^{\a}(\bar{x}_{i\cdot})$ are  biased because $\bar{\tau}$ differs from $\tau$ for our chosen seed for simulation; $\hat{\tau}_{\t}^\a(n_i)$ is biased although the asymptotic theory suggests that it is consistent. Fortunately, $\hat{\tau}_\i$, $\hat{\tau}_{\i}^{\a}$, $\hat{\tau}_{\i}^{\a}(\bar x_{i\cdot})$, $\hat{\tau}_{\i}^{\textup{ancova}}$, $\hat{\tau}_{\t}$, and $\hat{\tau}_{\t}^\a(n_i,\tilde{x}_{i\cdot})$ have small biases and the associated Wald-type confidence intervals are conservative. 
We  relegate more detailed results to Table \ref{tb::simulation-real-data} in the appendix.

   \begin{figure}[t] 
   
\begin{subfigure}{\textwidth}
  \centering
\includegraphics[width=\textwidth]{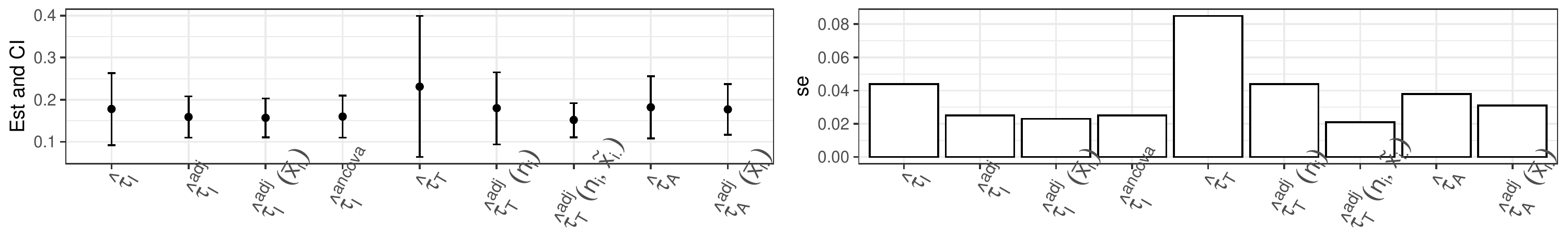}
\caption{ Point estimates with confidence intervals, and standard errors based on \citet{guiteras2015encouraging} }\label{fig::app}
\end{subfigure}

\bigskip \bigskip   
   
   \begin{subfigure}{\textwidth}
\includegraphics[width=\textwidth]{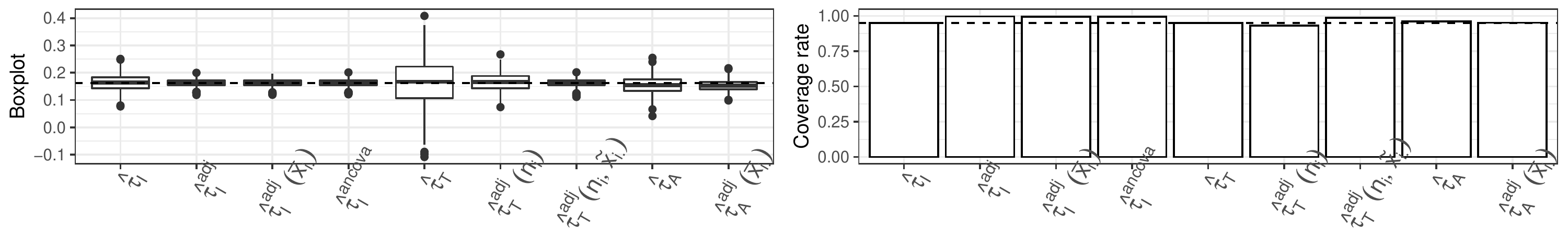}
\caption{Boxplots and coverage rates of the estimators based on the simulated data in Section  \ref{sec::simulation-real-data}}\label{fig::simulation-real-data}
\end{subfigure}

\caption{Application}
\label{fig:application}
\end{figure}

\section{Discussion}
\label{sec::discussion}

Cluster-randomized experiments are widely used in social sciences and public health. Currently, model-based inference is the dominant framework for analyzing them \citep{donner2000design, turner2017reviewanalysis}. Following a growing literature that focuses on the experimental design itself, we re-investigate the theoretical properties of several regression estimators under the design-based perspective that does not require correctly-specified models. We found that the commonly-used OLS estimators with individual data are suboptimal, and a simpler regression estimator based on scaled cluster totals is consistent and efficient with a Normal limiting distribution. We also proposed a regression estimator for a general weighted average treatment effect over clusters.

With clustered individual data, we have discussed regression estimators based on OLS, which are special cases of the generalized estimating equation with the independent working correlation matrix \citep{liang1986longitudinal}. Surprisingly, the simplest version of the generalized estimating equation suffices for generating a consistent and asymptotically Normal estimator for the average treatment effect.  It is interesting to extend the discussion to other distributions, link functions, and working correlation matrices. Another widely-used method  to analyze cluster-randomized experiments is the linear mixed effects model \citep{donner2000design, green2008analysis}. However, the causal interpretation of the estimator from this model is unknown under the design-based perspective that allows the model to be misspecified. This is an intriguing problem that is worth further investigations.

We focused on the asymptotic theory requiring a large number of clusters, but many cluster-randomized experiments have only a small number of clusters. In those cases, we can use bias-corrected robust standard errors \citep{angrist2008mostly} to improve the finite-sample coverage rates. More fundamentally, the asymptotic theory can break down with a small fixed number of clusters. Alternative asymptotic regimes are possible, but they must further restrict the outcome-generating process. One alternative approach is to apply the Bayesian machinery by imposing some parametric assumptions on the data generating process and prior distributions on the unknown parameters. \citet{feller2015hierarchical} reviewed the standard approach of using hierarchical models for causal inference with cluster-randomized experiments. This approach requires strong modeling assumptions, and it is a curious research topic to study its design-based properties. Another alternative approach is to use the Fisher randomization tests that yield finite-sample exact $p$-values \citep{gail1996design, zhang2012powerful, ding2018rank}.  This approach often restricts the model of the individual treatment effect, e.g., testing the sharp null hypothesis. Our asymptotic analysis, coupled with \citet{zhao2020covariate}, suggests that Fisher randomization tests with the estimators studentized by the appropriate standard errors deliver $p$-values that are exact under the sharp null hypothesis and asymptotically valid under the weak null hypothesis of $\tau = 0$. This is our final recommendation for analyzing cluster-randomized experiments with small $M$.

We ignored noncompliance, a common issue in cluster-randomized experiments \citep{frangakis2002clustered, small2008randomization, jo2008cluster, kang2018estimation, kang2018spillover}. Our analysis covers the case with complete randomization on clusters, but many clustered experiments are matched or stratified based on pretreatment covariates \citep{donner2000design, imai2009essential, schochet2021design}. We will investigate these directions in future research.

\section*{Acknowledgement}

We thank the associate editor, two reviewers, Zhichao Jiang, and Anqi Zhao for their insightful comments. 
The research of Peng Ding was partially funded by the U.S. National Science Foundation (grant \# 1945136).

\bibliographystyle{plainnat}
\bibliography{causal}


\newpage
\setcounter{equation}{0}
\setcounter{section}{0}
\setcounter{figure}{0}
\setcounter{example}{0}
\setcounter{proposition}{0}
\setcounter{corollary}{0}
\setcounter{theorem}{0}
\setcounter{table}{0}
\setcounter{condition}{0}

\renewcommand {\theproposition} {A\arabic{proposition}}
\renewcommand {\theexample} {A\arabic{example}}
\renewcommand {\thefigure} {A\arabic{figure}}
\renewcommand {\thetable} {A\arabic{table}}
\renewcommand {\theequation} {A\arabic{equation}}
\renewcommand {\thelemma} {A\arabic{lemma}}
\renewcommand {\thesection} {A\arabic{section}}
\renewcommand {\thetheorem} {A\arabic{theorem}}
\renewcommand {\thecorollary} {A\arabic{corollary}}
\renewcommand {\thecondition} {A\arabic{condition}}
\renewcommand {\thepage} {A\arabic{page}}

\setcounter{page}{1}

\begin{center}
\bf \Large
Supplementary material
\end{center}


\bigskip

Section \ref{appendix::complete-random} reviews and extends the results for completely randomized experiments.

Section \ref{sec::reg-individualdata} proves the results related to regression estimators based on individual-level data.

Section \ref{sec::reg-clustertotaldata} proves the results related to regression estimators based on scaled cluster totals.

Section  \ref{sec::weightedestimands-reg} proves the results related to the general weighted estimand and estimators.

Section \ref{sec::details-simulation} gives more details for  Sections \ref{sec:simulation-section} and \ref{sec::application}.

We will   use
 ``$\leqholder$'' for the steps invoking H\"older's inequality. Similar to the main text, ``$\se^2$'' denotes the variances of the estimators and ``var'' denotes the variances of other quantities.

\section{Completely randomized experiments with $n_i=1$}\label{appendix::complete-random}

We review and extend the results of completely randomized experiments to pave the way for the proofs in later sections. The regularity conditions here are weaker than the existing ones. We analyze  $\hat\tau$ with   standard error $\hat {\se}_\textsc{hw}(\hat\tau)$, and \cite{lin2013}'s estimator $\hat\tau^{\textup{adj}}$ with  standard error $\hat {\se}_\textsc{hw}(\hat\tau^{\textup{adj}})$.
When $n_i=1$, $N=M$, $\mathcal{T}$ and $\mathcal{C}$ index the treated and control units, respectively, and $n_\TG=eN$ and $n_\CG=(1-e)N$ are fixed. Let $\bar{Y}_\TG$ and $\bar Y_\CG$ be the means of the unit-level outcomes, and $\bar{x}_\TG$ and $\bar x_\CG$  be the means of the unit-level covariates, under treatment and control, respectively.

\subsection{Difference in means}

\begin{lemma}\label{l1}
$\hat\tau$ has expectation $E(\hat\tau)=\tau$ and variance
$
\se^2(\hat\tau) = V\{  \varepsilon_i(z)  \} / (M-1).
$

\end{lemma}

\begin{proof}[Proof of Lemma \ref{l1}]
It follows from \citet[][Theorems 6.1 and 6.2]{imbens2015causal}.
\end{proof}

\begin{lemma}\label{DIMA}
Under Assumption \ref{assume::4}, if $M^{-1}\sum_{i=1}^M Y_i(z)^2=o(M)$, then $\hat\tau=\tau+o_\mathbb P(1)$; if $M^{-1}\sum_{i=1}^M Y_i(z)^4=o(M)$ and $M\times \se^2(\hat\tau)\not\rightarrow 0$, then $(\hat\tau-\tau)/\se(\hat\tau)
\rightsquigarrow \mathcal{N} (0,1) $ and
$
M\times \hat{\se}_\textsc{hw}^2(\hat\tau)=  V_\textup{c}\{\varepsilon_i(z)\} +o_\mathbb P(1).
$
\end{lemma}

\begin{proof}[Proof of Lemma \ref{DIMA}] We first prove the consistency. By Lemma \ref{l1} and Assumption \ref{assume::4},
$$E(\hat\tau)=\tau, \quad \se^2(\hat\tau)
=O(M^{-2})\sum_{i=1}^M \left\{ \varepsilon_i(1)^2+\varepsilon_i(0)^2 \right\}
\le O(M^{-2})\sum_{i=1}^M \left\{ Y_i(1)^2+Y_i(0)^2 \right\}=o(1).$$
So  Chebyshev's inequality implies the consistency of $\hat{\tau}$.

We then prove the asymptotic Normality. Because $|\v_i(z)|
=|Y_i(z)-\bar Y(z)|\le |Y_i(z)|+|\bar Y(z)|\le 2\max_{1\le i \le M}|Y_i(z)|
$, we have
\begin{eqnarray*}
M^{-1}\max_{1\le i \le M} \varepsilon_i(z)^2
& \le&  O(M^{-1})\max_{1\le i \le M} Y_i(z)^2
= O(M^{-1})\left\{
\max_{1\le i \le M} M^{-1}Y_i(z)^4
\right\}^{1/2}M^{1/2}   \\
&\le & O(M^{-1})\left\{
M^{-1}  \sum_{i=1}^M Y_i(z)^4
\right\}^{1/2}M^{1/2}   =o(1) .
\end{eqnarray*}
So the asymptotic Normality of $\hat\tau$ follows from
Theorem 1 of \citet{li2017general}.

We finally prove the conservativeness of $\hat{\se}_\textsc{hw}(\hat\tau)$. \citet[][page 228]{angrist2008mostly} gave
$$
M\times \hat{\se}_\textsc{hw}^2(\hat\tau)=M\left\{ \frac{
\sum_{i\in \mathcal{T}}(Y_i-\bar Y_\TG)^2
}{n_\TG^2}+
\frac{
\sum_{i\in \mathcal{C}}(Y_i-\bar Y_\CG)^2
}{n_\CG^2}\right\}.
$$
Using the properties of simple random sampling \citep{li2017general}, we have
$$
E\left(\frac{M\sum_{i\in \mathcal{T}}Y_i^2}{n_\TG^2}\right)
=\frac{\sum_{i=1}^M Y_i(1)^2}{eM}, \quad \var\left(\frac{M\sum_{i\in \mathcal{T}}Y_i^2}{n_\TG^2}\right)
=O(M^{-2})\sum_{i=1}^M Y_i(1)^4=o(1),
$$
which, coupled with Chebyshev's inequality, imply
$$
\frac{M\sum_{i\in \mathcal{T}}Y_i^2}{n_\TG^2} =\frac{\sum_{i=1}^M Y_i(1)^2}{eM}+o_\mathbb P(1).
$$
Similar to the proof of consistency, we can show
$
\bar{Y}_\TG=\bar Y (1)+o_\mathbb P(1).
$
These two results imply that
\begin{eqnarray}
\label{eqn::S5}
M\left\{ \frac{
\sum_{i\in \mathcal{T}}(Y_i-\bar Y_\TG)^2
}{n_\TG^2}\right\}=
   \frac{1}{M}\sum_{i=1}^M \frac{ \varepsilon_i(1)^2}{e}+o_\mathbb P(1).
\end{eqnarray}
A result similar to \eqref{eqn::S5} holds for the control group. Therefore, the lemma holds.
\end{proof}

\subsection{\citet{lin2013}'s estimator}

\begin{assumption}\label{assume::S2}
The potential outcomes satisfy
$M^{-1} \sum_{i=1}^M Y_i(z)^2=o(M)
$
for $z=0,1$. The covariates satisfy
$
M^{-1} \sum_{i=1}^M
||x_i||_{\infty}^2 =O(1)$ and $M^{-1}  \max_{1\le i \le M} ||x_i||_{\infty}^2=o(1).$

\end{assumption}

Let $\hat Q(1)$ and $\hat Q(0)$ be the coefficients of $x_i$ in the OLS fit of $Y_i$ on $(1,x_i)$ using the observed data under treatment and control, respectively.

\begin{lemma}\label{l11}  Under Assumptions \ref{assume::4}, \ref{a4} and \ref{assume::S2}, $\hat Q(z) = o_\mathbb P(M^{1/2})$ for $z=0,1.$
\end{lemma}

\begin{proof}[Proof of Lemma \ref{l11}]
By symmetry, we only analyze
\begin{eqnarray}\label{e4}
\hat Q(1) =\left\{\sum_{i\in \mathcal{T}} (x_i-\bar x_\TG) (x_i-\bar x_\TG)^{\T}\right\}^{-1}
\sum_{i\in \mathcal{T}} (x_i-\bar x_\TG) (Y_i-\bar{Y}_\TG).
\end{eqnarray}
Without loss of generality, we assume that $x_i$ is one-dimensional. We first consider the numerator of $\hat Q(1) $.
By Assumption \ref{a4}, $
E(M^{-1}\sum_{i\in \mathcal{T}} x_iY_i)=M^{-1}\sum_{i=1}^M  ex_iY_i(1)=O(1)
$.
The properties of simple random sampling \citep{li2017general} and Assumption \ref{assume::S2} imply that
\begin{eqnarray*}
\var\left(M^{-1}\sum_{i\in \mathcal{T}} x_iY_i\right)
=O(M^{-2})\sum_{i=1}^Mx_i^2Y_i(1)^2
&\le& O(1)\left\{ M^{-1} \sum_{i=1}^M Y_i(1)^2\right\} \left(  M^{-1} \max_{1\le i \le M} x_i^2 \right) \\
&=& O(1) o(M) o(1) = o(M).
\end{eqnarray*}
By Chebyshev's inequality, we have
$
M^{-1}\sum_{i\in \mathcal{T}} x_iY_i=o_\mathbb P(M^{1/2}).
$
Similar to the proof of Lemma \ref{DIMA}, we can show that $\bar x_\TG=o_\mathbb P(1)$ and $\bar Y_\TG=\bar Y(1) +o_\mathbb P(1)$. Therefore,
\begin{eqnarray}\label{e5}
 M^{-1}\sum_{i\in \mathcal{T}} (x_i-\bar x_\TG) (Y_i-\bar{Y}_\TG)=o_\mathbb P(M^{1/2}).
\end{eqnarray}

We then consider the denominator of $\hat Q(1) $.
Similar to the proof of \eqref{eqn::S5}, we have
\begin{eqnarray}\label{e6}
\frac{1}{M}\sum_{i\in \mathcal{T}}(x_i-\bar{x}_\TG)^2=
\frac{e}{M}\sum_{i=1}^M x_i^2+o_\mathbb P(1) .
\end{eqnarray}
By Assumptions \ref{assume::4} and \ref{a4}, $e\sum_{i=1}^M x_i^2/M$ converges to a finite and invertible matrix. So by \eqref{e6},
\begin{eqnarray}\label{e66}
\left\{\frac{1}{M}\sum_{i\in \mathcal{T}}(x_i-\bar{x}_\TG)^2\right\}^{-1}=
\left(\frac{e}{M}\sum_{i=1}^M x_i^2\right)^{-1}+o_\mathbb P(1)=O_\mathbb P(1) .
\end{eqnarray}
From \eqref{e4}, \eqref{e5} and \eqref{e66}, we have $\hat Q(1)=o_\mathbb P(M^{1/2}).$
\end{proof}

\begin{lemma}\label{l14}
Let $r_i(z)$ be the residual from the OLS fit of $\varepsilon_i(z)$ on $x_i$, and
let
$
\se^2(\hat\tau^{\textup{adj}}) = V\{ r_i(z) \} / M.
   $
Under Assumptions \ref{assume::4}, \ref{a4} and \ref{assume::S2}, $ \hat\tau^{\textup{adj}}=\tau+o_\mathbb P(1)$; if further $M^{-1}\sum_{i=1}^M Y_i(z)^2=O(1),$ $M^{-1} \max_{1\le i \le M} Y_i(z)^2=o(1)$, and
$M\times \se^2(\hat\tau^{\textup{adj}})\not\rightarrow 0$, then $(\hat\tau^{\textup{adj}}-\tau)/\se(\hat\tau^{\textup{adj}})
\rightsquigarrow \mathcal{N} (0,1) $ and
$
M\times \hat{\se}_\textsc{hw}^2(\hat\tau^{\textup{adj}}) = V_\textup{c}\{ r_i(z)  \} + o_\mathbb P(1).
$
\end{lemma}

\begin{proof}[Proof of Lemma \ref{l14}]
By the property of the OLS,
$$\hat\tau^{\textup{adj}}=n_\TG^{-1}\sum_{i\in \mathcal T} \{Y_i-x_i^{\T}\hat Q(1)\}- n_\CG^{-1}\sum_{i\in \mathcal C}\{Y_i-x_i^{\T}\hat{Q}(0)\}.$$
By Lemma \ref{l1}, $ M^{-1}\sum_{i=1}^M
||x_i||_{\infty}^2=O(1)$, and Chebyshev's inequality, we have $\bar x_\TG=O_\mathbb P(M^{-1/2})$. Lemma \ref{l11} ensures $\hat Q(1)=o_\mathbb P(M^{1/2})$, which further implies $\bar x_\TG^{\T}\hat Q(1)=o_\mathbb P(1)$. Similarly, $\bar x_\CG^{\T}\hat Q(0)=o_\mathbb P(1)$. Thus, $\hat\tau^{\textup{adj}}-\hat\tau=o_\mathbb P(1)$. By Lemma \ref{DIMA}, we have $\hat\tau^{\textup{adj}}-\tau=o_\mathbb P(1)$.

The asymptotic Normality and variance estimation of $\hat\tau^{\textup{adj}}$ follow from \citet[][Proposition 3 and Theorem 8]{li2019rerandomization}. Although \citet{li2019rerandomization} assumed that $M^{-1} \sum_{i=1}^M Y_i(z)^2$ has a finite limit while we only assume $M^{-1} \sum_{i=1}^M Y_i(z)^2=O(1)$, the proof of \citet{li2019rerandomization} is still applicable.
\end{proof}

\section{Regression estimators based on individual-level data}
\label{sec::reg-individualdata}

\subsection{Some basic lemmas and their proofs}\label{subsection::lemmas-individualregressions}

We first show Lemma \ref{lem::18}, which is useful for deriving the probability limit of the inverse of a matrix.
\begin{lemma}\label{lem::18}
If $\Delta=O_\mathbb P(\mu)$ with $\mu=o(1)$, and $\Lambda$ converges in probability to a finite and invertible matrix, then
$(  {\Lambda} +  {\Delta} ) ^ { - 1 } -  {\Lambda} ^ { - 1 } =O_\mathbb P(\mu).$
\end{lemma}

\begin{proof}[Proof of Lemma \ref{lem::18}]
By $\Delta=o_\mathbb P(1)$, we have ${\Lambda} +  {\Delta}={\Lambda} +o_\mathbb P(1)$. Because $\Lambda$ converges in probability to a finite and invertible matrix,
$
 (  {\Lambda} +  {\Delta} ) ^ { - 1 }={\Lambda}^{-1} +o_\mathbb P(1).
$
\citet[][Lemma A10]{li2019rerandomization} stated that
$$
(  {\Lambda} +  {\Delta} ) ^ { - 1 } -  {\Lambda} ^ { - 1 } = {\Lambda} ^ { - 1 }  {\Delta} (  {\Lambda} +  {\Delta} ) ^ { - 1 }  {\Delta}  {\Lambda} ^ { - 1 } - {\Lambda} ^ { - 1 }  {\Delta}  {\Lambda} ^ { - 1 },
$$
which implies
$
(  {\Lambda} +  {\Delta} ) ^ { - 1 } -  {\Lambda} ^ { - 1 } =
O_\mathbb P(1) O_\mathbb P(\mu) O_\mathbb P(1) O_\mathbb P(\mu) O_\mathbb P(1) - O_\mathbb P(1) O_\mathbb P(\mu) O_\mathbb P(1)
= O_\mathbb P(\mu).
$
\end{proof}

In a cluster-randomized experiment with varying cluster sizes, $\bar Y_\TG, \ \bar x_\TG$, and $n_\TG/N$ are random but have probability limits $\bar Y(1), \ \bar x$, and $e$, respectively. The following lemma is useful for deriving their properties.

\begin{lemma}\label{moment}
Under Assumption \ref{assume::4}, if $\left(a_{ij}\right)_{1\le i \le M,1\le j\le n_i}$ satisfies $ N^{-1}\sum_{ij}a_{ij}^2=O(1)$, then
\begin{eqnarray}\label{e91}
N^{-1}\sum_{ij\in \mathcal T} a_{ij}- N^{-1} \sum_{ij} a_{ij}e=O_\mathbb P  (\Omega^{1/2})   .
\end{eqnarray}
When $a_{ij}=1$ for all $(i,j)$, \eqref{e91} reduces to
$
n_\TG/N-e=O_{\mathbb P}   (\Omega^{1/2})   .
$
If further $\Omega=o(1)$, then
\begin{eqnarray}\label{e92}
n_\TG^{-1} \sum_{ij\in \mathcal T} a_{ij}-  N^{-1} \sum_{ij} a_{ij}=O_\mathbb P
  (\Omega^{1/2})   .
\end{eqnarray}
\end{lemma}

\begin{proof}[Proof of Lemma \ref{moment}]
First,
$E\left(
\sum_{ij\in \mathcal T}  a_{ij}/N \right)=
\sum_{ij} a_{ij}e/N.$
By the properties of simple random sampling \citep{li2017general},
\begin{eqnarray*}
\var\left( \sum_{ij\in \mathcal T}a_{ij}/N\right)
&=&
\frac{(eM)^2}{N^2}\var\left( \frac{1}{eM} \sum_{i=1}^M
Z_i\sum_{j=1}^{n_{i}}a_{ij} \right)
\le  \frac{(eM)^2}{N^2}\frac{
\sum_{i=1}^M \left(\sum_{j=1}^{n_i}a_{ij}\right)^2
}{eM(M-1)}\\
&\leqholder&
\frac{eM}{(M-1)N^2}  \sum_{i=1}^M n_i \sum_{j=1}^{n_i}a_{ij}^2
\leq \frac{eM}{M-1}  \Omega \left( N^{-1} \sum_{ij} a_{ij}^2 \right)
=O(1)  \Omega O(1)
=O\left( \Omega\right).
\end{eqnarray*}
Chebyshev's inequality ensures \eqref{e91}.

As a special case, \eqref{e91} implies $n_\TG/N-e=O_\mathbb P(\Omega^{1/2})$. If we further assume $\Omega=o(1)$, then Lemma \ref{lem::18} implies $N/n_\TG=1/e+O_{\mathbb P}   (\Omega^{1/2})   $ with $\Lambda + \Delta = n_\TG/N, \Lambda = e$ and $\Delta = n_\TG/N-e$. So
\begin{eqnarray*}
n_\TG^{-1} \sum_{ij\in \mathcal T} a_{ij}
&=&N^{-1}\sum_{ij\in \mathcal T} a_{ij}\times  ( N/n_\TG )
= \left\{  N^{-1} \sum_{ij} a_{ij} e +O_\mathbb P  (\Omega^{1/2}) \right\}  \left\{ 1/e+O_{\mathbb P}   (\Omega^{1/2})  \right\} \\
&=& N^{-1} \sum_{ij} a_{ij}  +O_{\mathbb P}   (\Omega^{1/2})+  O_{\mathbb P}   (\Omega)
= N^{-1} \sum_{ij} a_{ij}  +O_{\mathbb P}   (\Omega^{1/2}),
\end{eqnarray*}
implying \eqref{e92}.
\end{proof}

Let $\hat{Q}_\textsc{i}(1)$ and $\hat{Q}_\textsc{i}(0)$ be the coefficients of $x_{ij}$ in the OLS fit of $Y_{ij}$ on $(1,x_{ij})$ using the units in the treatment and control groups, respectively.

\begin{lemma}\label{l3}
Under Assumptions \ref{assume::4}--\ref{a4}, if $\Omega  =o(1)$, then $\hat{Q}_\textsc{i}(z)-Q_\textsc{i}(z)=O_\mathbb P  (\Omega^{1/2})   $ for $z=0,1.$
\end{lemma}

\begin{proof}[Proof of Lemma \ref{l3}]
By symmetry, we only  analyze
$$
\hat{Q}_\textsc{i}(1)=\left\{ \sum_{ij\in \mathcal T} (x_{ij}-\bar{x}_\TG)
(x_{ij}-\bar{x}_\TG)^{\T}  \right\}^{-1} \sum_{ij\in \mathcal T} (x_{ij}-\bar{x}_\TG)
(Y_{ij}-\bar{Y}_\TG).
$$
Because $\Omega  =o(1)$, Lemma \ref{moment} ensures
$$
\frac{1}{N}\sum_{ij\in \mathcal T}x_{ij} Y_{ij}
= \frac{e}{N}\sum_{ij} x_{ij} Y_{ij}(1)+O_\mathbb P  (\Omega^{1/2})
,\quad \bar x_\TG=O_\mathbb P  (\Omega^{1/2})   , \quad \bar Y_\TG=\bar Y(1)+ O_\mathbb P  (\Omega^{1/2})   .$$
Therefore, the numerator and denominator of $\hat Q_\i(1)$ satisfy
\begin{eqnarray*}
\frac{1}{N}\sum_{ij\in \mathcal T} (x_{ij}-\bar{x}_\TG)
(Y_{ij}-\bar{Y}_\TG) &=& \frac{e}{N}\sum_{ij} x_{ij} Y_{ij}(1)+O_\mathbb P  (\Omega^{1/2})   ,\\
\frac{1}{N} \sum_{ij\in \mathcal T} (x_{ij}-\bar{x}_\TG)
(x_{ij}-\bar{x}_\TG)^{\T} &=&  \frac{e}{N}\sum_{ij} x_{ij} x_{ij}^\T + O_\mathbb P  (\Omega^{1/2}) .
\end{eqnarray*}
We can use Lemma \ref{lem::18} to derive $\hat{Q}_\textsc{i}(1)-Q_\textsc{i}(1)=O_\mathbb P  (\Omega^{1/2})   $.
\end{proof}

\subsection{Proof of consistency and asymptotic Normality}

\begin{proof}[Proof of Theorem \ref{thm::tau-i}]
We first prove the consistency. Define $\bar\varepsilon_\TG=
n_\TG^{-1}
\sum_{ij\in \mathcal T} \varepsilon_{ij}(1)
$. Because ${\Omega}=o(1)$, Lemma \ref{moment} ensures
$\bar \varepsilon_\TG=
O_\mathbb P(\Omega^{1/2}).
$
We have an analogous result for the control group. Therefore, $\hat\tau_{\textsc{i}} =\tau+o_\mathbb P(1)$.

We then prove the asymptotic Normality. Lemma \ref{moment}  ensures
\begin{eqnarray*}
M^{1/2}\left(\frac{1}{n_\TG} -\frac{1}{eN} \right)\sum_{ij\in \mathcal T} \varepsilon_{ij}(1)
= M^{1/2} e^{-1} \left(e-\frac{n_\TG}{N}\right)
 n_\TG^{-1} {\sum_{ij\in \mathcal T} \varepsilon_{ij}(1)}
=O_\mathbb P\left(M^{1/2}\Omega  \right) = O_\mathbb P\left(M^{1/2} M^{-2/3}  \right) = o_\mathbb{P}(1).
\end{eqnarray*}
An analogous result holds for the control group. These two results imply that $M^{1/2}(\hat\tau_{\textsc{i}}-\tau)  - R_M  =o_\mathbb P(1)$, where
\begin{eqnarray*}
R_M &=& M^{1/2} \left[  (eN)^{-1} \sum_{ij}Z_{ij}\varepsilon_{ij}(1) - \{ (1-e)N \}^{-1} \sum_{ij} (1-Z_{ij})\varepsilon_{ij}(0) \right]
\\
&=& M^{1/2}  \left[  (eM)^{-1}  \sum_{i=1}^M  Z_i
 \sum_{j=1}^{n_{i}}\varepsilon_{ij}(1)M/ N
-  \{ (1-e)M\}^{-1}    \sum_{i=1}^M   (1-Z_i) \sum_{j=1}^{n_{i}}\varepsilon_{ij}(0)M / N  \right] \\
&=&  M^{1/2}  \left[  (eM)^{-1}  \sum_{i=1}^M  Z_i \tilde{\varepsilon}_{i\cdot}(1)
-  \{ (1-e)M\}^{-1}    \sum_{i=1}^M   (1-Z_i)  \tilde{\varepsilon}_{i\cdot}(0)  \right]
\end{eqnarray*}
is the difference-in-means of the scaled cluster totals of the residuals in a completely randomized experiment on clusters. Because $V\{ \tilde{\varepsilon}_{i\cdot}(z) \}\not\rightarrow 0$, $( \hat\tau_{\textsc{i}} - \tau )/
\se(\hat\tau_{\textsc{i}})=M^{1/2}(\hat\tau_{\textsc{i}}-\tau)/V\{ \tilde{\varepsilon}_{i\cdot}(z) \}^{1/2}$ has the same limiting distribution as  $R_M/V\{ \tilde{\varepsilon}_{i\cdot}(z) \}^{1/2}\rightsquigarrow \mathcal{N} (0,1) $ ensured by Lemma \ref{DIMA} because the regularity condition
\begin{eqnarray*}
\frac{1}{M}\sum_{i=1}^M
\tilde{\varepsilon}_{i\cdot}(z)^4 &=&
\frac{1}{M}\sum_{i=1}^M \left\{\sum_{j=1}^{n_{i}}\varepsilon_{ij}(z)M/N \right\}^4
\leqholder \frac{M^3}{N^4}\sumM
n_i^3\sum_{j=1}^{n_i}\varepsilon_{ij}(z)^4 \\
&\le& \frac{M^3}{N^4}\max_{1\le i \le M} n_i^3 O(N)
=M^3 \Omega^3 O(1)
=o(M)
\end{eqnarray*}
holds by $N^{-1}\sum_{ij}\varepsilon_{ij}(z)^4=O(1)$ and $\Omega = o(M^{-2/3})$.

We leave the proofs of the standard errors to the next subsection due to its length.
\end{proof}

\begin{proof}[Proof of Theorem \ref{thm2}]
We first prove the consistency.
By the property of the OLS,
$$
\hat\tau_{\textsc{i}}^\textup{adj}
=n_\TG^{-1}{\sum_{ij\in \mathcal T}  \{ Y_{ij} -  x_{ij}^{\T} \hat{Q}_\textsc{i}(1)\}  }
-n_\CG^{-1} {\sum_{ij\in \mathcal C}  \{ Y_{ij} -  x_{ij}^{\T} \hat{Q}_\textsc{i}(0)\} }.
$$
Because $\hat{Q}_\textsc{i}(1) = O_\mathbb P (1)$ and $\bar x_\TG
=o_\mathbb P (1)$, we have $\bar x_\TG^\T \hat{Q}_\textsc{i}(1) = o_\mathbb P (1)$. Similarly, $\bar x_\CG^\T  \hat{Q}_\textsc{i}(0)  =o_\mathbb P (1)$. Thus, $\hat\tau_{\textsc{i}}^\textup{adj}-\hat\tau_{\textsc{i}}=o_\mathbb P (1)$, which, coupled with Theorem \ref{thm::tau-i}, implies $\hat\tau_{\textsc{i}}^\textup{adj}
=\tau+o_\mathbb P(1)$.

We then prove the asymptotic Normality. Lemma \ref{moment} ensures $\bar x_\TG=O_\mathbb P (\Omega^{1/2})$ and Lemma \ref{l3} ensures
$
\hat{Q}_\textsc{i}(1) - Q_\textsc{i}(1)=O_\mathbb P (\Omega^{1/2}).
$
These two results, coupled with the assumption $\Omega  =o(M^{-2/3})$, imply
$$
M^{1/2}\bar x_\TG^{\T} \{\hat{Q}_\textsc{i}(1) - Q_\textsc{i}(1) \}
=M^{1/2} O_\mathbb P (\Omega^{1/2})O_\mathbb P (\Omega^{1/2})
=M^{1/2}O_\mathbb P (\Omega)
=M^{1/2} o_\mathbb P ( M^{-2/3})
=o_\mathbb P(1).
$$
Similarly, $M^{1/2}\bar x_\CG ^{\T} \{\hat{Q}_\textsc{i}(0) - Q_\textsc{i}(0)\}=o_\mathbb P(1)$ holds for the control group. So $
(\hat\tau_{\textsc{i}}^\textup{adj} - \tau)/\se(\hat\tau_{\textsc{i}}^\textup{adj}) =
M^{1/2}(\hat\tau_{\textsc{i}}^\textup{adj} - \tau)/V\{\tilde r_{i\cdot}(z)\}^{1/2}$ has the same asymptotic distribution as
$$
M^{1/2}\left[n_\TG^{-1} {\sum_{ij\in \mathcal T}  \left\{\varepsilon_{ij}(1) -  x_{ij}^{\T}  Q_\textsc{i}(1) \right\}    }
-n_\CG^{-1}{\sum_{ij\in \mathcal C}  \left\{\varepsilon_{ij}(0) -  x_{ij}^{\T}  Q_\textsc{i}(0)\right\}    }\right]/V\{\tilde r_{i\cdot}(z)\}^{1/2}
.
$$
Define the potential outcomes as $\varepsilon_{ij}(z)-x_{ij}^{\T} Q_\textsc{i}(z)$ for $z=0,1$, then we can apply Theorem \ref{thm::tau-i} to show the asymptotic distribution of $\hat\tau_{\textsc{i}}^\textup{adj}$.

We leave the proofs of the standard errors to the next subsection due to its length. Here we verify that if the $n_i$'s are equal and $x_{ij}=\bar{x}_{i\cdot}$ do not vary with $j$, then $\se^2(\hat\tau_{\textsc{i}}^\textup{adj})\le \se^2(\hat\tau_{\textsc{i}})$. This holds because
\begin{eqnarray*}
Q_\textsc{i}(z)
&=& \left( \sum_{ij} x_{ij}   x_{ij}^{\T} \right) ^ { - 1 } \sum _ { ij } x_{ij} \varepsilon_{ij}(z)
= \left( \sum_{i=1}^M n_i \bar{x}_{i\cdot}  \bar{x}_{i\cdot}^{\T} \right)^{-1}\sum_{i=1}^M \bar{x}_{i\cdot}  \sum_{j=1}^{n_i}\varepsilon_{ij}(z)\\
&=&\left( \sum_{i=1}^M n_i \bar{x}_{i\cdot}  n_i \bar{x}_{i\cdot}^{\T} \right)^{-1}\sum_{i=1}^M n_i\bar{x}_{i\cdot} \sum_{j=1}^{n_i}\varepsilon_{ij}(z)
= \left( \sum_{i=1}^M \tilde{x}_{i\cdot}\tilde{x}_{i\cdot}^{\T} \right)
^ { - 1 } \sum_{i=1}^M \tilde{x}_{i\cdot}\tilde{\varepsilon}_{i\cdot}(z),
\end{eqnarray*}
which equals the coefficient of $\tilde{x}_{i\cdot}$ in the OLS fit of $\tilde{\varepsilon}_{i\cdot}(z)$ on $\tilde{x}_{i\cdot}$.
\end{proof}

\begin{proof}[Proof of Corollary \ref{pro2}]
We only derive the probability limit of $\hat{Q}_\textsc{i}$, the coefficient of $x_{ij}$ in the OLS fit of $Y_{ij}$ on $(1, Z_{ij}, x_{ij})$. Other derivations are similar to those of Theorem \ref{thm2}. Therefore, we omit them.

We can obtain $\hat{Q}_\textsc{i}$ using the Frisch--Waugh--Lovell Theorem \citep{angrist2008mostly}.
The residual from the OLS fit of $Y_{ij}$ on $(1, Z_{ij})$ is $\check{Y}_{ij}=Y_{ij}-Z_{ij}\bar Y_{\TG}-(1-Z_{ij})\bar Y_{\CG}$, and the
 residual from the OLS fit of $ x_{ij}$ on $(1,Z_{ij})$ is $\check{x}_{ij}=x_{ij}-Z_{ij}\bar x_{\TG}-(1-Z_{ij})\bar x_{\CG}$. So
$\hat{Q}_\textsc{i}  = (\sum_{ij} \check x_{ij}\check x_{ij}^\T)^{-1}\sum_{ij} \check x_{ij}\check Y_{ij}$. By Lemma \ref{moment},
the numerator of $\hat{Q}_\textsc{i}$ satisfies
\begin{eqnarray*}
N^{-1}\sum_{ij} \check x_{ij}\check Y_{ij}
&=&N^{-1}\sum_{ij\in \mathcal T} (x_{ij}-\bar x_{\TG})(Y_{ij}-\bar Y_{\TG})+N^{-1}
\sum_{ij\in \mathcal C} (x_{ij}-\bar x_{\CG})(Y_{ij}-\bar Y_{\CG})\\
&=&N^{-1}\sum_{ij\in \mathcal T} x_{ij}Y_{ij}- n_\TG\bar x_\TG \bar Y_\TG/N+
N^{-1}\sum_{ij\in \mathcal C} x_{ij}Y_{ij}-n_\CG \bar x_\CG \bar Y_\CG/N\\
&=&e\sum_{ij}x_{ij}Y_{ij}(1)/N+(1-e)\sum_{ij}x_{ij}Y_{ij}(0)/N+O_\mathbb P  (\Omega^{1/2}).
\end{eqnarray*}
Its denominator has a similar result. By Lemma \ref{lem::18},
$\hat{Q}_\textsc{i} -e Q_\textsc{i}(1)-(1-e)Q_\textsc{i}(0)=O_\mathbb P  (\Omega^{1/2}) $.
\end{proof}

\subsection{Proofs of the results about the standard errors}\label{sec::2.2}

We use the following notation in the following proofs. Let $[\cdot]_{(2,2)}$ denote the $(2,2)$ element of the matrix inside $[\cdot]$; let $[\cdot]_{(1,1)+(2,2)-2(1,2)}$ denote the sum of the $(1,1)$th and the $(2,2)$th elements minus twice the $ (1,2)$th element of the matrix $[\cdot]$; let $(1-2,1-2)$ denote the submatrix of the first two rows and the first two columns.

Define $X_i$ as an $n_i\times (2+2p_x)$ matrix with row $j$ equaling $[1 ,Z_ { ij } ,  x _ {ij}^{\T} , Z_{ij}x _ { ij }^{\T}]$, $j=1,\ldots,n_i$. Stack $X_i$ together to obtain an $N\times (2+2p_x)$ matrix $X$. Define $\hat r_{ij}$ as the residual from the OLS fit of $Y_{ij}$ on $(1, Z_{ij}, x_{ij}, Z_{ij}x_{ij})$.
Define an $n_i\times n_i$ matrix $\hat{U}_{i}=(\hat{r}_{ij}\hat{r}_{ik})_{1\le j,k\le n_{i}}$. The cluster-robust standard error for the coefficient of $Z_{ij}$ equals
\begin{eqnarray}\label{eq::lzse-individual-reg-form1}
\hat{\se}_\textsc{lz}^2(\hat\tau_{\textsc{i}}^\textup{adj})
= \left[ \left(  { X } ^ { \T}  { X } \right) ^ { - 1 } \left( \sum _ { i = 1 } ^ {M}  { X } _ { i} ^ {\T} \hat{U} _ { i }  { X } _ { i } \right)  \left(  { X } ^ \T  { X } \right) ^ { - 1 } \right]_{(2,2)} .
\end{eqnarray}
This formula covers both the cases with and without covariate adjustment, and we simply set $x_{ij}$ to be empty in the latter case. When $n_i=1$, it reduces the heteroskedasticity-robust standard error.
The lemma below gives an equivalent form of \eqref{eq::lzse-individual-reg-form1} that is easier to analyze.

\begin{lemma}\label{LZ}
Define $\tilde{X}_i$ as an $n_i\times (2+2p_x)$ matrix with row $j$ equaling $\tilde{x}_{ij}^\T= [ Z_ { ij } , 1 - Z_ { ij } ,  Z_{ij}x _ {ij}^{\T} , (1-Z_{ij}) x _ { ij }^{\T} ]$, stacked as an $N\times (2+2p_x)$ matrix  $\tilde{X}$. We have
\begin{eqnarray}\label{eq::lzse-individual-reg-form2}
 \hat \se_\textsc{lz}^2(\hat\tau_{\textsc{i}}^\textup{adj})
 = \left[ \left(\tilde {  { X } } ^\T \tilde {  { X } }  \right) ^ { - 1 } \left( \sum_{i=1}^M
 {  \tilde{ X}_{i} }  ^\T\hat{U}_{i}{  {\tilde  X_{i} } }
 \right)  \left( \tilde {  { X } } ^\T \tilde {  { X } }  \right) ^ { - 1 }\right]_ { (1,1)+(2,2)-2(1,2) }.
\end{eqnarray}
\end{lemma}

\begin{proof}[Proof of Lemma \ref{LZ}]
Define
$$
R=
  \left( \begin{array} { c c c c c } { 1 } & { 1 } & {0} & { 0 } \\ { 1 } & { 0 } & { 0} & { 0 } \\ { 0} & { 0 } & { I_{p_x} } & { I_{p_x}} \\{ 0} & { 0 } & { I_{p_x} } & { 0} \end{array} \right),\qquad
  R^{-1}=
\left( \begin{array} { c c c c } { 0 } & { 1 } & { 0 } & { 0 } \\ { 1 } & { -1 } & { 0 } & { 0 } \\ { 0 } & { 0 } & { 0 } & { I_{p_x} } \\ { 0 } & { 0 } & { I_{p_x} } & { -I_{p_x} } \end{array} \right),\qquad
X_i=\tilde X_i R,\qquad
{X}=\tilde X R.
$$
  Then the cluster-robust standard error has the following equivalent forms:
\begin{eqnarray*}
\eqref{eq::lzse-individual-reg-form1}
&=&\left[\left(R^{\T}\tilde{{X}}^{\T} \tilde{{X}}
R\right)^{-1}\left(\sum_{i=1}^M R^{\T}\tilde{{X}}_{i}^{\T} \hat{U}_{i} \tilde{{X}}_{i}R
\right)
\left(R^{\T}\tilde{{X}}^{\T} \tilde{X}
R\right)^{-1}\right]_{(2,2)}\\
&=&\left[R^{-1} \left( \tilde {  { X } } ^\T \tilde {  { X } }  \right) ^ { - 1 } \left( \sum_{i=1}^M
 {  \tilde{ X}_{i} }  ^\T\hat{U}_{i}{  {\tilde  X_{i} } }
 \right)  \left( \tilde {  { X } } ^\T \tilde {  { X } }  \right) ^ { - 1 }R^{-1}\right]_{(2,2)} = \eqref{eq::lzse-individual-reg-form2}.
\end{eqnarray*}
\end{proof}

Partition $G= N^{-1} \tilde{{X}}^{\T} \tilde{X}$ into
$$
G= \left( \begin{array}{c c ;{2pt/2pt} c c } n_\TG/N & {0} & {{
\sum_{ij\in \mathcal T}x_{ij}^{\T}
}/N} & {0} \\ {0} & n_\CG/{N} & {0} & {
\sum_{ij\in \mathcal C}x_{ij}^{\T}/{N}
} \\
\hdashline[2pt/2pt]
\sum_{ij\in \mathcal T}x_{ij}/N
 & {0} & { {
\sum_{ij\in \mathcal T}x_{ij}
x_{ij}^{\T}
}/{N}} & {0} \\
0 & {\sum_{ij\in \mathcal C}x_{ij}}/{N}& {0} & {
{\sum_{ij\in \mathcal C}x_{ij}x_{ij}^{\T}}/{N}
}\end{array}\right)=\left(\begin{array} {c c} G_{11} &
 G_{12}
 \\ G_{21}& G_{22} \end{array}\right).
$$
Define
$H_i=\tilde{X}_i^{\T} \hat{U}_i \tilde{X}_i$ and $H=\sum_{i=1}^M H_i M/N^2
$. Below Lemma \ref{l7} gives an explicit formula for $H_i$, and Lemma \ref{l8} gives the probability limit of $H$.

\begin{lemma}\label{l7}
$H_i$ is symmetric and
{\small
$$H_i= \left( \begin{array}{cccc}
{Z_i(\sum_{j=1}^{n_{i}}\hat{r}_{ij})^2} & 0 & {Z_i
\sum_{j=1}^{n_{i}}\hat{r}_{ij}\sum_{j=1}^{n_{i}}
\hat{r}_{ij}x_{ij}^{\T}
} & {0} \\
0 & { (1-Z_i)(\sum_{j=1}^{n_{i}} \hat{r}_{ij})^2 } & 0 &
{(1-Z_i)}
\sum_{j=1}^{n_{i}}\hat{r}_{ij}\sum_{j=1}^{n_{i}}
\hat{r}_{ij}x_{ij}^{\T}
 \\
{*} & {0} & Z_i
\sum_{j=1}^{n_{i}}\hat{r}_{ij}x_{ij}
\sum_{j=1}^{n_{i}}\hat{r}_{ij}x_{ij}^{\T}
& 0 \\
0  & {*} & {0} & (1-Z_i)
\sum_{j=1}^{n_{i}}\hat{r}_{ij}x_{ij}
\sum_{j=1}^{n_{i}}\hat{r}_{ij}x_{ij}^{\T}\end{array}\right),
$$}
where the * elements can be determined by symmetry.
\end{lemma}

\begin{proof}[Proof of Lemma \ref{l7}]
By definition,
{\small
\begin{eqnarray*}
H_i &=& \sum_{j=1}^{n_{i}}\sum_{k=1}^{n_{i}}\hat{r}_{ij} \hat{r}_{ik}
\tilde{x}_{ij}
\tilde{x}_{ik}^{\T}
\\
&=&\sum_{j=1}^{n_{i}}\sum_{k=1}^{n_{i}}\hat{r}_{ij} \hat{r}_{ik}
\left( \begin{array}{cccc}{Z_{ij}Z_{ik}} & {
Z_{ij}(1-Z_{ik})
} & {Z_{ij}Z_{ik}x_{ik}^{\T}} & {Z_{ij} (1-Z_{ik})x_{ik}^{\T}}
 \\
 {(1-Z_{ij})Z_{ik}} & {(1-Z_{ij})(1-Z_{ik})} & {
 (1-Z_{ij})Z_{ik}x_{ik}^{\T}
 } & {
 (1-Z_{ij})
 \left(1-Z_{ik}\right)x_{ik}^{\T}
 }
 \\{Z_{ij}x_{ij}Z_{ik}} & {
Z_{ij}x_{ij} (1-Z_{ik})
} & {Z_{ij}x_{ij} Z_{ik}x_{ik}^{\T}} & {Z_{ij} x_{ij}  (1-Z_{ik})x_{ik}^{\T}}
  \\
{(1-Z_{ij})x_{ij} Z_{ik}} & {
(1-Z_{ij})x_{ij} (1-Z_{ik})
} & {(1-Z_{ij})x_{ij}
Z_{ik}x_{ik}^{\T}} & {(1-Z_{ij}) x_{ij} (1-Z_{ik})x_{ik}^{\T}}\end{array}\right)
\\
&=&\sum_{j=1}^{n_{i}}\sum_{k=1}^{n_{i}} \hat{r}_{ij} \hat{r}_{ik}
\left( \begin{array}{cccc}{Z_i} &0& {Z_ix_{ik}^{\T}} & 0\\
0 & 1-Z_i& 0& {
 (1-Z_i)x_{ik}^{\T}
 }
 \\{Z_ix_{ij} } & {
0
} & {Z_ix_{ij} x_{ik}^{\T}} & 0
  \\
0 & {
(1-Z_i)x_{ij}
} & 0 & {(1-Z_{i}) x_{ij} x_{ik}^{\T}}\end{array}\right)
\\
&=&\left( \begin{array}{cccc}
{Z_i(\sum_{j=1}^{n_{i}}\hat{r}_{ij})^2} & 0 & {Z_i
\sum_{j=1}^{n_{i}}\hat{r}_{ij}\sum_{j=1}^{n_{i}}
\hat{r}_{ij}x_{ij}^{\T}
} & {0} \\
0 & { (1-Z_i)(\sum_{j=1}^{n_{i}} \hat{r}_{ij})^2 } & 0 &
{(1-Z_i)}
\sum_{j=1}^{n_{i}}\hat{r}_{ij}\sum_{j=1}^{n_{i}}
\hat{r}_{ij}x_{ij}^{\T}
 \\
{*} & {0} & Z_i
\sum_{j=1}^{n_{i}}\hat{r}_{ij}x_{ij}
\sum_{j=1}^{n_{i}}\hat{r}_{ij}x_{ij}^{\T}
& 0 \\
0  & {*} & {0} & {(1-Z_i)
\sum_{j=1}^{n_{i}}\hat{r}_{ij}x_{ij}
\sum_{j=1}^{n_{i}}\hat{r}_{ij}x_{ij}^{\T}}\end{array}\right).
\end{eqnarray*}
}
\end{proof}

\begin{lemma}\label{l8}
Let $r_{ij}(z)$ be the residual from the OLS fit of $\varepsilon_{ij}(z)$ on $x_{ij}$. Under Assumptions \ref{assume::4}--\ref{a4}, if $ \Omega  =o(M^{-2/3} )$, then
\begin{eqnarray}
\frac{M}{N^2}\sum_{i=1}^M Z_i
\left(\sum_{j=1}^{n_{i}} \hat r_{ij}\right)^{2}-\frac{M}{N^2} \sum_{i=1}^M e \left\{ \sum_{j=1}^{n_{i}}r_{ij}(1) \right\}^{2}
&=&o_\mathbb P (1),\label{a10}\\
\frac{M}{N^2}\sum_{i=1}^MZ_i \sum_{j=1}^{n_{i}}
 \hat{r}_{ij} \sum_{j=1}^{n_{i}} \hat{r}_{ij} x_{ij}
&=&O_\mathbb P ( M \Omega  ),\label{a11}\\
\frac{M}{N^2}\sum_{i=1}^MZ_i\sum_{j=1}^{n_{i}} \hat{r}_{ij}x_{ij}  \sum_{j=1}^{n_{i}} \hat{r}_{ij}x_{ij}^{\T}
&=&O_\mathbb P\left( M \Omega   \right).\label{aa12}
\end{eqnarray}
We have similar results for the control group. Partitioning $H$ in the same way as $G$, we have $H_{11},H_{12},H_{21}$, and $H_{22}$ all of order $O_{\mathbb P}(   M\Omega    )$, so $H=O_{\mathbb P}(   M\Omega    ).$
\end{lemma}

\begin{proof}[Proof of Lemma \ref{l8}]
Without loss of generality, we assume that $x_i$ is one-dimensional and omit the superscript $^{\T}$s.
First, we prove (\ref{a10}). By the definition of $\hat{r}_{ij}$,
\begin{eqnarray*}
\frac{M}{N^2} \sum_{i=1}^M Z_i  \left(\sum_{j=1}^{n_{i}} \hat{r}_{ij} \right)^{2}
&=&\frac{M}{N^2}\sum_{i=1}^M Z_i   \left[ \sum_{j=1}^{n_{i}} \left\{Y_{ij}(1)
-\bar{Y}_\TG-(x_{ij}-\bar{x}_\TG)\hat{Q}_\i(1)
\right\}  \right]^2 \\
&=&\frac{M}{N^2}\sum_{i=1}^M Z_i \left(
\sum_{j=1}^{n_{i}} \left[\left\{\varepsilon_{ij}(1)-x_{ij}Q_\i(1)\right\}
+x_{ij}Q_\i(1)-\bar{\varepsilon}_\TG
-(x_{ij}-\bar x_\TG)\hat{Q}_\i(1)
\right]\right)^2\\
&=&\frac{M}{N^2} \sum_{i=1}^M Z_i\left( \sum_{j=1}^{n_{i}} \left[r_{ij}(1)
+x_{ij}\{Q_\i(1)-\hat{Q}_\i(1)\}
-\{\bar{\varepsilon}_\TG -\bar{x}_\TG\hat{Q}_\i(1) \}
\right]
  \right)^2\\
  & =& T_1 + T_2 + T_3 + T_4 - T_5 - T_6 .
\end{eqnarray*}
where
{\small
$$
\begin{array}{lllllll}
T_1 &=&  \frac{M}{N^2} \sum_{i=1}^M Z_i \{\sum_{j=1}^{n_{i}}  r_{ij}(1) \}^2,&&
T_2 &=& \frac{M}{N^2} \sum_{i=1}^M Z_i ( \sum_{j=1}^{n_{i}}  x_{ij} )^2 \{Q_\i(1)-\hat{Q}_\i(1)\}^2,\\
T_3 &=& \frac{M}{N^2} \sum_{i=1}^M Z_i n_i^2  \{\bar{\varepsilon}_\TG -\bar{x}_\TG\hat{Q}_\i(1) \}^2 ,&&
T_4 &=&  \frac{2M}{N^2} \sum_{i=1}^M Z_i  \sum_{j=1}^{n_{i}}  r_{ij}(1) \sum_{j=1}^{n_{i}}  x_{ij} \{Q_\i(1)-\hat{Q}_\i(1)\},\\
T_5 &=& \frac{2M}{N^2} \sum_{i=1}^M Z_i  \sum_{j=1}^{n_{i}}  r_{ij}(1) n_i  \{\bar{\varepsilon}_\TG -\bar{x}_\TG\hat{Q}_\i(1)\}, &&
T_6 &=& \frac{2M}{N^2} \sum_{i=1}^M Z_i  \sum_{j=1}^{n_{i}}  x_{ij}  n_i \{Q_\i(1)-\hat{Q}_\i(1)\} \{\bar{\varepsilon}_\TG -\bar{x}_\TG\hat{Q}_\i(1) \} .
\end{array}
$$}
We claim that except for $T_1$, all other terms $T_2$--$T_6$ are of order $o_\mathbb P (1)$. We show that $T_4=o_\mathbb P (1)$ and omit the proofs for other terms. It is bounded from the above by
$$
\left| T_4\right|
\le  \frac{2M}{N^2} \sum_{i=1}^M \left|    \sum_{j=1}^{n_i} r_{ij}(1) \sum_{j=1}^{n_i}x_{ij}      \right| \left|Q_\i(1)-\hat Q_\i(1)\right|
\le  \frac{M}{N^2} \sum_{i=1}^M \left[\left\{\sum_{j=1}^{n_i} r_{ij}(1)\right\}^2+\left( \sum_{j=1}^{n_i}x_{ij}      \right)^2 \right]\left|Q_\i(1)-\hat Q_\i(1)\right| .
$$
By Lemma \ref{l3}, $Q_\i(1)-\hat{Q}_\i(1)=O_\mathbb P  (\Omega^{1/2})$, and by
 Assumption \ref{assume::1},
$$
 \frac{M}{N^2}\sum_{i=1}^M
\left\{ \sum_{j=1}^{n_{i}}r_{ij}(1) \right\}^2
 {\leqholder}  \frac{M}{N^2}\sum_{i=1}^M
n_i\sum_{j=1}^{n_{i}} r_{ij}(1)^2
= O\left(M\Omega  \right),\quad
\frac{M}{N^2}\sum_{i=1}^M {\left(\sum_{j=1}^{n_{i}}x_{ij}\right)^2} = O\left(M\Omega  \right),
$$
we have
$\left| T_4\right| \leq O(M\Omega) O_\mathbb P(\Omega^{1/2})=o_\mathbb P(1)$ using the assumption $\Omega  =o(M^{-2/3})$. To finish the proof of \eqref{a10}, we only need to verify that $T_1$ differs from its mean by a term of order $o_\mathbb P(1)$, which follows from Chebyshev's inequality and the variance calculation:
$$
\var(T_1) \le
\frac{M^2}{N^4}\sum_{i=1}^M \left\{ \sum_{j=1}^{n_{i}}r_{ij}(1) \right\}^4
\leqholder \frac{M^2}{N^4}\sum_{i=1}^M n_i^3\sum_{j=1}^{n_{i}}r_{ij}(1)^4
\le O\left(M^2 \Omega^3\right)=o(1) .
$$

%

Second, we prove \eqref{a11}. By definition of $\hat{r}_{ij}$,
\begin{eqnarray*}\label{aa19}
&&\frac{M}{N^2}\sum_{i=1}^MZ_i\sum_{j=1}^{n_{i}}
 \hat{r}_{ij} \sum_{j=1}^{n_{i}} \hat{r}_{ij}x_{ij}  \\
&=&\frac{M}{N^2}\sum_{i=1}^M Z_i\sum_{j=1}^{n_{i}}\left[
\varepsilon_{ij}(1)-
x_{ij}\hat{Q}_\i(1)
- \{\bar{\varepsilon}_\TG-\bar x_\TG  \hat{Q}_\i(1) \}
\right]\sum_{j=1}^{n_{i}}\left[
\varepsilon_{ij}(1)-
x_{ij}\hat{Q}_\i(1)
- \{\bar{\varepsilon}_\TG-\bar x_\TG  \hat{Q}_\i(1) \}
\right]x_{ij} \\
&=& T_7 - T_8 - T_9 - T_{10} + T_{11} + T_{12} - T_{13} + T_{14} + T_{15},
\end{eqnarray*}
where
{\small
$$
\begin{array}{lllllll}
T_7&=& \frac{M}{N^2}
\sum_{i=1}^M Z_i \sum_{j=1}^{n_i}
\varepsilon_{ij}(1)\sum_{j=1}^{n_i}\varepsilon_{ij}(1)x_{ij},&&
T_8 &=& \frac{M}{N^2}\sum_{i=1}^M Z_i\sum_{j=1}^{n_i}\varepsilon_{ij}(1)
\sum_{j=1}^{n_i}x_{ij}^2\hat{Q}_\i(1) ,\\
T_9&=&  \frac{M}{N^2}\sum_{i=1}^M Z_i\sum_{j=1}^{n_i}\varepsilon_{ij}(1)
\sum_{j=1}^{n_i}x_{ij}\{\bar \varepsilon_\TG - \bar x_\TG \hat{Q}_\i(1)\}, &&
T_{10} &=& \frac{M}{N^2}\sum_{i=1}^M Z_i
\sum_{j=1}^{n_i}x_{ij}\sum_{j=1}^{n_i}\varepsilon_{ij}(1)x_{ij}\hat{Q}_\i(1), \\
T_{11} &=& \frac{M}{N^2}
\sum_{i=1}^M Z_i \sum_{j=1}^{n_i}
x_{ij}\sum_{j=1}^{n_i}x_{ij}^2 \hat{Q}_\i(1)^2,&&
T_{12} &=& \frac{M}{N^2}\sum_{i=1}^M
 Z_i (\sum_{j=1}^{n_i}x_{ij} )^2\hat{Q}_\i(1)\{\bar\v_\TG
 -\bar x_\TG \hat{Q}_\i(1)\},\\
 T_{13} &=& \frac{M}{N^2}
\sum_{i=1}^M Z_i n_i\sum_{j=1}^{n_i}
\varepsilon_{ij}(1)x_{ij}\{\bar \v_\TG-
\bar x_\TG\hat{Q}_\i(1)\}, &&
T_{14} &=& \frac{M}{N^2}
\sum_{i=1}^M Z_i
n_i\sum_{j=1}^{n_i}x_{ij}^2\{\bar \v_\TG-
\bar x_\TG\hat{Q}_\i(1)\}
\hat{Q}_\i(1),\\
T_{15} &=& \frac{M}{N^2}
\sum_{i=1}^M Z_in_i \sum_{j=1}^{n_i}
x_{ij}\{\bar\varepsilon_\TG-\bar x_\TG \hat{Q}_\i(1)\}^2.
\end{array}
$$}
We claim that all terms $T_7$--$T_{15}$ are of order $O_\mathbb P\left(M\Omega  \right)$. We show that $T_{7}=O_\mathbb P\left(M\Omega  \right)$ and omit the proofs for other terms: by Assumption \ref{assume::1},
\begin{eqnarray*}
|T_7|&\le& \frac{M}{N^2}\sum_{i=1}^M\left|
\sum_{j=1}^{n_{i}} \varepsilon_{ij}(1)\sum_{j=1}^{n_{i}}\varepsilon_{ij}(1)x_{ij}
\right|\\
&\le& O\left(\frac{M}{N^2}\right)\sum_{i=1}^M \left[
 \left\{\sum_{j=1}^{n_{i}} \varepsilon_{ij}(1)\right\}^2+\left\{\sum_{j=1}^{n_{i}}\varepsilon_{ij}(1)x_{ij}\right\}^2 \right]\\
&\leqholder & O\left(\frac{M}{N^2}\right)
\sum_{i=1}^M \left\{
n_i\sum_{j=1}^{n_{i}}\varepsilon_{ij}(1)^2+n_i\sum_{j=1}^{n_{i}}\varepsilon_{ij}(1)^2
x_{ij}^2
 \right\}\\
&\le& O\left(\frac{M}{N^2}\right)\sum_{i=1}^M
n_i
\sum_{j=1}^{n_{i}} \left\{ \varepsilon_{ij}(1)^2+ \varepsilon_{ij}(1)^4+ x_{ij}^4\right\}\le O\left(M\Omega  \right).
\end{eqnarray*}
Therefore, \eqref{a11} holds.

The proof of \eqref{aa12} is similar to (\ref{a11}), and we omit it. Therefore, ${H}=O_{\mathbb P}(   M\Omega    )$ follows from \eqref{a10}--\eqref{aa12}.
\end{proof}

Now we analyze the standard errors in Theorems \ref{thm::tau-i} and \ref{thm2}.

\begin{proof}[Proof of the results about the standard errors in Theorem \ref{thm::tau-i}]
We first analyze $\hat{\se}_\textsc{lz}^2(\hat\tau_{\textsc{i}})$. The matrix $\tilde{X}^\T \tilde{X}$ simplifies to
$$
\tilde{X}^\T \tilde{X}=\sumM \tilde{X}_i^\T \tilde{X}_i
= \sumM Z_i \begin{pmatrix}
n_i&0\\ 0 & 0
\end{pmatrix}
+ \sumM (1-Z_i) \begin{pmatrix}
0&0\\ 0 & n_i
\end{pmatrix}
=\begin{pmatrix}
n_\TG&0\\ 0 & n_\CG
\end{pmatrix}.
$$
The residual from the OLS of $Y_{ij}$ on $(1,Z_{ij})$ is $\hat{\varepsilon}_{ij}$ defined in Section  \ref{sec::asymptotics-individual-reg} of the  main text, so the matrix $\sumM \tilde{X}_i^\T \hat{U}_i \tilde{X}_i$ simplifies to
\begin{eqnarray*}
\sumM \tilde{X}_i^\T \hat{U}_i \tilde{X}_i
&=&  \sumM Z_i \begin{pmatrix}
\left( \sum_{j=1}^{n_i} \hat{\varepsilon}_{ij}\right)^2 & 0\\ 0& 0
\end{pmatrix}
+ \sumM (1-Z_i) \begin{pmatrix}
0&0\\ 0& \left( \sum_{j=1}^{n_i} \hat{\varepsilon}_{ij}\right)^2
\end{pmatrix} \\
&=&\begin{pmatrix}
 \sumM Z_i\left( \sum_{j=1}^{n_i} \hat{\varepsilon}_{ij}\right)^2 & 0\\ 0& \sumM (1-Z_i) \left( \sum_{j=1}^{n_i} \hat{\varepsilon}_{ij}\right)^2
\end{pmatrix}.
\end{eqnarray*}
These two matrices are both diagonal, and it is straightforward to use Lemma \ref{LZ} to obtain the formula of $\hat{\se}_\textsc{lz}^2(\hat\tau_{\textsc{i}})$  in Theorem \ref{thm::tau-i} of the main text. The proof of its conservativeness is the same as that of Theorem \ref{thm2}. So we omit it.

We then analyze $\hat{\se}_\textsc{hw}^2(\hat\tau_{\textsc{i}})$.
\citet[][page 228]{angrist2008mostly} give
$$N\times \hat{\se}^2_\textsc{hw}(\hat\tau_{\textsc{i}})=
N\left(
{\sum_{ij\in \mathcal T} \hat{\varepsilon}_{ij}^2  }/{n_\TG^2}+
{\sum_{ij\in \mathcal C} \hat{\varepsilon}_{ij}^2  }/{n_\CG^2}
 \right)
=N\left\{
{\sum_{ij\in \mathcal T} (Y_{ij}-\bar{Y}_\TG)^2  }/{n_\TG^2}+
{\sum_{ij\in \mathcal C} (Y_{ij}-\bar{Y}_\CG)^2  }/{n_\CG^2}
 \right\}.$$
Therefore,
\begin{eqnarray*}
&&N{\sum_{ij\in \mathcal T} (Y_{ij}-\bar{Y}_\TG)^2  }/{n_\TG^2}
= N{\sum_{ij\in \mathcal T}\{\varepsilon_{ij}(1)-\bar{\varepsilon}_\TG\}^2}/{n_\TG^2} \\
&=& N\left\{ \sum_{ij\in \mathcal T}\varepsilon_{ij}(1)^2 -n_\TG\bar{\varepsilon}_\TG^2    \right\}/{n_\TG^2}
= N\left\{{\sum_{ij\in \mathcal T}
\varepsilon_{ij}(1)^2
}/{n_\TG^2}-{\bar{\varepsilon}_\TG^2}/{n_\TG}\right\}.
\end{eqnarray*}
By Lemma \ref{moment} and Slutsky's Theorem, we have $\bar{\varepsilon}_\TG=o_\mathbb P(1)$ and
\begin{eqnarray*}
N{
 \sum_{ij\in \mathcal T}\varepsilon_{ij}(1)^2
 }/{n_\TG^2}&=&(N/n_\TG)\left\{n_\TG^{-1}\sum_{ij\in \mathcal T}\v_{ij}(1)^2\right\}
 =\left\{1/e+o_\mathbb P(1)\right\}\left\{N^{-1}\sum_{ij}  {\varepsilon_{ij}(1)^2} +o_\mathbb P(1)\right\}\\
&=&(eN)^{-1}\sum_{ij}  {\varepsilon_{ij}(1)^2} +o_\mathbb P(1).
\end{eqnarray*}
Therefore,
$N {\sum_{ij\in \mathcal T} (Y_{ij}-\bar{Y}_\TG)^2  }/{n_\TG^2}=(eN)^{-1}\sum_{ij} {\varepsilon_{ij}(1)^2}
+o_\mathbb P(1).$
We have the same result for control group. Therefore, we get the probability limit of $\hat{\se}_\textsc{hw}^2(\hat\tau_{\textsc{i}})$.
\end{proof}

\begin{proof}[Proof of the standard error in Theorem \ref{thm2}]
Lemma \ref{LZ} implies that
$$
M\times \hat {\se} _ {\textsc{lz} }^2 (\hat\tau_{\textsc{i}}^\textup{adj})
=  \left[  {G} ^ { - 1 }  {H}  {G} ^ { - 1 }  \right] _ {(1,1)+(2,2)-2(1,2) } .
$$
So the key is to analyze $  {G} ^ { - 1 }  {H}  {G} ^ { - 1 }$.
Define
$$
\Lambda=\left( \begin{array} { c c;{2pt/2pt} c c} e&0&0&0  \\
 0  & 1-e  & 0 & 0\\
 \hdashline[2pt/2pt]
 { 0 } & { 0} & e\sum _ { ij}x_{ij} x_{ij}^{\T}  / { N } & 0 \\
 0 & 0 & 0 & (1-e)\sum_{ij}   x_{ij}  x_{ij}^{\T} / { N }\end{array} \right)=\left(\begin{array} {c c} \Lambda_{11} &
 0
 \\ 0 & \Lambda_{22} \end{array}\right)
 $$
 as the expectation of $G.$
We claim that
 \begin{eqnarray}\label{eq::key-cluster-robust-se}
 \left[  {G} ^ { - 1 }  {H}  {G} ^ { - 1 } \right] _ { ( 1 - 2,1 - 2 ) } - \left[ {\Lambda} ^ { - 1 }  {H}  {\Lambda} ^ { - 1 } \right] _ { ( 1 - 2,1 - 2 ) }
 = o _ {\mathbb  P } ( 1 ),
 \end{eqnarray}
 which implies the conclusion because by Lemmas \ref{l7} and \ref{l8},
\begin{eqnarray*}
\left[   { {\Lambda} } ^ { - 1 }   {  {H} }   {  {\Lambda} } ^ { - 1 } \right] _ {(1,1)+(2,2)-2(1,2) }
&=&
\frac{M\sum_{i=1}^M Z_i(\sum_{j=1}^{n_i}\hat{r}_{ij})^2}{N^2e^2}+
\frac{M\sum_{i=1}^M (1-Z_i)(\sum_{j=1}^{n_i}\hat{r}_{ij})^2}{N^2(1-e)^2}\\
&=&\frac{M\sum_{i=1}^M \{\sum_{j=1}^{n_i} {r}_{ij}(1)\}^2}{N^2e } +
\frac{M\sum_{i=1}^M \{\sum_{j=1}^{n_i} {r}_{ij}(0)\}^2}{N^2(1-e) } +
o_\mathbb P(1).
\end{eqnarray*}

 Now we complete the proof by showing \eqref{eq::key-cluster-robust-se}. Define
 $
 {\Delta}= G- \Lambda ,
$ and $\Psi=   {  {G} } ^ { - 1 } -   {  {\Lambda} } ^ { - 1 }.
 $
By Lemma \ref{moment},
$ {\Delta}=  O_{\mathbb P} (\Omega^{1/2})$. By Assumptions \ref{assume::4} and \ref{a4}, $\Lambda$ converges to a finite and invertible matrix. So Lemma \ref{lem::18} ensures
$
\Psi = O_{\mathbb P}  (\Omega^{1/2})   .
$
By definition,
\begin{eqnarray*}
G^ { - 1 }  {H}  {G} ^ { - 1 } -  {\Lambda} ^ { - 1 }  {H}  {\Lambda} ^ { - 1 } &=& \left(  {\Lambda} ^ { - 1 } +  {\Psi} \right)  {H} \left(  {\Lambda} ^ { - 1 } + {\Psi} \right) -  {\Lambda} ^ { - 1 }  {H}  {\Lambda} ^ { - 1 }\\ &=&  {\Psi}  {H}  {\Lambda} ^ { - 1 } +  {\Lambda} ^ { - 1 }  {H}  {\Psi} +  {\Psi}  {H}  {\Psi} ,
\end{eqnarray*}
with
$$ {\Psi}  {H}  {\Lambda} ^ { - 1 } = \left( \begin{array} { l l } {  {\Psi} _ { 11 } } & { {\Psi} _ { 12 } } \\ {  {\Psi} _ { 21 } } & {  {\Psi} _ { 22 } } \end{array} \right) \left( \begin{array} { l l } {  {H} _ { 11 } } & {  {H} _ { 12 } } \\ {  {H} _ { 21 } } & { {H} _ { 22 } } \end{array} \right) \left( \begin{array} { c c } {  \Lambda _ { 11 } ^ { - 1 } } & { 0 } \\ { 0 } & { \Lambda_ { 22 } ^ { - 1 } } \end{array} \right)= \left( \begin{array} { c c c } { {  {\Psi} } _ { 11 }   {  {H} } _ { 11 }  \Lambda_ { 11 } ^ { - 1 } } & { +   {  {\Psi} } _ { 12 }   {  {H} } _ { 21 } \Lambda_ { 11 } ^ { - 1 } } & { * } \\ { * } & { } & { * } \end{array} \right),$$
where $\Psi_{11},\Psi_{12},\Psi_{21}$, and $\Psi_{22}$ are $2\times2,2\times2p_x,2p_x\times2$, and $2p_x\times2p_x$ submatrices of $\Psi$, respectively, and the $*$ elements do not matter in the proof.
 Therefore, by Lemma \ref{l8},
\begin{eqnarray*}
\left[  {\Psi}  {H} {\Lambda} ^ { - 1 } \right] _ { ( 1-2,1-2 ) } &=&  {\Psi} _ { 11 }  {H} _ { 11 }\Lambda_ { 11 } ^ { - 1 } +  {\Psi} _ { 12 }  {H} _ { 21 }\Lambda_ { 11 } ^ { - 1 }\\
&=& O _ { \mathbb P }   (\Omega^{1/2})    O _ { \mathbb P } \left(   M\Omega     \right) O _ \mathbb{ P } ( 1 )
 + O _ {\mathbb P }  (\Omega^{1/2})     O _ { \mathbb P } \left(  M\Omega   \right) O _ \mathbb{ P } ( 1 )=o_{\mathbb P}(1), \\
  {\Psi}  {H}  {\Psi}
 &=& O _ { \mathbb P }   (\Omega^{1/2})    O _ { \mathbb P } \left(   M\Omega     \right)  O _ { \mathbb P }   (\Omega^{1/2})
 = O _ { \mathbb P } ( M \Omega^2 )=o_\mathbb P(1),
\end{eqnarray*}
which imply \eqref{eq::key-cluster-robust-se}.
\end{proof}

\section{Regression estimators based on scaled cluster totals}
\label{sec::reg-clustertotaldata}

For Theorems \ref{thm::tau-T} and \ref{thm::tau-T-adj}, because in $\hat\tau_{\textsc{t}}$ and $\hat\tau_{\textsc{t}}^\textup{adj}$, we view $\tilde{Y}_{i\cdot}=\sum_{j=1}^{n_i}Y_{ij}M/N$ as a new outcome, we only need to check the regularity conditions on $\tilde{Y}_{i\cdot}(z)= \sum_{j=1}^{n_i}Y_{ij}(z)M/N$ given the results in completely randomized experiments.

\begin{proof}[Proof of Theorem \ref{thm::tau-T}]
The assumption $\Omega  =o(1)$ implies
$$
\frac{1}{M}\sum_{i=1}^M \left\{ {\sum_{j=1}^{n_{i}}Y_{ij}(z)M/N} \right\}^2
\leqholder \frac{M}{N^2}\sum_{i=1}^M n_i\sum_{j=1}^{n_{i}}Y_{ij}(z)^2
=O\left(   M\Omega    \right)=o(M).
$$
The assumption $M^{2/3}\Omega  =o(1)$ implies
$$
\frac{1}{M}\sum_{i=1}^M \left\{{\sum_{j=1}^{n_{i}}Y_{ij}(z)}M/N \right\}^4
\leqholder \frac{M^3}{N^4}\sum_{i=1}^M n_i^3\sum_{j=1}^{n_{i}}Y_{ij}(z)^4
=O\left(M^3\Omega^3\right)=o(M).
$$
Lemma \ref{DIMA} guarantees the asymptotic property of $\hat\tau_{\textsc{t}}$.
\end{proof}

\begin{proof}[Proof of Theorem \ref{thm::tau-T-adj}]
The assumption $ \Omega  =o(1)$ implies
$$
\frac{1}{M}\sum_{i=1}^M \left\{{\sum_{j=1}^{n_{i}}Y_{ij}(z)}  M/N\right\}^2
\leqholder \frac{M}{N^2}\sum_{i=1}^M n_i\sum_{j=1}^{n_{i}}Y_{ij}(z)^2
=O\left(   M\Omega    \right)=o(M).
$$
The assumption $M \Omega  =O(1)$ implies
$$
\frac{1}{M}\sum_{i=1}^M \left\{ {\sum_{j=1}^{n_{i}}Y_{ij}(z)} M/N \right\}^4
\leqholder  \frac{M^3}{N^4}\sum_{i=1}^M n_i^3\sum_{j=1}^{n_{i}}Y_{ij}(z)^4
=O\left(M^3\Omega^3\right)=O(1).
$$
Lemma \ref{l14} guarantees the asymptotic property of $\hat\tau_{\textsc{t}}^\textup{adj}$.
\end{proof}

\begin{proof}[Proof of Proposition \ref{proposition::1}]
By definition,
\begin{eqnarray}\label{last1}
E(\hat\tau_\textsc{i}-\tau)^2
&=&E\left\{ {\sum_{ij\in \mathcal T}\varepsilon_{ij}(1)}/{n_\TG}
- {\sum_{ij\in \mathcal C}\varepsilon_{ij}(0)}/{n_\CG}\right\}^2
 = E\left[ \left\{
n_\CG \sum_{ij\in \mathcal T} \varepsilon_{ij}(1)-
n_\TG \sum_{ij\in \mathcal C} \varepsilon_{ij}(0)
\right\}\big/(n_\TG n_\CG)\right]^2\nonumber\\
&\ge&
\left(
 \frac{2}{N}\right)^4E\left\{ n_\CG   \sum_{ij\in \mathcal T} \varepsilon_{ij}(1) -
n_\TG  \sum_{ij\in \mathcal C} \varepsilon_{ij}(0)
 \right\}^2,
\end{eqnarray}
where the last inequality follows from
$
4 n_\TG n_\CG\le
(n_\TG+n_\CG) ^2
=N^2
$.

If $e=1/2$ and $\bar Y(z)=0$ for $z=0,1$, then $\tau=0$ and
\begin{eqnarray}\label{last2}
E(\hat\tau_\textsc{t}-\tau)^2
=E\left( {2\sum_{ij\in \mathcal T} Y_{ij} }/{N}- {2\sum_{ij\in \mathcal C} Y_{ij} }/{N}\right)^2
=\left(\frac{2}{N}\right)^2
E\left\{\sum_{ij\in \mathcal T} \varepsilon_{ij}(1)+\sum_{ij\in \mathcal T}  \varepsilon_{ij}(0) \right\}^2.
\end{eqnarray}
From (\ref{last1}) and (\ref{last2}), we only need to prove
\begin{eqnarray}\label{last3}
E\left\{\sum_{ij\in \mathcal T}\varepsilon_{ij}(1)+\sum_{ij\in \mathcal T}\varepsilon_{ij}(0)\right\}^2\left(\frac{N}{2}\right)^2
\le E\left\{ n_\CG\sum_{ij\in \mathcal T} \varepsilon_{ij}(1)-
n_\TG \sum_{ij\in \mathcal C} \varepsilon_{ij}(0)  \right\}^2.
\end{eqnarray}
Because $\sum_{ij}\varepsilon_{ij}(z)=\sum_{ij\in \mathcal T}\varepsilon_{ij}(z)+\sum_{ij\in \mathcal C}\varepsilon_{ij}(z)=0$, we have
\begin{eqnarray*}
& &\left\{ n_\CG\sum_{ij\in \mathcal T} \varepsilon_{ij}(1)-
n_\TG \sum_{ij\in \mathcal C} \varepsilon_{ij}(0)  \right\}^2
+
\left\{ n_\TG\sum_{ij\in \mathcal C} \varepsilon_{ij}(1)-
n_\CG \sum_{ij\in \mathcal T} \varepsilon_{ij}(0)  \right\}^2\\
&=&
\left\{ n_\CG\sum_{ij\in \mathcal T} \varepsilon_{ij}(1)+
n_\TG \sum_{ij\in \mathcal T} \varepsilon_{ij}(0)  \right\}^2
+
\left\{ n_\TG\sum_{ij\in \mathcal T} \varepsilon_{ij}(1)+
n_\CG \sum_{ij\in \mathcal T} \varepsilon_{ij}(0)  \right\}^2\\
&\ge&\frac{1}{2}\left\{ n_\CG\sum_{ij\in \mathcal T} \varepsilon_{ij}(1)+
n_\TG \sum_{ij\in \mathcal T} \varepsilon_{ij}(0)
+
 n_\TG\sum_{ij\in \mathcal T} \varepsilon_{ij}(1)+
n_\CG \sum_{ij\in \mathcal T} \varepsilon_{ij}(0)  \right\}^2\\
&=&\frac{N^2}{2}\left\{  \sum_{ij\in \mathcal T} \varepsilon_{ij}(1)+
 \sum_{ij\in \mathcal T} \varepsilon_{ij}(0)
 \right\}^2.
\end{eqnarray*}
Moreover, by $e=1/2$, the probabilities of receiving treatment assignments $(Z_i)_{1\le i \le M}$ and $(1-Z_i)_{1\le i \le M}$ are the same. Therefore, \eqref{last3} holds.
\end{proof}

\section{General weighted estimand and regression-based estimators}
\label{sec::weightedestimands-reg}

\subsection{Proof of Proposition \ref{t14}}

Proposition \ref{t14}  helps to simplify the proofs of Theorems \ref{thm::wls-without-covariates} and \ref{thm::wls-with-covariates}. We prove it now although it appears later than Theorems \ref{thm::wls-without-covariates} and \ref{thm::wls-with-covariates} in the  main paper.  We only prove the result for $\hat\tau_{\textsc{a}\omega}^\textup{adj}$ because the proof for $\hat\tau_{\textsc{a}\omega}$ is similar, and the proofs for $\hat\tau_{\omega\textsc{a}}$ and $\hat\tau_{\omega\textsc{a}}^{\textup{adj}}$ follow by definition. Also, we assume $c_i=x_{ij}=\bar{x}_{i\cdot}$ in the following proof, so we can avoid the notation $\bar c_\omega$ because
$
\bar c_\omega=\sum_{i=1}^M \omega_i c_i=  N^{-1} \sum_{i=1}^M n_i \bar{x}_{i\cdot} =  N^{-1} \sum_{ij}x_{ij} =0.
$

We can verify that the coefficients and heteroskedasticity-robust standard errors are identical from WLS and OLS with both outcomes and regressors multiply by $\omega_i^{1/2}$. So our proofs are based on the latter although the theorems in the main paper are stated in terms of the former.

\begin{lemma}\label{l18}
If $x_{ij}=\bar{x}_{i\cdot}$, then $\hat Q_{\textsc{a}\omega}(1) = \hat Q_\textsc{i}(1)$ based on the treated data, where $\hat Q_{\textsc{a}\omega}(1) $ equals the
 the coefficient of $ \omega_i^{1/2}\bar{x}_{i\cdot}$ in the OLS fit of $\omega_i^{1/2}\bar{Y}_{i\cdot}$ on $(\omega_i^{1/2},\omega_i^{1/2}\bar{x}_{i\cdot})$, and $\hat Q_\textsc{i}(1)$ equals the coefficient of $x_{ij} $ in the OLS fit of $Y_{ij}$ on $(1,x_{ij})$. The same conclusion holds for the control group.
\end{lemma}

\begin{proof}[Proof of Lemma \ref{l18}]
We use the Frisch--Waugh--Lovell Theorem \citep{angrist2008mostly} for the regression using the treated data. The coefficient of $\omega_i^{1/2}$ in the OLS fit of $\omega_i^{1/2}\bar{Y}_{i\cdot}$ on $\omega_i^{1/2}$ is
$
{\sum_{i=1}^M Z_i\bar{Y}_{i\cdot}\omega_i^{1/2}\omega_i^{1/2}}\big/{\sum_{i=1}^M Z_i \omega_i^{1/2}\omega_i^{1/2}}=
n_\TG^{-1} {\sum_{ij\in \mathcal T}Y_{ij}}=\bar Y_\TG
$
with residual
$\check Y_{i}=\omega_i^{1/2}(\bar{Y}_{i\cdot}-\bar{Y}_\TG)$. Similarly, the residual from the OLS fit of $\omega_i^{1/2}\bar x_{i\cdot}$ on $\omega_i^{1/2}$ is
$\check x_{i}=\omega_i^{1/2}(\bar{x}_{i\cdot} -\bar{x}_\TG )$. So $ \hat Q_{\textsc{a}\omega}(1)   =(\sum_{i=1}^M  Z_i\check x_i \check x_i^{\T})^{-1}\sum_{i=1}^M  Z_i \check x_i \check Y_i$. Its numerator reduces to
\begin{eqnarray*}
\sum_{i=1}^M Z_i \check x_{i} \check Y_{i}
&=&\sum_{i=1}^M \omega_iZ_i(\bar{x}_{i\cdot}-\bar{x}_\TG) (\bar{Y}_{i\cdot}-\bar{Y}_\TG)
=N^{-1} \sum_{ij\in \mathcal T}
(\bar{x}_{i\cdot}-\bar{x}_\TG) (Y_{ij}-\bar{Y}_\TG).
\end{eqnarray*}
Its denominator has a similar formula. Therefore,
\begin{eqnarray*}
\hat Q_{\textsc{a}\omega}(1)
&=&\left\{\sum_{i=1}^M \omega_iZ_i(\bar{x}_{i\cdot}-\bar{x}_\TG)( \bar{x}_{i\cdot}-\bar{x}_\TG)^{\T} \right\}^{-1}\sum_{i=1}^M \omega_iZ_i(\bar{x}_{i\cdot}-\bar{x}_\TG) (\bar{Y}_{i\cdot}-\bar{Y}_\TG)\\
&=&\left\{ \sum_{ij\in \mathcal T} (  \bar{x}_{i\cdot}-\bar{x}_\TG)
( \bar{x}_{i\cdot}-\bar{x}_\TG)^{\T}  \right\}^{-1} \sum_{ij\in \mathcal T} ( \bar{x}_{i\cdot}-\bar{x}_\TG)
(Y_{ij}-\bar{Y}_\TG) =  \hat Q_\textsc{i}(1).
\end{eqnarray*}
\end{proof}

Because of the equivalence in Lemma \ref{l18}, we use $\hat Q (1)$ in the following proof for both $\hat Q_{\textsc{a}\omega}(1)$ and $ \hat Q_\textsc{i}(1)$
when $x_{ij}=\bar{x}_{i\cdot}$.

\begin{proof}[Proof of Proposition \ref{t14}]
By the property of the OLS using the treated data, the coefficient of $\omega_i^{1/2}$ in the OLS fit of $\omega_i^{1/2}\bar{Y}_{i\cdot}$ on $(\omega_i^{1/2},\omega_i^{1/2}\bar{x}_{i\cdot})$ is the same as the coefficient of $\omega_i^{1/2}$ in the OLS fit of $\omega_i^{1/2}\{
\bar{Y}_{i\cdot}-\bar{x}_{i\cdot}^{\T}  \hat Q (1) \} $ on $\omega_i^{1/2}$:
$$
 {
\sum_{i=1}^M Z_i \{  \bar{Y}_{i\cdot}-\bar{x}_{i\cdot}^{\T}  \hat Q (1)   \} \omega_i^{1/2}\omega_i^{1/2}
}\big/
{\sum_{i=1}^M Z_i\omega_i^{1/2} \omega_i^{1/2}}=
n_\TG^{-1} {\sum_{ij\in \mathcal T} \{  Y_{ij}-\bar{x}_{i\cdot}^{\T}  \hat Q (1)\} } .
$$
We have a similar result for the control group.
Because $\hat\tau_{\textsc{a}\omega}^\textup{adj}$ is the difference between the coefficients of $\omega_i^{1/2}$ in the OLS fit of $ \omega_i^{1/2}\bar{Y}_{i\cdot}$ on $(\omega_i^{1/2},\omega_i^{1/2}\bar{x}_{i\cdot})$ in treatment and control groups, we have $\hat\tau_{\textsc{a}\omega}^\textup{adj}=\hat\tau_{\textsc{i}}^\textup{adj}(\bar{x}_{i\cdot})$ based on the formula of $\hat\tau_{\textsc{i}}^\textup{adj}$ in the proof of Theorem \ref{thm2}.

The residual from the OLS fit of $\omega_i^{1/2}\bar{Y}_{i\cdot}$ on $( \omega_i^{1/2},\omega_i^{1/2}\bar{x}_{i\cdot})$ in treatment group is
$$
\omega_i^{1/2}\left\{\bar{Y}_{i\cdot}-\bar Y_\TG-(\bar{x}_{i\cdot}-\bar x_\TG)^{\T}  \hat Q (1) \right\}=
\omega_i^{1/2}{\sum_{j=1}^{n_{i}} \left\{Y_{ij}-\bar Y_\TG-(\bar{x}_{i\cdot}-\bar x_\TG)^{\T} \hat Q (1)  \right\}}/n_i=
\omega_i^{1/2}{\sum_{j=1}^{n_{i}}\hat{r}_{ij}}/n_i,
$$
where $\hat r_{ij}$ is defined in Section  \ref{sec::2.2}.
We have a similar result for the control group. Therefore, the residual from the OLS fit of $\omega_i^{1/2}\bar{Y}_{i\cdot}$ on $(\omega_i^{1/2},\omega_i^{1/2}Z_i,\omega_i^{1/2}\bar{x}_{i\cdot},\omega_i^{1/2}Z_i \bar{x}_{i\cdot} )$ is $\omega_i^{1/2}{\sum_{j=1}^{n_{i}}\hat{r}_{ij}}/n_i$.

Let $\tilde X_{\textsc{a}\omega}$ be an $M\times (2+2p_x)$ matrix with row $i$ equaling $\tilde x_i^\T
=[\omega_i^{1/2}Z_i,
\omega_i^{1/2}(1-Z_i),\omega_i^{1/2}\bar{x}_{i\cdot}^{\T}Z_i,
\omega_i^{1/2}\bar{x}_{i\cdot}^{\T}(1-Z_i)]$. Recall the definitions of $G, H_i$ and $H$ in Section  \ref{sec::2.2}. We can verify that
$
\tilde X_{\textsc{a}\omega}^{\T} \tilde X_{\textsc{a}\omega}=G,
$
and
\begin{eqnarray*}
& &\tilde x_{i} \left(\omega_i^{1/2}\sum_{j=1}^{n_i}\hat {r}_{ij}/n_i\right)^2\tilde x_{i}^\T  \\
&=&\left( \begin{array} { c c } \omega_i^{1/2}Z_i\\
\omega_i^{1/2}(1-Z_i)\\ \omega_i^{1/2} \bar{x}_{i\cdot}  Z_i\\ \omega_i^{1/2} \bar{x}_{i\cdot} (1-Z_i)  \end{array} \right)\left(
\omega_i^{1/2} \sum_{j=1}^{n_i}\hat {r}_{ij}/n_i\right)^2
\left(\omega_i^{1/2}Z_i,
\omega_i^{1/2}(1-Z_i),\omega_i^{1/2}\bar{x}_{i\cdot}^{\T}Z_i,
\omega_i^{1/2}\bar{x}_{i\cdot}^{\T}(1-Z_i)\right)
= H_i/N^2 .
\end{eqnarray*}
So Lemma \ref{LZ} implies
\begin{eqnarray*}
 M\times \hat{\se}_\textsc{hw}^2(\hat\tau_{\textsc{a}\omega}^\textup{adj})
&=&
M\left[\left(\tilde X_{\textsc{a}\omega}^{\T} \tilde X_{\textsc{a}\omega}\right)^{-1} \left\{\sum_{i=1}^M
\tilde x_{i}\left(\omega_i^{1/2}\sum_{j=1}^{n_i}\hat {r}_{ij}/n_i\right)^2\tilde x_{i}^{\T}\right\}
\left(\tilde X_{\textsc{a}\omega}^{\T} \tilde X_{\textsc{a}\omega}\right)^{-1}\right]_{(1,1)+(2,2)-2(1,2)}\\
&=&\left[MG^{-1}
\sum_{i=1}^M H_i
G^{-1}/N^2\right]_{(1,1)+(2,2)-2(1,2)} \\
&=&  [G^{-1}HG^{-1}]_{(1,1)+(2,2)-2(1,2)}=M\times \hat{\se}_\textsc{lz}^2\{\hat\tau_{\textsc{i}}^\textup{adj}(\bar{x}_{i\cdot})\}.
\end{eqnarray*}
\end{proof}

\subsection{Proofs of Theorems \ref{thm::wls-without-covariates} and \ref{thm::wls-with-covariates}}
Proposition \ref{t14} shows that when the cluster size is $n_i$, OLS estimators based on individual-level data are equivalent to WLS estimators based on cluster averages with $\omega_i$ as the weight. We can use this result reversely to prove Theorems \ref{thm::wls-without-covariates} and \ref{thm::wls-with-covariates}. If we want to analyze WLS estimators based on cluster averages with $\pi_i$ as the weight, we can build a finite population with cluster sizes $m_i$ so that $\pi_i=m_i/N_\pi$ with $N_\pi=\sum_{i=1}^M m_{i}$, and establish the properties of $\hat\tau_{\textsc{a}\pi}$ and $\hat\tau_{\textsc{a}\pi}^\textup{adj}$ by applying Theorems \ref{thm::tau-i}, \ref{thm2} and Proposition \ref{t14} on this population. We can find those $m_i$'s when the $\pi_i$'s are rational numbers, that is, they equal ratios of positive integers. Although the $m_i$'s do not exist when the $\pi_i$'s are general real numbers, we can use rational numbers to approximate real numbers arbitrarily well. So the idea of the proof also applies to general real $\pi_i$'s.
This is a pure technical issue, but we include the proof for completeness.
We only prove Theorem \ref{thm::wls-with-covariates} for $\hat\tau_{\textsc{a}\pi}^\textup{adj}$ below because the proof of Theorem \ref{thm::wls-without-covariates} for $\hat\tau_{\textsc{a}\pi}$ is similar.

\begin{proof}[Proof of Theorem \ref{thm::wls-with-covariates}]
We first prove the case with rational $\pi_{i}$'s.
Let $\pi_i=m_i/ N_\pi $. Construct a finite population with $M$ clusters and $N_\pi=\sum_{i=1}^M m_i$ units, where cluster $i$ has $m_i$ units. Unit $(i,j)$ has potential outcomes $\bar{Y}_{i\cdot}(z)$ and unit-level covariates $x_{ij}=\bar{x}_{i\cdot}=c_i-\bar c_\pi$ ($i=1,\ldots,M; \ j=1,\ldots,m_i$).
Proposition \ref{t14} ensures that $\hat\tau_{\textsc{a}\pi}^\textup{adj} = \hat\tau_{\textsc{i}}^\textup{adj}(c_i)$. Therefore, we only need to check the regularity conditions in Theorem \ref{thm2} to prove Theorem \ref{thm::wls-with-covariates}.

The regularity conditions in Theorem \ref{thm::wls-with-covariates} and Assumptions \ref{assume::wls-cluster-average} and \ref{a66} imply  the conditions for the consistency and asymptotic Normality of $\hat\tau_{\textsc{i}}^\textup{adj}(c_i)$ in Theorem \ref{thm2}: Assumption \ref{assume::4} holds automatically; Assumption \ref{assume::1} holds because
$$
{  N_\pi^{-1} \sum_{i=1}^M m_i \bar{Y}_{i\cdot}(z)^4}=\sum_{i=1}^M
\pi_i \bar{Y}_{i\cdot}(z)^4=O(1) ,\quad { N_\pi^{-1} \sum_{i=1}^M m_i ||\bar{x}_{i\cdot}||_{\infty}^4}=\sum_{i=1}^M \pi_{i} ||c_i-\bar c_\pi||_{\infty}^4=O(1);
$$
Assumption \ref{a4} holds because
${N_\pi^{-1}  \sum_{i=1}^M\sum_{j=1}^{m_i} (c_i-\bar c_\pi)\bar{Y}_{i\cdot}(z)}=\sum_{i=1}^M \pi_i
(c_i-\bar c_\pi)\bar{Y}_{i\cdot}(z)$ converges to a finite vector, and ${N_\pi^{-1}  \sum_{i=1}^M\sum_{j=1}^{m_i} (c_i-\bar c_\pi)(c_i-\bar c_\pi)^\T}=\sum_{i=1}^M \pi_i
(c_i-\bar c_\pi)(c_i-\bar c_\pi)^\T$ converges to a finite and invertible matrix.
We can also verify the conditions on the maximum of the proportions of clusters by noticing that $\max_{1\le i \le M}m_i /{ N_\pi }=\max_{1\le i \le M}\pi_i$.

The asymptotic variance and estimated variance are $ V \{r_i(z)\}$ and ${ V }_\textup{c}\{r_i(z)\}$, respectively, where
$$
r_i(z)=\sum_{j=1}^{m_i} \{ \bar{Y}_{i\cdot}(z)-\bar Y_\pi(z)- (c_i-\bar c_\pi)^\T Q_{\textsc{a}\pi}(z)\}M/ N_\pi
=M\pi_i\{ \bar{Y}_{i\cdot}(z)-\bar Y_\pi(z)- (c_i-\bar c_\pi)^\T Q_{\textsc{a}\pi}(z) \}.
$$

We then prove the case with general real $\pi_i$'s.
Assume there is a series of $\epsilon_M$ which converges to 0 when $M\rightarrow +\infty$.
Because the population and the number of treatment assignments are finite, the formulas and probability limits of OLS coefficients, variance and estimated variance are continuous functions of $\pi_i$, we can find rational numbers $\ddot\pi_i$ close to $\pi_i$ so that for all treatment assignments, $\sum_{i=1}^M \ddot\pi_i \bar{Y}_{i\cdot}(z)^4=O(1),\  \sum_{i=1}^M \ddot\pi_i ||c_i||_{\infty}^4=O(1),\ M^{2/3}
\max_{1\le i \le M}|\ddot \pi_i-\pi_i|=o(1)$, $  \sum _ { i = 1 } ^ { M } \ddot\pi_{i}(c_i-\bar c_{\ddot\pi}) \bar{Y}_{i\cdot}(z)$ converges to a finite vector, and $  \sum _ { i = 1 } ^ { M }\ddot\pi_{i} (c_i-\bar c_{\ddot\pi}) (c_i-\bar c_{\ddot\pi})^{\T}$ converges to a finite and invertible matrix, $M^{1/2}|\hat\tau_{\textsc{a}\ddot\pi}^\textup{adj}-\hat\tau_{\textsc{a}\pi}^\textup{adj}|<\epsilon_M$, $M^{1/2}|\tau_{\ddot\pi}-\tau_{\pi}|<\epsilon_M$, $M^{1/2}|  \se(\hat\tau_{\textsc{a}\ddot\pi}^\textup{adj})- \se(\hat\tau_{\textsc{a}\pi}^\textup{adj})|<\epsilon_M$, $M|\hat \se_{\textsc{hw}}^2(\hat\tau_{\textsc{a}\ddot\pi}^\textup{adj})-\hat \se_{\textsc{hw}}^2(\hat\tau_{\textsc{a}\pi}^\textup{adj})|<\epsilon_M$, $
  | V_\textup{c}\{\ddot r_i(z) \}-V_\textup{c}\{r_i(z)\}|<\epsilon_M
$, where those notations with double dots are corresponding values derived from the OLS fit of $\ddot\pi_i^{1/2}\bar{Y}_{i\cdot}$ on $\{  \ddot\pi_i^{1/2}
,\ddot\pi_i^{1/2}Z_i,\ddot\pi_i^{1/2}(c_i-\bar c_{\ddot\pi})
,\ddot\pi_i^{1/2}Z_i(c_i-\bar c_{\ddot\pi})\} $.
By the results for rational $\ddot\pi_{i}$'s and Slutsky's Theorem, Theorem \ref{thm::wls-with-covariates} holds.
\end{proof}

\subsection{Proofs of Theorems \ref{thm8} and \ref{thm9}}
Finally, we prove Theorems \ref{thm8} and \ref{thm9}. Because $\hat\tau_{\pi\textsc{a}}$ and $\hat\tau_{\pi\textsc{a}}^\textup{adj}$ treat $M\pi_i\bar{Y}_{i\cdot}$ as a new outcome, we only need to check the regularity conditions on $M\pi_i\bar{Y}_{i\cdot}(z)$ given the results in completely randomized experiments.
\begin{proof}[Proof of Theorem \ref{thm8}]
The assumption $\sum_{i=1}^M \pi_i^2\bar{Y}_{i\cdot}(z)^2=o(1)$ implies
$
M^{-1} \sum_{i=1}^M \left\{ M\pi_i\bar{Y}_{i\cdot}(z) \right\}^2
= M\sum_{i=1}^M \pi_i^2\bar{Y}_{i\cdot}(z)^2 =o(M).
$
The assumption $\sum_{i=1}^M M^2\pi_i^4\bar{Y}_{i\cdot}(z)^4=o(1)$ implies
$
M^{-1}  \sum_{i=1}^M \left\{M\pi_i\bar{Y}_{i\cdot}(z)\right\}^4
= M^3\sum_{i=1}^M \pi_i^4\bar{Y}_{i\cdot}(z)^4 =o(M).
$
Lemma \ref{DIMA} guarantees the asymptotic properties of $\hat\tau_{\pi\textsc{a}}$.
\end{proof}
\begin{proof}[Proof of Theorem \ref{thm9}]
The assumption $\sum_{i=1}^M \pi_i^2\bar{Y}_{i\cdot}(z)^2=o(1)$ implies
$
M^{-1}   \sum_{i=1}^M \left\{ M\pi_i\bar{Y}_{i\cdot}(z) \right\}^2
= M\sum_{i=1}^M \pi_i^2\bar{Y}_{i\cdot}(z)^2 =o(M).
$
The assumption $\sum_{i=1}^M M^3\pi_i^4\bar{Y}_{i\cdot}(z)^4=O(1)$ implies
$
M^{-1}   \sum_{i=1}^M \left\{M\pi_i\bar{Y}_{i\cdot}(z)\right\}^4
= M^3\sum_{i=1}^M \pi_i^4\bar{Y}_{i\cdot}(z)^4 =O(1).
$
Lemma \ref{l14} guarantees the asymptotic properties of $\hat\tau_{\pi\textsc{a}}^\textup{adj}$.
\end{proof}

\section{More detailed results}
\label{sec::details-simulation}

Tables \ref{tab:simulation-studies} and \ref{tab:simulation-studies2} give more details for Section \ref{sec:simulation-section}. Table \ref{table::application} gives more details for Section \ref{sec::application}. 
 
 \begin{table}[ht]
 \centering
      \caption{Simulation: Part I} \label{tab:simulation-studies}
          \begin{subtable}[h]{\textwidth}
        \centering
\caption{Simulation in Section \ref{sec::compare-est-se} }\label{tb::sec1}
\begin{tabular}{lcccccccc}
\hline
& $\hat{\tau}_\textsc{i}$ & $\hat{\tau}_{\textsc{i}}^\textup{adj}$ & $\hat{\tau}_{\textsc{i}}^\textup{adj}(\bar x_{i\cdot})$&
$\hat{\tau}_{\textsc{i}}^\textup{ancova}$\\
 \hline
bias ($\times$100)                    &$0.21 $ & $0.02 $ & $-0.11$ & $0.18 $ \\
se ($\times$100)                      &$7.89 $ & $6.39 $ & $6.32 $ & $6.79 $ \\
$\hat{\se}_\textsc{ols}$ ($\times$100)&$4.66 $ & $4.48 $ & $4.56 $ & $4.55 $ \\
$\hat{\se}_\textsc{hw}$ ($\times$100) &$4.86 $ & $4.48 $ & $4.75 $ & $4.60 $ \\
$\hat{\se}_\textsc{lz}$ ($\times$100) &$8.75 $ & $7.59 $ & $7.39 $ & $7.80 $ \\
rmse ($\times$100)                    &$7.88 $ & $6.38 $ & $6.32 $ & $6.79 $ \\
coverage rate${}_\textsc{ols}$        &$0.749$ & $0.837$ & $0.841$ & $0.815$ \\
coverage rate${}_\textsc{hw}$         &$0.774$ & $0.835$ & $0.855$ & $0.817$ \\
coverage rate${}_\textsc{lz}$         &$0.968$ & $0.978$ & $0.966$ & $0.968$ \\
 \hline
&
$\hat{\tau}_{\textsc{t}}$&$\hat{\tau}_{\textsc{t}}^\textup{adj}(n_i)$&$\hat{\tau}_{\textsc{t}}^\textup{adj}(n_i,\tilde{x}_{i\cdot})$
&$\hat\tau_{\textsc{a}} $
&$\hat\tau_{\textsc{a}}^\textup{adj}(\bar x_{i\cdot})$& $\hat{\tau}_{\textsc{a}\omega}$
& $\hat{\tau}_{\textsc{a}\omega}^\textup{adj}(\bar x_{i\cdot})$\\
\hline
bias ($\times$100)                    &$0.46 $ & $0.18 $ & $-0.11$ & $-8.74$ & $-8.96$ & $0.21 $ & $-0.11$ \\
se ($\times$100)                      &$12.80$ & $6.42 $ & $3.48 $ & $8.14 $ & $6.58 $ & $7.89 $ & $6.32 $ \\
$\hat{\se}_\textsc{ols}$ ($\times$100)&$10.70$ & $6.87 $ & $4.85 $ & $7.95 $ & $6.39 $ & $7.72 $ & $6.23 $ \\
$\hat{\se}_\textsc{hw}$ ($\times$100) &$13.87$ & $6.91 $ & $4.31 $ & $9.07 $ & $7.68 $ & $8.75 $ & $7.39 $ \\
rmse ($\times$100)                    &$12.80$ & $6.42 $ & $3.48 $ & $11.94$ & $11.12$ & $7.88 $ & $6.32 $ \\
coverage rate${}_\textsc{ols}$        &$0.892$ & $0.963$ & $0.994$ & $0.811$ & $0.711$ & $0.935$ & $0.936$ \\
coverage rate${}_\textsc{hw}$         &$0.957$ & $0.961$ & $0.987$ & $0.869$ & $0.824$ & $0.968$ & $0.966$ \\
 \hline
\end{tabular}
    \end{subtable}

\bigskip
    \begin{subtable}[h]{\textwidth}
        \centering
\caption{Simulation in Section \ref{sec::noise_covariate}}\label{table::noise_covariate}
\begin{tabular}{lcccccccccc}
\hline
& $\hat{\tau}_\i$ &$\hat{\tau}_{\i}^{\a}$&$\hat{\tau}_{\i}^{\a}(\bar x_{i\cdot})$
&$\hat{\tau}_{\i}^{\textup{ancova}}$
&$\hat{\tau}_{\t}$&$\hat{\tau}_{\t}^\a(n_i)$&
$\hat{\tau}_{\t}^\a(n_i,\tilde{x}_{i\cdot})$&$\hat\tau_{\textsc{a}}$
&$\hat\tau_{\textsc{a}}^{\a}(\bar{x}_{i\cdot})$\\
 \hline
bias ($\times$100)  &        $-0.11$ & $-0.11$ & $-0.15$ & $-0.11$ & $0.15 $ & $-0.17$ & $-0.21$ & $-10.17$ & $-10.04$ \\
se ($\times$100)&            $6.09 $ & $6.09 $ & $6.05 $ & $6.08 $ & $11.58$ & $3.24 $ & $3.28 $ & $6.32  $ & $6.28  $ \\
$\hat{\se}$ ($\times$100) &  $7.68 $ & $7.67 $ & $7.61 $ & $7.67 $ & $13.48$ & $4.63 $ & $4.62 $ & $7.97  $ & $7.89  $ \\
rmse ($\times$100) &         $6.09 $ & $6.08 $ & $6.04 $ & $6.08 $ & $11.58$ & $3.24 $ & $3.28 $ & $11.97 $ & $11.84 $ \\
coverage rate&               $0.982$ & $0.982$ & $0.979$ & $0.982$ & $0.973$ & $0.993$ & $0.99 $ & $0.808 $ & $0.813 $ \\
\hline
\end{tabular}
    \end{subtable}
\end{table}

 \begin{table}[ht]
 \centering
      \caption{Simulation: Part II} \label{tab:simulation-studies2}
      \addtocounter{subtable}{+2} 
    \begin{subtable}[h]{\textwidth}
        \centering
\caption{Simulation in Section  \ref{sec::6.2}}\label{table::S2}
\begin{tabular}{lcccccccc}
\hline
& $\hat{\tau}_\textsc{i}$ &$\hat{\tau}_{\textsc{i}}^{\textup{adj}}$
&$\hat{\tau}_{\textsc{i}}^{\textup{ancova}}$
&$\hat{\tau}_{\textsc{t}}$&$\hat{\tau}_{\textsc{t}}^\textup{adj}(n_i)$&
$\hat{\tau}_{\textsc{t}}^\textup{adj}(n_i,\tilde{x}_{i\cdot})$&$\hat\tau_{\textsc{a}}$
&$\hat\tau_{\textsc{a}}^{\textup{adj}}(\bar{x}_{i\cdot})$\\
 \hline
bias ($\times$10)  &         $0.06 $   &$0.06 $&$0.05 $ &$0.01 $ &$-0.05 $&$0.01 $&$7.51 $ &$7.44 $\\
se ($\times$10)&             $1.56 $   &$1.69 $&$1.68 $ &$1.55 $ &$0.65  $&$0.36 $&$2.50 $ &$1.82 $\\
$\hat{\se}$ ($\times$10) &  $2.18 $   &$2.31 $&$2.25 $ &$2.17 $ &$0.87  $&$0.50 $&$3.49 $ &$2.52 $\\
rmse ($\times$10) &          $1.56 $   &$1.69 $&$1.68 $ &$1.55 $ &$0.65  $&$0.36 $&$7.92 $ &$7.66 $\\
coverage rate&               $0.989$   &$0.991$&$0.988$ &$0.990$ &$0.992 $&$0.991$&$0.400$ &$0.099$\\
\hline
\end{tabular}
    \end{subtable}

\bigskip
    \begin{subtable}[h]{\textwidth}
        \centering
\caption{Simulation in Section  \ref{sec::Extremeclusters1} }
\begin{tabular}{lccccccccc}
\hline
& $\hat{\tau}_\textsc{i}$ &$\hat{\tau}_{\textsc{i}}^{\textup{adj}}$
&$\hat{\tau}_{\textsc{i}}^{\textup{ancova}}$
&$\hat{\tau}_{\textsc{t}}$&$\hat{\tau}_{\textsc{t}}^\textup{adj}(n_i)$&
$\hat{\tau}_{\textsc{t}}^\textup{adj}(n_i,\tilde{x}_{i\cdot})$&$\hat{\tau}_{\textsc{t}}^\textup{adj}(\bar{x}_{i\cdot})$
&$\hat\tau_{\textsc{a}}$
&$\hat\tau_{\textsc{a}}^{\textup{adj}}(\bar{x}_{i\cdot})$\\
 \hline
bias ($\times$10)  &         $0.04  $&$ 0.00  $&$ -0.03 $&$ 0.01  $&$ -0.31 $&$ -0.31 $&$ -0.05 $&$ 6.53  $&$ 6.46  $\\
se ($\times$10)&             $1.93  $&$ 1.88  $&$ 1.88  $&$ 1.92  $&$ 1.90  $&$ 1.61  $&$ 1.69  $&$ 2.50  $&$ 1.82  $\\
$\hat{\se}$ ($\times$10) &  $2.66  $&$ 2.55  $&$ 2.59  $&$ 2.65  $&$ 1.91  $&$ 1.39  $&$ 2.30  $&$ 3.49  $&$ 2.52  $\\
rmse ($\times$10) &          $1.93  $&$ 1.87  $&$ 1.88  $&$ 1.92  $&$ 1.92  $&$ 1.64  $&$ 1.69  $&$ 6.99  $&$ 6.71  $\\
coverage rate&               $0.971 $&$ 0.985 $&$ 0.983 $&$ 0.977 $&$ 0.668 $&$ 0.505 $&$ 0.986 $&$ 0.549 $&$ 0.223 $\\
\hline
\end{tabular}
     \end{subtable}

\bigskip
    \begin{subtable}[h]{\textwidth}
        \centering
\caption{Simulation in Section  \ref{sec::extremeclusters2} }\label{tb::1dominatingcluster}
\begin{tabular}{lcccccccc}
\hline
& $\hat{\tau}_\textsc{i}$ &$\hat{\tau}_{\textsc{i}}^{\textup{adj}}$&
$\hat{\tau}_{\textsc{i}}^{\textup{ancova}}$&
$\hat{\tau}_{\textsc{t}}$&$\hat{\tau}_{\textsc{t}}^\textup{adj}(n_i)$&
$\hat{\tau}_{\textsc{t}}^\textup{adj}(n_i,\tilde{x}_{i\cdot})$&$\hat\tau_{\textsc{a}}$
&$\hat\tau_{\textsc{a}}^{\textup{adj}}(\bar{x}_{i\cdot})$\\
 \hline
bias ($\times$10)  &         $-0.30$& $-0.26$ &$-0.06$ &$-0.01$ & $-0.59$ & $-0.40$ & $0.01 $ & $-0.01$ \\
se ($\times$10)&             $3.41$ & $1.64 $ &$1.28 $ &$7.13 $ & $2.97 $ & $2.03 $ & $0.70 $ & $0.60 $ \\
$\hat{\se}$ ($\times$10) &  $1.96$ & $0.70 $ &$0.69 $ &$6.93 $ & $0.92 $ & $0.90 $ & $0.70 $ & $0.60 $ \\
rmse ($\times$10) &          $3.42$ & $1.66 $ &$1.28 $ &$7.13 $ & $3.03 $ & $2.07 $ & $0.70 $ & $0.60 $ \\
coverage rate&               $0.720$& $0.383$ &$0.628$ &$1.000$ & $0.186$ & $0.459$ & $0.948$ & $0.951$ \\
\hline
\end{tabular}
  \end{subtable}
\end{table}

   \begin{table}\centering
   \caption{Data analysis and additional simulation in Section  \ref{sec::application}}\label{table::application}
       \begin{subtable}[h]{\textwidth}
       \centering
\caption{Point estimate, standard error and confidence interval from \citet{guiteras2015encouraging}}\label{table::version2}
\begin{tabular}{lcccccccccc}
\hline
& $\hat{\tau}_\i$ &$\hat{\tau}_{\i}^{\a}$&$\hat{\tau}_{\i}^{\a}(\bar x_{i\cdot})$
&$\hat{\tau}_{\i}^{\textup{ancova}}$
&$\hat{\tau}_{\t}$&$\hat{\tau}_{\t}^\a(n_i)$&
$\hat{\tau}_{\t}^\a(n_i,\tilde{x}_{i\cdot})$&$\hat\tau_{\textsc{a}}$
&$\hat\tau_{\textsc{a}}^{\a}(\bar{x}_{i\cdot})$\\
 \hline
est  &        $0.178 $ & $0.159 $ & $0.157$ & $0.160 $ & $0.231$ & $0.180$ & $0.152$ & $0.182$ & $0.177$ \\
se  &            $0.044$ & $0.025$ & $0.023 $ & $0.025 $ & $0.085 $ & $0.044 $ & $0.021$ & $0.038 $ & $0.031 $ \\
lower CI&  $0.092 $ & $0.110 $ & $0.111  $ & $0.110$ & $0.064 $ & $0.094$ & $0.111$ & $0.108 $ & $0.117 $ \\
upper CI &         $0.263 $ & $0.208$ & $0.203 $ & $0.210 $ & $0.399 $ & $0.265 $ & $0.192 $ & $0.256 $ & $0.237 $ \\
\hline
\end{tabular}
\end{subtable}

\bigskip
    \begin{subtable}[h]{\textwidth}
           \centering
\caption{Simulation in Section  \ref{sec::simulation-real-data}}\label{tb::simulation-real-data}
\begin{tabular}{lcccccccccc}
\hline
& $\hat{\tau}_\i$ &$\hat{\tau}_{\i}^{\a}$&$\hat{\tau}_{\i}^{\a}(\bar x_{i\cdot})$
&$\hat{\tau}_{\i}^{\textup{ancova}}$
&$\hat{\tau}_{\t}$&$\hat{\tau}_{\t}^\a(n_i)$&
$\hat{\tau}_{\t}^\a(n_i,\tilde{x}_{i\cdot})$&$\hat\tau_{\textsc{a}}$
&$\hat\tau_{\textsc{a}}^{\a}(\bar{x}_{i\cdot})$\\
 \hline
bias ($\times$100)  &        $0.05 $ & $0.00 $ & $-0.01$ & $0.00 $ & $-0.03$ & $0.21 $ & $-0.07$ & $-0.90$ & $-1.01$ \\
se ($\times$100)&            $3.04 $ & $1.26 $ & $1.26 $ & $1.26 $ & $8.61 $ & $3.20 $ & $1.41 $ & $2.96 $ & $1.83 $ \\
$\hat{\se}$ ($\times$100) &  $3.24 $ & $1.98 $ & $1.95 $ & $2.00 $ & $8.57 $ & $3.19 $ & $1.89 $ & $3.25 $ & $2.22 $ \\
rmse ($\times$100) &         $3.04 $ & $1.26 $ & $1.26 $ & $1.26 $ & $8.61 $ & $3.21 $ & $1.41 $ & $3.09 $ & $2.09 $ \\
coverage rate&               $0.951$ & $0.996$ & $0.994$ & $0.995$ & $0.950$ & $0.931$ & $0.986$ & $0.961$ & $0.952$ \\
\hline
\end{tabular}
\end{subtable}
\end{table}

\end{document}